\documentclass[a4paper,reqno]{article}

    \usepackage[english]{babel} 

    \usepackage{graphicx}
    \usepackage[]{color}  
    \usepackage[dvipsnames]{xcolor} 
    \usepackage{csquotes} 
    \usepackage{appendix} 
    \usepackage[a4paper]{geometry} 
    \usepackage{cancel} 

    \usepackage{amsmath,amssymb,amsfonts,amsthm} 
    \usepackage{mathtools,todonotes} 
    \usepackage{tikz-cd}  
    \usepackage[all,cmtip]{xy} 
    \usepackage[bbgreekl]{mathbbol} 
    \usepackage{array} 

    \usepackage{hyperref} 

\theoremstyle{definition} 
    \newtheorem{definition}{Definition}

\theoremstyle{plain} 
    \newtheorem{theorem}[definition]{Theorem}
    \newtheorem{proposition}[definition]{Proposition}
    \newtheorem{lemma}[definition]{Lemma}
    \newtheorem{corollary}[definition]{Corollary}
    
    \newtheorem{assumption}[definition]{Assumption}

\theoremstyle{remark} 
    \newtheorem{remark}[definition]{Remark}

\usepackage[bibstyle=alphabetic,citestyle=alphabetic,useprefix,giveninits=true, sorting=ynt, sortcites, minbibnames=99,maxbibnames=99,backend=biber]{biblatex}  
\renewbibmacro{in:}{} 
\bibliography{bibliography}  
\DeclareSourcemap{
  \maps[datatype=bibtex]{
    \map[overwrite]{ 
      \step[fieldsource=doi, final]
      \step[fieldset=url, null]
      \step[fieldset=eprint, null]
      }
     \map[overwrite]{ 
      \step[fieldsource=eprint, final]
      \step[fieldset=pages, null]
      \step[fieldset=eid, null]
      \step[fieldset=journal, null]
    }  
  }
}

    \newcommand{\Tr}[0]{\text{Tr}}
    
    \renewcommand{\Tilde}{\widetilde}

    \DeclareMathOperator{\Ima}{Im}

    \newcommand{\Ker}[1]{\mathrm{Ker}{(#1)}}

    \newcommand{\qsp}[2]{\,\ensuremath{\raise.5ex\hbox{$#1$}\big\slash\raise-.5ex\hbox{$#2$}}}

    \makeatletter
        \newcommand{\zzlabel}[1]{\ifmeasuring@\else\ltx@label{#1}\fi} 
    \makeatother

    \newcounter{terms}[equation] 
    \newcommand{\unl}[2]{\underline{#1}_{\refstepcounter{terms} \zzlabel{#2} \theterms}} 

    \newcommand{\reft}[2]{(\ref{#1}.\ref{#2})} 

\usepackage{authblk}
\title{Boundary structure of gauge and matter fields coupled to gravity\thanks{A. S. C. acknowledges partial support of SNF Grant No. 200020 192080 and of the Simons Collaboration
on Global Categorical Symmetries. G. C. acknowledges partial support of SNF Grant No P500PT 203085. This
research was (partly) supported by the NCCR SwissMAP, funded by the Swiss National Science Foundation.}}
\author[1]{Giovanni Canepa}
\author[2]{Alberto S. Cattaneo}
\author[2,3]{Filippo Fila-Robattino}

\renewcommand\footnotemark{}

\affil[1]{Section de Math\'ematiques, Universit\'e de Gen\`eve, \href{mailto:giovanni.canepa.math@gmail.com}{giovanni.canepa.math@gmail.com}}
\affil[2]{Institut f\"ur Mathematik, Universit\"at Z\"urich, \href{mailto:cattaneo@math.uzh.ch}{cattaneo@math.uzh.ch}}
\affil[3]{Scuola Internazionale Superiore di Studi Avanzati, Trieste, \href{mailto:ffilarob@sissa.it}{ffilarob@sissa.it}}
\date{}

\begin{document}
\maketitle

\begin{abstract}
    The boundary structure of $3+1$-dimensional gravity (in the Palatini--Cartan formalism) coupled to to gauge (Yang--Mills) and matter (scalar and spinorial) fields is described through the use of the Kijowski--Tulczijew construction. In particular, the reduced phase space is obtained as the reduction of a symplectic space by some first class constraints and a cohomological description (BFV) of it is presented.
\end{abstract}
\tableofcontents

\section{Introduction}

In this paper we will study the boundary structure of general relativity (in 3+1 dimensions in the Palatini--Cartan formalism) coupled to different types of fields, such as a scalar field, a Yang--Mills field, and a spinor field. Our goal is to describe the \emph{reduced phase space} of the aforementioned theories coupled to gravity in two ways: (i) through a symplectic space and constraints on it and (ii)  using a cohomological description, the BFV formalism.

The reduced phase space can be considered as the fundamental building block of the analysis of field theories on manifolds with boundary. If the boundary is a Cauchy surface, we can define it to be the space of possible initial conditions. Often in the literature, the reduced phase space is obtained through Dirac's algorithm, while in this paper we follow and expand the description given for the gravity field alone in \cite{CCS2020} using the Kijowski and Tulczijew (KT) construction \cite{KT1979}. This construction roughly goes as follows: a space of boundary fields together with a closed two-form and some constraint  functions are derived from the variation of the action and the Euler--Lagrange equations in the bulk. Then, if the two-form is degenerate (as will be more precisely explained in Section \ref{s:KT}) and its kernel is regular, we perform a quotient and obtain a symplectic space, which we call the \emph{geometric phase space}. On it we define the constraints of the theory deriving them in a suitable way from the Euler--Lagrange equations. This last step might present some technical difficulties as the constraints defined as the restriction to the boundary of the Euler--Lagrange equations might not be basic with respect to the reduction of the two form. This is precisely the case at hand where both gravity alone and each of the composite theories have such problem. We overcome it by fixing  convenient representatives of the equivalence classes of the quotient and express the constraints in terms of them.

One of the reason of the choice of the KT construction is that it is automatically compatible with the cohomological description of the reduced phase space given by the BFV formalism (after Batalin--Fradkin--Vilkovisky \cite{BV1,BV2,BV3}). Indeed, if the constraints form a first class system (meaning that the Poisson brackets between them are proportional to the constraints themselves), it is possible to describe the space of functions over the reduced phase space as the zeroth cohomology of  a cohomological (i.e., odd and squaring to zero) vector field on a graded manifold constructed out of the geometric phase space and the constraints.

The BFV formalism was born as the hamiltonian version of the BV formalism, which was developed to overcome the degeneracy problems that one encounters when defining the partition function of gauge theories. It is a generalization of the constructions of Faddeev and Popov and of the BRST procedure \cite{FP1967,BRS1976, T1975} to encompass more general type of symmetries. The BV and BFV formalisms are related and it is possible to construct BV-BFV theories in which additional conditions are added to guarantee compatibility between bulk and boundary data \cite{CMR2012}. A quantization scheme has also been developed for such theories \cite{CMR2,CMR2012}.

Furthermore, given a BV theory on the bulk, under some regularity assumptions, it is possible to induce a BFV theory on the boundary. Crucially, for both gravity in the coframe formalism and the composite theories object of this article, in dimension $N \geq 4$,\footnote{In dimension $N=3$ gravity is topological and it is possible to induce a BFV theory from the BV one \cite{CaSc2019}. However, this is no longer true if we add matter fields. We postpone the study of this particular case to future work and consider in this article only the case $N=4$.} these regularity conditions of the BV theory are not satisfied (in the standard formulation, see \cite{CS2017}) and we have hence to resort to the alternative method described above to obtain a BFV theory. It is worth noting that from a BFV theory is then possible to obtain a full BV-BFV theory on cylindrical manifolds through the AKSZ construction \cite{AKSZ}. Because of the mentioned quantization scheme, one of the key point of this article is that it constitutes the first step towards the quantization of gravity together with matter fields. 

The formulation of general relativity used in this article will be the Palatini--Cartan (PC) or coframe one, which is classically equivalent to the  standard Einstein--Hilbert theory formulated in terms of the metric.The PC theory has several advantages when considering manifolds with boundaries, since it is expressed in terms of forms and connections which have a better behaviour when restricted to submanifolds.
For the same reason, in the case of the scalar and Yang--Mills fields we will use the first order formulation of these theories.

The main condition that we assume in the derivation of the boundary structure is the non-degeneracy of the induced metric on the boundary. In other words, we require the boundary to be time-like or space-like but not light-like. This last case will be object of future studies.

The article is structured as follows. We introduce the relevant constructions, KT and BFV, in Sections \ref{s:KT} and \ref{s:BFV_theory} respectively. Then we give an overview of the Palatini--Cartan formalism of gravity and its reduced phase space in Section \ref{s:PC_intro}. In particular we recall the results of \cite{CCS2020} where this theory has been analyzed  with the two methods mentioned above. In Sections \ref{s:scalar}, \ref{s:YM} and \ref{s:spinor} we then consider the coupled theories of gravity with a scalar field, a Yang--Mills field and a spinor field respectively. For each theory we describe the bulk theory, apply the KT construction and present the reduced phase space in terms of a symplectic space and some constraints on them with the corresponding structure of the Poisson brackets. Then we give the fully detailed description of the corresponding BFV theories.

Some of the results in this paper first appeared in \cite{F2021}.

\begin{remark}

    In this article we focus on the case of field theories defined on a four-dimensional manifold, this being the most interesting physical case. Some of the technical lemmata however are formulated and proven for a generic $N$. The generalization to $N>4$ does not bring to a different structure of the boundary theories, as was shown in \cite{CCS2020} for gravity alone, but only little modifications have to be taken into account. In particular we expect Theorems \ref{thm:first-class-constraints_scalar}, \ref{thm:first-class-constraints_YM} to hold verbatim in the generic $N\geq 4$ case. The case $N=3$ is different, since it is possible to induce directly a BFV theory from a BV on the bulk for pure gravity \cite{CaSc2019}. However, adding a scalar field spoils this possibility, leading to a non regular kernel of the  preboundary BFV form. The same happens when coupling 3d gravity with a Yang--Mills field in the first order formalism. In these cases we can proceed as described in Sections \ref{s:scalar} and \ref{s:YM}, keeping in mind that for $N=3$ we do not have a kernel in the direction of $\omega$ and hence we do not have to fix the additional vector field $e_n$.
\end{remark}

\bigskip
\textbf{Acknowledgments.} The authors would like to thank Valentino Huang for the useful remarks and comments.

\section{Preliminaries}
In this section we describe some of the mathematical background required in the rest of the paper. In particular, Section \ref{s:KT} is devoted to the Kijowski--Tulczijew (KT) construction, Section \ref{s:BFV_theory} to the BFV formalism and Section \ref{s:PC_intro} to the Palatini--Cartan gravity theory.
\subsection{The KT construction and the reduced phase space}\label{s:KT}
We describe here the Kijowski--Tulczijew \cite{KT1979} construction that we will use in the main part of the paper to describe the reduced phase space of the field theories considered.

\begin{remark}
    In order to keep the description simple, we describe the construction without details which are collected in the footonotes.
\end{remark}

Let $M$ be an an N--dimensional manifold with boundary $\partial M=: \Sigma$ and let $F$ be a vector bundle on $M$. For a large variety of theories---and in particular the ones at hand---the space of fields $F_M$ is in general defined as the space of smooth local sections $\phi$ on $F$, i.e. $F_M:=\Gamma(M,F)$, which is in general an infinite--dimensional manifold (inheriting the structure of a Fréchet space) on which we assume that Cartan calculus is defined. 
A field theory on $M$ is then specified by an action functional $S_M$, obtained by integrating a Lagrangian density $L(\phi)$.\footnote{
To define precisely such objects, one first needs to define the local calculus on $M\times F_M$. 
Let us consider the infinite jet bundle $J^\infty F$. The 
smooth local sections of the infinite jet bundle $\Gamma(M,J^\infty F)$,  can also be obtained by the jet prolongation $j^\infty \colon \Gamma(M,F) \rightarrow  \Gamma(M,J^\infty F)$. We can define a map $e_\infty$ by precomposing $j^\infty$ with the evaluation map $\text{ev}\colon M\times F_M\rightarrow F:(x,\phi)\mapsto \phi(x)$, i.e.
\begin{align*}
    e_\infty \colon {M\times F_M} \xrightarrow{(\mathrm{id},j^\infty)} {M\times \Gamma(M,J^\infty F)} \xrightarrow{\mathrm{ev}} {J^\infty F}
\end{align*}
It is a well known fact \cite{anderson} that differential forms on $J^\infty F$ carry a double degree, defining a bicomplex with respect to a veritcal differential $d_V$ and a horizontal differential $d_H$, such that $d=d_V+d_H$ is the usual de Rham differential. In particular, this implies that $d_V^2=0$, $ d_H^2=0$ and $d_V d_H + d_H d_V =0$.
It is then possible to define local forms on $M\times F_M$ by pulling back forms on $J^\infty F$ along $e_\infty$. This produces a double complex of local forms defined by
    \begin{equation}\label{e:localform}
        \Omega_{\mathrm{loc}}^{(p,q)}(M\times F_M):=e_\infty^*\Omega^{(p,q)}(J^\infty F),
    \end{equation}
where $p$ is the vertical degree and $q$ the horizontal one. The differentials are defined by $d:=e_\infty^*d_H$ and $\delta:=e_\infty^*d_V$, representing respectively the de Rham differential on differential forms on $M$ and the ``variational differential'' on forms on $F$. In particular, $d$ measures variations of fields at the space--time level, while $\delta$ measures variations of the field configuration at a given space--time point. A Lagrangian $L$ is defined to be an $(N,0)$ local form which, when evaluated at a field configuration $\phi$, is called Lagrangian density $L(\phi)$.
}

 The integral over $M$  of the Lagrangian density defines the action functional
    \begin{equation}
        S_M:=\int_M L(\phi)=\int_ML(\phi,\partial\phi,\partial^2\phi,\cdots,\partial^k \phi)dx^1\wedge \cdots \wedge dx^N.
    \end{equation}
When we act with $\delta$ on the Lagrangian, we obtain the variational formula \cite{Zuckerman} 
    \begin{equation}\label{e:variation_Lagrangian}
        \delta L= E(L) - d\alpha,
    \end{equation}
where $E(L)$ contains the Euler--Lagrange equations, and $\alpha$ is defined up to $d$-exact terms.\footnote{$E(L)$ is a $(N,1)$ local form and has the further property that it only depends on the 0-jet part of the field $\phi$, and it is independent of variations of $L$ by $d$-exact terms, i.e. $E(L+dK)=E(L)$. Such forms are also known as local source forms. Furthermore $\alpha \in \Omega_{\mathrm{loc}}^{(N-1,1)}$.}

If we integrate \eqref{e:variation_Lagrangian} on $M$, due to Stokes' theorem, $d\alpha$ gives rise to a boundary term. It was first noted by Kijowski and Tulczijew \cite{KT1979} that this boundary term defines a one form on the space of boundary fields over  $\Sigma$ which is analogous to the Liouville form in symplectic geometry.

In particular, defining the space of preboundary fields $\tilde{F}_\Sigma$ as the space of germs of fields at $\Sigma\times \{0\}$ on $\Sigma\times[0,\epsilon]$, the variation of the action $S_M$ yelds
    \begin{equation}
    \delta S_M= E(L)_M - \tilde{\pi}_\Sigma^*(\tilde{\alpha}_\Sigma),
    \end{equation}
where $E(L)_M$ arises after the integration of $E(L)$ , $\tilde{\pi}_\Sigma \colon F_M \rightarrow \tilde{F}_\Sigma$ is the natural surjective submersion to the space of preboundary fields and $\tilde{\alpha}_\Sigma$ is a one form on $\tilde{F}_\Sigma$ found after integrating $\alpha$.\footnote{More precisely, $E(L)_M$ is a $(0,1)$ local form on $M$ and $\tilde{\alpha}_\Sigma$ is $(0,1)$ local form on $\Sigma$.} Now $\tilde{\varpi}_\Sigma:=\delta \tilde{\alpha}_\Sigma$ is by definition a $\delta$-closed local two-form and, assuming that its kernel $\Ker{\tilde{\varpi}}:=\{ X\in T\tilde{F}_\partial \hspace{1mm}\vert\hspace{1mm} \iota_X\tilde{\varpi}=0 \}$ defines a regular distribution, it is a presymplectic form on $\tilde{F}_\Sigma$. By Frobenius' theorem, ker$(\tilde{\varpi}_\Sigma)$ 
is an involutive distribution on the space of
preboundary fields, hence we are able to consider the symplectic reduction $F_\Sigma:=\tilde{F}_\Sigma /\sim$ defined as the leaf space of the foliation, which we assume to be smooth. $F_\Sigma$ is called the geometric phase space of the theory and it is by definition a symplectic manifold with symplectic form $\varpi_\Sigma$ induced by $\tilde{\varpi}_\Sigma$.

Considering the induced surjective submersion $\pi_\Sigma \colon F_M\rightarrow F_\Sigma$ and assuming that $\alpha_\Sigma$ on $F_\Sigma$ is well defined, we obtain
    \begin{equation*}
        \delta S = E(L)_M - \pi_\Sigma^*(\alpha_\Sigma).   
    \end{equation*}

We can now define $EL_M:=\{\phi \in F_M \hspace{1mm}\vert \hspace{1mm} E(L)(\phi)=0\}$ as the zero locus of the Euler--Lagrange equations, i.e. the space of physically relevant fields. When restricted to the boundary, the EL equations split into equations containing the derivatives of the fields in a transversal direction and the remaining equations. They are respectively called evolution equations and constraints. In order to consider the physical space of fields of the theory on the boundary, one needs to impose the constraints on the space of boundary fields. In principle this could be done on the space of preboundary fields, taking into account the fact that the kernel of the presymplectic form might be enlarged; however, it is better for our purposes to impose them on the geometric phase space. Since this last space is a quotient, before proceeding we have to make sure that the restriction of the constraints is basic with respect to the reduction of the kernel of the pre-symplectic form. As we will see, this is not always the case and we might have to reformulate the constraints in order to have a basic expression.

In more mathematical terms, following \cite{CMR2012b}, we define $L_\Sigma:=\pi_\Sigma(EL_M)$ as the projection to geometric space of the solutions to the EL equations. In  general $L_\Sigma$ is isotropic with respect to $\varpi_\Sigma$, and sometimes also coisotropic, hence Lagrangian. This is the case of good field theory. Hoewever, we are interested in the space $C_\Sigma$ of Cauchy data, i.e. the submanifold of the geometric phase space that can be completed to an element belonging to $L_{\Sigma\times[0,\epsilon]}$ for $\epsilon$ small enough (more appropriately, one should work on jets in $\epsilon$). The evolution equations will then contain derivatives along the direction of $[0,\epsilon]$, while the zero locus of the constraints defines $C_\Sigma$. Note that, if $L_\Sigma$ is Lagrangian for $\epsilon$ small, then $C_\Sigma$ is coisotropic.. In our example this fact will be clear, since the constraints are found to be first class, i.e. local functions on the geometric phase space which are in involution with respect to the canonical Poisson structure on $F_\Sigma$ induced by the symplectic one. Finally, the physical space of the theory on the boundary is the symplectic reduction $\underline{C}_\Sigma$ of $C_\Sigma$, called the reduced phase space. The result in principle might not be smooth. Hence to describe it we resort to  its cohomological resolution known as the BFV formalism.

 \subsection{Some notes about the BFV formalism}\label{s:BFV_theory}
Because of the technical difficulties and the smoothness issues of a direct description of the reduced phase space, the BFV formalism offers a useful alternative.

The starting point is the symplectic manifold $(F_{\Sigma}, \varpi_{\Sigma})$ (the geometric phase space) and the set of constraints $\psi_i$, i.e. the restrictions of the EL equations to the boundary which are not evolutionary equations. The fundamental assumption is that these constraints form a first class set, i.e. $\{\psi_i,\psi_j\}=f_{ij}^k\psi_k$ for some functions $f_{ij}^k$ on $F_{\Sigma}$.

Given this setting the BFV formalism describes the functions on the reduced phase space as the cohomology of a suitable operator on a graded manifold which is a given \emph{extension} of the geometric phase space. Let $\lambda_i\in W_i$ be some odd Lagrange multiplyiers of degree +1 such that we can express the constraints in the integral form $$\Psi_i = \int_\Sigma \lambda_i \psi_i. $$
We consider the space $\mathcal{F}_{BFV}=F_{\Sigma} \times \Pi_{i}T^*W_i$ and denote by $\lambda_i^\dag$ the coordinates on the fibers of $T^*W_i$. This space has a natural symplectic structure given by
\begin{align*}
    \varpi_{BFV}= \varpi_\Sigma+\int_\Sigma \left( \sum_i \delta \lambda_i \delta \lambda_i^\dag\right) .
\end{align*}

On this symplectic space we define the function 
\begin{align*}
    S_{BFV}= \int_{\Sigma} \left(\lambda_i \psi_i + f_{ij}^k\lambda^{\dag}_k\lambda_i\lambda_j +R\right)
\end{align*}
where $R$ is a term of higher order in the $\lambda^\dagger$'s chosen so that $\{S_{BFV},S_{BFV}\}=0$ (Classical Master Equation). The function $S_{BFV}$ is called BFV action and it has been proven that it is always possible to find $R$ such that the classical master equation is satisfied\cite{BV3, Stasheff1997, Schaetz:2008}. We call $Q_{BFV}$ its Hamiltonian vector field. The key result is then given by the fact that $Q$ acts as a differential on functions on the space of fields and its cohomology in degree zero is isomorphic to $C^\infty(\underline{C}_{\Sigma})$ as a Poisson algebra when $\underline{C}_{\Sigma}$ is smooth. Hence $(\mathcal{F}_{BFV}, Q_{BFV})$ is a cohomological resolution of  $C^\infty(\underline{C}_{\Sigma})$.

	\subsection{The Palatini--Cartan formalism}\label{s:PC_intro}
	In this article we consider the first-order formulation of gravity in which the classical fields are a coframe and a connection. This formulation is classically equivalent to the original one in terms of the metric. In this section we describe the setting of this theory, the classical action in the bulk and its reduced phase space through the KT construction as first described in \cite{CCS2020}.
	\subsubsection{Classical space of fields}
	
Let $M$ be an $N$-dimensional manifold and let $P$ be an $SO(N-1,1)$-principal bundle on it. We consider an $N$-dimensional vector space $(V, \eta)$ with a Minkowski product, on which we can let the Lie group $SO(N-1,1)$ act via the fundamental representation $\rho\colon SO(N-1,1) \rightarrow \mathrm{End}(V)$. Next we consider the adjoint vector bundle $\mathcal{V}:= P \times_{\rho}V$. Finally, we require that there is an isomorphism $e\colon TM \rightarrow \mathcal{V}$. The first field of the theory is then an explicit choice of isomorphism $e\colon TM\to\mathcal{V}$, a.k.a.\ a vielbein
(the Lorentzian metric in the classically equivalent Einstein--Hilbert formalism will be recovered by pull back: $g =\eta(e,e)$)\footnote{Note that we can pull back the fiber metric eta and this defines a Lorentzian metric on M, so the setting described
above assumes that M admits a Lorentzian structure.}.
	
	The other field that we consider is a connection on $P$.
	Let $\omega\in \Omega^1(P,\mathfrak{so}(N-1,1))$ be the associated connection 1-form. We want to consider the gauge field as a dynamical field of the theory. The following proposition gives a useful way to include it in this setting.
	\begin{proposition}
	    The space of principal connections on $P$ over $M$ is an affine space modeled on $\mathcal{A}(M)=\Omega^1(M,\wedge^2 \mathcal{V})$.
	\end{proposition}
	\begin{proof}
	It is well known that it is possible to identify the affine space of   principal connections as the space of one forms with values in the corresponding Lie algebra $\mathfrak{so}(N-1,1)$.   Furthermore, it is possible to identify $\mathfrak{so}(N-1,1)$ with $\wedge^2\mathcal{V}$ by means of $\eta$.
 \end{proof}

	We define the space of $(i,j)$-forms to be the differential $i$-forms with values in the $j$-th exterior power of $V$, namely
	\begin{align*}
	    \Omega^{(i,j)}(M):=\Omega^i(M,\wedge^j \mathcal{V}). 
	\end{align*}
	
	The space of fields of our theory is then defined to be
	\begin{align*}
	    \mathcal{F}_{PC}:=\Omega_{nd}^{(1,1)}\times \mathcal{A}(M),
	\end{align*} where $\Omega_{nd}^{(1,1)}$ is the space of vielbeins as nondegenerate one-forms with values in $V$.
	This formalism has the further advantage that all the fields are expressed as differential forms and hence can easily be restricted to a suitable submanifold of $M$ (e.g. its boundary, if it has one).
	
	\subsubsection{Classical action}

	We are looking for an action functional that gives the same Euler--Lagrange  locus modulo symmetries as Einstein--Hilbert theory. The Palatini--Cartan action is
	\begin{equation}
		S_{PC}:=\int_M \frac{1}{(N-2)!}e^{N-2} \wedge F_\omega + \frac{\Lambda}{N!} e^N,
	\end{equation}
	
	where $e^k:=\underbrace{e\wedge e \wedge \cdots \wedge e}_{\text{k times}}$	and $F_\omega:=d\omega+\frac{1}{2} [\omega,\omega]$ is the curvature associated to $\omega$ which we regard as a $(2,2)$ form.
	We can find equations of motion by varying the action
	
	\begin{align}
		\label{variation PC}
			\delta S_{PC}&= \int_M \frac{1}{(N-3)!}e^{N-3} \delta e F_\omega - \frac{1}{(N-2)!}e^{N-2}d_\omega(\delta \omega)+ \frac{\Lambda}{(N-1)!}e^{N-1}\delta e \\
			&=\int_{M}\bigg[\frac{1}{(N-3)!}e^{N-3} F_\omega + \frac{\Lambda}{(N-1)!}e^{N-1}\bigg]\delta e + \frac{1}{(N-2)!}d_\omega(e^{N-2})\delta\omega \nonumber\\ &\phantom{=\int_{M}}-\frac{1}{(N-2)!}d(e^{N-2}\delta\omega),\nonumber
	\end{align}
	where we used integration by parts and the fact that $\delta_\omega F_\omega=-d_\omega(\delta \omega)$.\footnote{$
				\delta\omega F_\omega=\delta_\omega(d\omega + \frac{1}{2}[\omega,\omega])=-d\delta\omega +\frac{1}{2}[\delta\omega,\omega]-\frac{1}{2}[\omega,\delta\omega]
				=-d(\delta\omega)-[\omega,\delta\omega]=-d_\omega(\delta\omega).$} The last term in \eqref{variation PC} will produce a boundary term if $\partial M \neq \emptyset$, due to Stokes theorem.
	
	Then we find equations of motion
	\begin{eqnarray}
		&&e^{N-3}d_\omega e  =  0; \label{torsion} \\
		&&\frac{1}{(N-3)!}e^{N-3} F_\omega + \frac{\Lambda}{(N-1)!}e^{N-1}  =  0. \label{einstein}
	\end{eqnarray}
	
	Equation \eqref{torsion} is equivalent to $d_\omega e=0$ because of the non-degeneracy condition (and because $e^{N-3}$ is injective in this case \cite{CCS2020}). Furthermore, it fixes $\omega$ to be torsionless, and since it is compatible with $\eta$, then $d_\omega e=0$ implies the metricity condition $d_{e^*(\omega)}g=0$, which is uniquely solved by the Levi-Civita metric connection.
	
	After imposing \eqref{torsion}, we find that $\eqref{einstein}$ is equivalent to Einstein's field equation, with the addition of a cosmological constant $\Lambda$. 
	
	\begin{remark}
		It is important to notice that, even if $e$ is an isomorphism, $e\wedge \cdot$ might not be, indeed $e^{N-3}\wedge F_\omega=0$ is not equivalent to the flatness condition $F_\omega=0$
	\end{remark}
	
	\begin{remark}
		There are two ways of showing that the PC and EH theories are equivalent. The first one is to rewrite equation \eqref{einstein} after imposing \eqref{torsion} and see that it actually yelds Einstein's field equation. The other way is to use \eqref{torsion} and rewrite the action $S_{PC}$ in terms of the metric tensor, to see that it is equivalent to the Einstein--Hilbert action. This is seen very easily by noticing that
			\begin{equation}
				\frac{e^N}{N!}=\sqrt{-\det(g)}d^Nx = \text{Vol}_g, \qquad \frac{e^{N-2}}{(N-2)!}F_\omega= R \text{Vol}_g ,
			\end{equation}
		where $R$ is the Ricci scalar.
	\end{remark}

\subsubsection{The reduced phase space of Palatini--Cartan gravity}	\label{s:rps_gravity}
We present here the results of \cite{CCS2020} concerning the structure of the reduced phase space of Palatini--Cartan gravity.  The results of this section have been obtained through the Kijowski--Tulcjiev (KT) construction (described in Section \ref{s:KT}) and are the background construction that we will  adapt when adding matter and gauge fields in the following sections. 

The starting point of the KT analysis is the boundary term that we get when varying the action \eqref{variation PC}:
\begin{align*}
    \Tilde\alpha = \frac{1}{(N-2)!}\int_{\Sigma}  e^{N-2}\delta\omega.
\end{align*}

\begin{assumption}\label{Assumption_metric_non_deg}
We further assume that the bulk vielbein satisfies the extra nondegeneracy condition that the induced boundary
metric $g^\partial$, defined by $g^\partial:=\iota_\Sigma^*e^*(\eta)$, is nondegenerate.\footnote{One might also consider the stronger condition that
the induced boundary metric is space-like, but this is not needed for the following considerations.} This is an open
condition on the space of bulk field that ensures that the constrained submanifold $C_\Sigma$ is coisotropic.
\end{assumption}

	The classical fields on the boundary will again be indicated by $(e,\omega)$. The inclusion $\iota:\Sigma\hookrightarrow M$ of $\Sigma$ in $M$ induces the bundles $P|_\Sigma:=\iota^*(P)$ and $\mathcal{V}|_\Sigma:=\iota^*(\mathcal{V})$. The fields are respectively defined as
	\begin{itemize}
		\item $e$ is a nondegenerate section of $T^*\Sigma\otimes\mathcal{V}\vert_\Sigma$, meaning that (i) at each point the three components are linearly independent and (ii) the underlying metric $g$, defined by $g:=e^*(\eta)$, is nondegenerate (because of Assumption \ref{Assumption_metric_non_deg};
		\item $\omega$ is an element of the space of connections $\mathcal{A}_\Sigma$, locally modeled by $\Gamma(T^*\Sigma\otimes\bigwedge^2\mathcal{V}\vert_\Sigma)$.
	\end{itemize}
	We denote the space of preboundary fields as $\tilde{F}_\partial=\Omega^{(1,1)}_{\partial,\text{n.d.}}\times\mathcal{A}_\Sigma$.

	We note that $\tilde{\alpha}$ is the integral of a local $($top$,1)$ form on $\tilde{F}_\partial \times\Sigma$ as defined in \eqref{e:localform} and therefore a 1-form on $\tilde{F}_\partial$. By taking its variation (the variational vertical differential), we obtain a two-form on $\tilde{F}_\partial$
	\begin{equation}
		\label{omega grav}
		\tilde{\varpi}:=\delta \alpha= \frac{1}{(N-3)!}\int_\Sigma e^{N-3}\delta e \delta \omega.
	\end{equation}
	By construction, $\tilde{\varpi}$ is closed on $\tilde{F}_\partial$ and satisfies the first requirement to be a symplectic form on $\tilde{F}_\partial$.  However, it is degenerate, namely ker$(\tilde{\varpi}):=\{ X\in T\tilde{F}_\partial \hspace{1mm}\vert\hspace{1mm} \iota_X\tilde{\varpi}=0 \}\neq\{0\}$. 
	In \cite{CCS2020} it was proven that ther kernel is regular. Hence, in order to get rid of this degeneracy, we can perform a symplectic reduction.\footnote{The vector fields in the kernel of the presymplectic form span a smooth involutive distribution. The quotient space $\tilde{F}_\partial/\text{ker}(\tilde{\varpi})$ is the set of leaves in the foliation induced by ker$(\tilde{\varpi})$. In our case, the vector fields in the kernel only act, at fixed $e$, as translations of the connection $\omega$, therefore it is easy to see that the quotient space is still a smooth manifold.} The quotient space $F_\partial$ will be called the geometric phase space of the theory 
	\begin{equation}\label{e:GPS_gravity}
		F_\partial:=\frac{\tilde{F}_\partial}{\text{ker}(\tilde{\varpi})},
	\end{equation}
	with the canonical projection $\pi_\partial\colon \tilde{F}_\partial \rightarrow F_\partial$. Hence the space of boundary fields is a bundle ${F}^{\partial}\rightarrow\Omega_{nd}^1(\Sigma, \mathcal{V})$ with local trivialization on an open $\mathcal{U}_{\Sigma} \subset \Omega_{nd}^1(\Sigma, \mathcal{V})$
\begin{equation*}
{F}^\partial \simeq \mathcal{U}_{\Sigma} \times \mathcal{A}^{red}(\Sigma),
\end{equation*}
	where $\mathcal{A}^{red}(\Sigma)$ is the space of equivalence classes of connections $\omega \in \mathcal{A}(\Sigma)$  under the equivalence relation $\omega  \sim \omega + v $ for every $v \in \Omega^{1,2}(\Sigma)$ such that $e^{N-3}v =0$.
The corresponding symplectic form is
\begin{align}\label{classical-boundary-symplform}
\varpi = \frac{1}{(N-3)!}\int_{\Sigma} e^{N-3} \delta e \delta [\omega].
\end{align}
In order to define the constraints on this quotient space, and to give an explicit description of the reduced phase space, it is better to fix a representative of the equivalence relation described above, since the restriction of the EL equations to the boundary are not invariant under the equivalence relation. A convenient choice is given by the following construction. We choose a section $e_n$ of $\mathcal{V}|_\Sigma$ and we restrict the space of fields by the conditions that $e_1,e_2,e_3,e_n$ form a basis, where $e_a := e(\partial_a)$.\footnote{there is actually no restriction in the space-like case; otherwise, one has to work on charts of the space of fields and pick an $e_n$ for each chart}
We denote by $F_{e_n}$ the space of preboundary fields $\tilde{F}_\partial$ together with $e_n \in \mathcal{V}$ completing the basis.
On this space we have the following theorem:

\begin{theorem}[\cite{CCS2020}]\label{thm:omegadecomposition}
Suppose that  $g^{\partial}$, the metric induced on the boundary, is nondegenerate. Given any $\widetilde{\omega} \in \Omega^{1,2}$, there is a unique decomposition 
\begin{equation} \label{omegadecomp}
\widetilde{\omega}= \omega +v
\end{equation}
with $\omega$ and $v$ satisfying 
\begin{equation}\label{omegareprfix2}
 e^{N-3}v=0 \quad \text{ and } \quad e_n  e^{N-4}d_{\omega} e \in \Ima W_1^{\partial,(1,1)}.
\end{equation}
\end{theorem}

Let we denote by $F'_{e_n}$ the subspace of $F_{e_n}$ of the fields satisfying \eqref{omegareprfix2}. 
\begin{corollary}[\cite{CCS2020}]
    $F'_{e_n}$ is symplectomorphic to $F_\partial$.
\end{corollary}

Hence from now on we will require \eqref{omegareprfix2} and work on $F'_{e_n}$. The space of coframes and connections satisfying this last equation is the geometric phase space of the PC gravity theory.
We can now analyse the restriction of the Euler--Lagrange equations on the boundary to see which further constraints they impose on the geometric phase space. 
In order to simplify the computation of their Hamiltonian vector fields, it is convenient to rewrite the constraints on
$F'_{e_n}$ as discussed in \cite{CCS2020}:
\begin{align*}
    L_c &= \int_{\Sigma} c  e^{N-3} d_{\omega} e,\\
    P_{\xi}&= \int_{\Sigma}  \iota_{\xi} e  e^{N-3} F_{\omega} + \iota_{\xi} (\omega-\omega_0)  e^{N-3} d_{\omega} e,\\
    H_{\lambda} &= \int_{\Sigma} \lambda e_n \left(\frac{1}{(N-3)!} e^{N-3}F_\omega + \frac{1}{(N-1)!}\Lambda e^{N-1}\right),
\end{align*}
where $\omega_0$ is a reference connection and $c \in\Omega^{0,2}_\partial$, $\xi \in\mathfrak{X}(\Sigma)$ and $\lambda\in \Omega^{0,0}_\partial$ are Lagrange multipliers.

From now on we are going to consider the fields $c,\xi$ and $\lambda$ to be odd fields (shifted by 1 in a suitable supermanifold). This will be useful later for the BFV formalism. For more details we refer to \cite{CCS2020}. 

The constraints above are of first class, hence defining a coisotropic submanifold of the geometric phase space. The structure is specified by the following

\begin{theorem}[\cite{CCS2020}] \label{thm:first-class-constraints}
 Under Assumption \ref{Assumption_metric_non_deg}, the functions $L_c$, $P_{\xi}$, $H_{\lambda}$ are well defined on ${F}^{\partial}_{PC}$ and define a coisotropic submanifold  with respect to the symplectic structure $\varpi_{PC}$. In particular they satisfy the following relations
\begin{subequations}\label{brackets-of-constraints}
\begin{eqnarray}
\{L_c, L_c\} = - \frac{1}{2} L_{[c,c]} & \{P_{\xi}, P_{\xi}\}  =  \frac{1}{2}P_{[\xi, \xi]}- \frac{1}{2}L_{\iota_{\xi}\iota_{\xi}F_{\omega_0}} \\
\{L_c, P_{\xi}\}  =  L_{\mathrm{L}_{\xi}^{\omega_0}c} & \{L_c,  H_{\lambda}\}  = - P_{X^{(a)}} + L_{X^{(a)}(\omega - \omega_0)_a} - H_{X^{(n)}} \\
\{H_{\lambda},H_{\lambda}\}  =0 & \{P_{\xi},H_{\lambda}\}  =  P_{Y^{(a)}} -L_{ Y^{(a)} (\omega - \omega_0)_a} +H_{ Y^{(n)}} 
\end{eqnarray}
\end{subequations}
where $X= [c, \lambda e_n ]$, $Y = \mathrm{L}_{\xi}^{\omega_0} (\lambda e_n)$ and $Z^{(a)}$, $Z^{(n)}$ are the components of $Z\in\{X,Y\}$ with respect to the frame $(e_a, e_n)$.\footnote{The notation $\mathrm{L}_{\xi}^{\omega}$ denotes the covariant Lie derivative along the odd vector field $\xi$ with respect to a connection $\omega$:
\begin{align*}
\mathrm{L}_{\xi}^{\omega} A = \iota_{\xi} d_{\omega} A -  d_{\omega} \iota_{\xi} A \qquad A \in \Omega^{i,j}_{\partial}.
\end{align*}}
\end{theorem}

The data of this theorem can be translated into the BFV formalism as explained in Section \ref{s:BFV_theory}. The result is the following theorem.

\begin{theorem}[\cite{CCS2020}]\label{thm:BFVgravity}
Under Assumption \ref{Assumption_metric_non_deg}, let $\mathcal{F}_{PC}$ be the bundle

\begin{equation}
\mathcal{F}_{PC} \longrightarrow \Omega_{nd}^1(\Sigma, \mathcal{V}),
\end{equation}
with local trivialisation on an open $\mathcal{U}_{\Sigma} \subset \Omega_{nd}^1(\Sigma, \mathcal{V})$
\begin{equation}\label{LoctrivF1}
\mathcal{F}_{PC}\simeq \mathcal{U}_{\Sigma} \times \mathcal{A}(\Sigma) \oplus T^* \left(\Omega_{\partial}^{0,2}[1]\oplus \mathfrak{X}[1](\Sigma) \oplus C^\infty[1](\Sigma)\right) =: \mathcal{U}_{\Sigma} \times \mathcal{T}_{PC} ,
\end{equation}
and fields denoted by $e \in \mathcal{U}_{\Sigma}$ and $\omega \in \mathcal{A}(\Sigma)$ in degree zero such that they satisfy the structural constraint $e_n  e^{N-4} d_{\omega} e \in \Ima W_1^{\partial,(1,1)}$, ghost fields $c \in\Omega_{\partial}^{0,2}[1]$, $\xi \in\mathfrak{X}[1](\Sigma)$ and $\lambda\in \Omega^{0,0}[1]$ in degree one, $c^\dag\in\Omega_{\partial}^{N-1,N-2}[-1]$, $\lambda^\dag\in\Omega_{\partial}^{N-1,N}[-1]$ and $\xi^\dag\in\Omega_\partial^{1,0}[-1]\otimes\Omega_{\partial}^{N-1,N}$ in degree minus one, together with a fixed  $e_n \in \Gamma(\mathcal{V})$, completing the image of elements $e \in\mathcal{U}_{\Sigma}$ to a basis of  $\mathcal{V}$;
define a symplectic form and an action functional on $\mathcal{F}$ respectively by
\begin{align}
\varpi_{PC} = \int_{\Sigma} & \frac{1}{(N-3)!}e^{N-3} \delta e \delta \omega + \delta c \delta c^\dag + \delta \lambda \delta \lambda^\dag + \iota_{\delta \xi} \delta \xi^\dag,\label{symplectic_form_NC1} \\
S_{PC}= \int_{\Sigma} & \frac{1}{(N-3)!}c  e^{N-3} d_{\omega} e + \frac{1}{(N-3)!}\iota_{\xi} e  e^{N-3} F_{\omega} +\frac{1}{(N-3)!} \iota_{\xi} (\omega-\omega_0)  e^{N-3} d_{\omega} e\nonumber\\ 
&+ \lambda e_n \left(\frac{1}{(N-3)!} e^{N-3}F_\omega + \frac{1}{(N-1)!}\Lambda e^{N-1}\right) +\frac{1}{2} [c,c] c^{\dag}\nonumber\\ 
& - \mathrm{L}_{\xi}^{\omega_0} c c^{\dag} +\frac{1}{2}\iota_{\xi}\iota_{\xi}F_{\omega_0}c^{\dag}  + [c, \lambda e_n ]^{(a)}(\xi_a^{\dag}- (\omega - \omega_0)_a c^\dag) + [c, \lambda e_n ]^{(n)}\lambda^\dag \nonumber\\
 &- \mathrm{L}_{\xi}^{\omega_0} (\lambda e_n)^{(a)}(\xi_a^{\dag}- (\omega - \omega_0)_a c^\dag) - \mathrm{L}_{\xi}^{\omega_0} (\lambda e_n)^{(n)}\lambda^\dag - \frac{1}{2}\iota_{[\xi,\xi]}\xi^{\dag}. \label{action_NC1}
\end{align}
Then the triple $(\mathcal{F}_{PC}, \varpi_{PC}, S_{PC})$ defines a BFV structure on $\Sigma$.\end{theorem}

\section{Real scalar field theory coupled to gravity}\label{s:scalar}

In this section we explore the boundary structure for the field theory generated by the coupling of gravity and a real scalar field theory. As we will see, the structure of the constraints of gravity is not directly affected by this coupling. Nonetheless, the kernel of the two-form induced from the bulk on the boundary changes in a non-trivial way, resulting in an additional structural constraint that fixes some components of the \emph{momentum} of the scalar field on the boundary.

\begin{remark}
    In this section we analyse only the case of a real scalar field. However the results presented here can be extended without big effort to the case of multiplets, or to the case of multiple scalar fields.
\end{remark}
	\subsection{Real scalar field in the first order formalism}

	We now consider a scalar field $\phi\in \mathcal{C}^{\infty}(M)$ as a smooth function on space--time.
	In order to couple the scalar field to gravity in the Palatini-Cartan formalism, it is useful to consider the first-order formulation introducing a new field $\Pi\in\Omega^{(0,1)}(M)$, i.e.\ a section of the ``Poincaré'' bundle $\mathcal{V}$.
	The idea behind the introduction of this new field is to avoid to consider the term
        \begin{equation}
			\frac{1}{2}g^{\mu\nu}\nabla_\mu\phi\nabla_\nu \phi,
		\end{equation} 
	which usually appears in the Klein--Gordon Lagrangian on an arbitrary background, because it involves the inverse $g^{\mu\nu}$ of the metric tensor, which is hard to deal with in calculations in terms of the vielbein.
	
	The new field $\Pi$ is a priori independent of $\phi$, but after the equations of motion are found, it will assume the role of the momentum associated to the scalar field.
	
	The minimal coupling (in the massless case) is described by the action
		\begin{equation}
		\begin{split}
		S&=S_{PC}+S_{scal} \hspace{10mm}\text{with} \\
		S_{PC}&=\int_M \frac{1}{(N-2)!}e^{N-2} \wedge F_\omega + \frac{\Lambda}{N!} e^N \\
		S_{scal}&=\int_M \frac{1}{(N-1)!}e^{(N-1)}\wedge \Pi \wedge d\phi +  \frac{1}{2N!}e^N (\Pi,\Pi),
		\end{split}
		\end{equation}
	where $(\cdot,\cdot)$ is a shorthand notation for the pairing $\eta$ in $\mathcal{V}$. In an orthonormal (with respect to the Minkowski metric) basis $\{v_a\}$ of $\mathcal{V}$, $\forall A=A^a v_a,B=B^bv_b\in \mathcal{V}$ it reads:	
		\begin{equation}
			(A,B):=A^aB^b\eta_{ab}.
		\end{equation}
	The variation of the action yields
		\begin{equation}
			\begin{split}
					\delta S=&\int_{M}\bigg[\frac{1}{(N-3)!}e^{N-3} F_\omega + \frac{\Lambda}{(N-1)!}e^{N-1}+\frac{1}{(N-2)!}e^{N-2}\Pi d\phi + \frac{1}{2(N-1)!}e^{N-1}(\Pi,\Pi)\bigg]\delta e +\\
					& \qquad + \frac{1}{(N-2)!}d_\omega(e^{N-2})\delta\omega + \frac{1}{(N-1)!}e^{N-1}d\phi\delta \Pi +  \frac{1}{N!}e^d(\Pi,\delta \Pi) +\\
					& \qquad \frac{1}{(N-1)!}d(e^{N-1}\Pi)\delta \phi+- d\bigg( \frac{1}{(N-2)!}e^{N-2}\delta \omega + \frac{1}{(N-1)!}e^{N-1}\Pi\delta\phi  \bigg).	
			\end{split}	
		\end{equation}
	We notice that the variation of the action produces a boundary term which, applying Stokes' theorem, is given by
	\begin{equation}\label{bound term scalar}
		\tilde{\alpha}:=\int_{\partial M}\frac{1}{(N-2)!}e^{N-2}\delta\omega + \frac{1}{(N-1)!}e^{N-1}\Pi\delta\phi
	\end{equation}
	This is the term corresponding to the local 1-form on the space of preboundary fields defined in Section \ref{s:KT}. Its variation will produce the pre-symplectic form which will be essential to construct the reduced phase space in the next section.
	
	From the variation of the action we also find the equations of motion, which are given by
		\begin{eqnarray}\label{eom scalar}
			&d_\omega e  =  0;  \\
			&\frac{1}{(N-3)!}e^{N-3} F_\omega + \frac{\Lambda}{(N-1)!}e^{N-1}  + \frac{1}{(N-2)!}e^{N-2}\Pi d\phi + \frac{1}{2(N-1)!}e^{N-1}(\Pi,\Pi)=0 ;\\
			& d(e^{N-1}\Pi)=0 \label{gauss_0}; \\
			&e^{N-1} (d\phi - (e,\Pi))=0 	\label{Pi eq_0},
		\end{eqnarray}
	where, to find the equation of motion corresponding to $\delta\Pi$, we used the following identity,\footnote{proved in Lemma \hyperref[(1,0)]{\ref*{useful identities}.(\ref*{(1,0)})} in Appendix \ref{a:tec_and_lproofs}} which holds for every $A,B\in\Omega^{(0,1)} $:
		\begin{equation}
			\frac{1}{N}e^N(A,B)= (-1)^{|A|+|B|}e^{N-1}(e,A)B.
		\end{equation}	
	We can further simplify equation \eqref{gauss_0}, in fact, using $d(e^{N-1}\Pi)=d_\omega(e^{N-1}\Pi)$ because top forms transform trivially under the action of the Lie algebra, then $d(e^{N-1}\Pi)=d_\omega(e^{N-1})\Pi+e^{N-1}d_\omega\Pi$, but $d_\omega e=0$, therefore we find that   \eqref{gauss_0} is equivalent to
		\begin{equation}
			e^{N-1}d_\omega\Pi=0 \label{gauss_1}.
		\end{equation} 
	Furthermore, we can also simplify \eqref{Pi eq_0}, since $W_{N-1}^{(1,0)}\colon\Omega^{(1,0)}(M)\rightarrow\Omega^{(N,N-1)}(M):A\mapsto e^{N-1}A$ is injective.\footnote{see Lemma \hyperref[lem: W_N-1 bijective]{\ref*{useful id bulk}.(\ref*{lem: W_N-1 bijective})} in Appendix \ref{a:tec_and_lproofs}} Therefore we obtain
		\begin{equation}
			d\phi + (e,\Pi)=0 \label{dynamics Pi}.
		\end{equation}
    This equation fixes $\Pi$ in terms of the derivatives of $\phi$, while \eqref{gauss_1} is then just the usual Klein--Gordon equation for a massless scalar field on an arbitrary background. To see this, we compute the scalar field part of the Lagrangian after having imposed the constraint and plug it into the action, showing that we recover the usual Klein--Gordon Lagrangian on a curved background.
	
	First of all,
	we have
		\begin{equation}
			\begin{split}
				\frac{e^N}{N!}&=\frac{1}{N!}\epsilon_{a_1\cdots a_N} e^{a_1}_{\mu_1}\cdots e^{a_N}_{\mu_N} dx^{\mu_1}\cdots dx^{\mu_N}=\epsilon_{a_1\cdots a_N}e^{a_1}_1\cdots e^{a_N}_N d^Nx\\
				&=\text{det}(e)d^Nx.
			\end{split}
		\end{equation}
	Then, since $\det(g)=-\det(e)^2$, we obtain $e^N/N!=\sqrt{-\det(g)}d^Nx=\text{Vol}_g$ as the canonical volume form. 
	In coordinates (with respect to the local basis $\{e_\mu\}$ of $V$), assuming that the metric is nondegenerate, eq. \eqref{dynamics Pi} reads 
		\begin{equation}
			\Pi^\mu=-g^{\mu\nu}\partial_\nu \phi.
		\end{equation}
	Finally we can compute the term in the scalar part of the action, using the previous identity
		\begin{equation}
			\begin{split}
				S_{\text{scal}}=\int_M \frac{1}{(N-1)!}e^{N-1}\Pi d\phi + \frac{1}{2N!}e^N (\Pi,\Pi)=-\int_M \frac{1}{2}(\text{Vol}_g) g^{\mu\nu}\partial_\mu \phi \partial_\nu\phi,
			\end{split}
		\end{equation}		
	which is exactly the Klein--Gordon Lagrangian, once we notice that $\nabla_\mu \phi=\partial_\mu \phi$.

	\subsection{Classical Boundary Structure in \texorpdfstring{$N=4$}{Lg}}
	We now assume our space--time manifold $M$ to be a  4-dimensional manifold with boundary $\Sigma:=\partial M$ and we study the boundary structure of the theory using the KT construction (see Section \ref{s:KT}). In particular we show that the constraints defining the reduced phase space are 
	first class, thus defining a coisotropic submanifold as their zero locus.  We also show that the scalar field coupling does not modify the boundary structure of pure gravity.

	\subsubsection{The Reduced Phase Space}
	Many of the results which we  present in this section are an extension of what have been shown in \cite{CCS2020} and recalled in Section \ref{s:rps_gravity}.
	We start by considering the boundary term \eqref{bound term scalar} that is found after the variation of the action; for $N=4$ it reads
	\begin{equation}
		\tilde{\alpha}:=\int_{\Sigma}\frac{1}{2}e^{2}\delta\omega + \frac{1}{3!}e^{3}\Pi\delta\phi.
	\end{equation}
	We again indicate the classical fields on the boundary by $(e,\omega,\phi,\Pi)$. The fields are  defined as in  Section \ref{s:rps_gravity} and additionally we have:
	\begin{itemize}
		\item $\phi\in\mathcal{C^{\infty}}(\Sigma)$ is a smooth function on $\Sigma$;
		\item $\Pi$ is an element of $\Omega^{(0,1)}_\partial:=\Omega^{(0,1)}(\Sigma)$, where we define $\Omega^{(i,j)}(\Sigma):=\Gamma(\bigwedge^iT^*\Sigma\otimes\bigwedge^j\mathcal{V}|_\Sigma)$.
	\end{itemize}
	Hence we denote the space of preboundary fields as $\tilde{F}_\partial=\Omega^{(1,1)}_{\partial,\text{n.d.}}\times\mathcal{A}_\Sigma\times \mathcal{C^{\infty}}(\Sigma)\times\Omega^{(0,1)}_\partial$.
	The next step is to take the variation of $\tilde{\alpha}$ and obtain a closed two-form on $\tilde{F}_\partial$:
	\begin{equation}
		\label{omega scal}
		\tilde{\varpi}:=\delta \alpha= \int_\Sigma e\delta e \delta \omega + \frac{1}{3!}\delta(e^3\Pi)\delta\phi.
	\end{equation}
	As before, this two-form is degenerate.
	Considering a generic vector field
	$X=\mathbb{X}_e\frac{\delta}{\delta e}+\mathbb{X}_\omega\frac{\delta}{\delta \omega}+\mathbb{X}_\phi\frac{\delta}{\delta \phi}+\mathbb{X}_\Pi\frac{\delta}{\delta \Pi}$,\footnote{the components of the vector fields are such that  $\mathbb{X}_e\in\Omega^{(1,1)}_{\partial\text{n.d}}$, $\mathbb{X}_\omega\in\mathcal{A}_\Sigma$, $\mathbb{X}_\phi\in\mathcal{C^{\infty}}(\Sigma)$ and $\mathbb{X}_\Pi\in\Omega^{(0,1)}_\partial$} we explicitly find the kernel of $\tilde{\varpi}$ as those vector fields satisfying $\iota_X\tilde{\varpi}=0$, which is equivalent to the following system of equations:
	\begin{eqnarray}
		&& e\mathbb{X}_e=0 \label{X_e ker};\\
		&& e \mathbb{X}_\omega +\frac{1}{2}e^2\Pi\mathbb{X}_\phi=0\label{X_omega ker} ;\\
		&& \frac{1}{2}e^2\Pi\mathbb{X}_e+\frac{1}{3!}e^3\mathbb{X}_\Pi=0\label{Pi ker};\\
		&& e^3\mathbb{X}_\phi=0 \label{phi ker}.
	\end{eqnarray}
	
	Defining $W_k^{\partial(i,j)}:=e^k\wedge:\Omega_\partial^{(i,j)}\rightarrow\Omega_\partial^{(i+k,j+k)} $, by Lemmas 	\hyperref[lem:Wep11]{\ref*{lem:We_boundary}.(\ref*{lem:Wep11})} and 
	\hyperref[lem: ker We0]{\ref*{lem:We_boundary}.(\ref*{lem: ker We0})} $W_1^{\partial(1,1)}$ and $W_3^{\partial(0,0)}$ are both injective, therefore \eqref{X_e ker} and \eqref{phi ker} are solved respectively by $\mathbb{X}_e=0$ and $\mathbb{X}_\phi=0$. 
	\eqref{X_omega ker} and \eqref{Pi ker} reduce to $e\mathbb{X}_\omega=0$ and $e^3\mathbb{X}_\Pi=0$. The geometric phase space is then found to be a bundle over $\Omega^{(1,1)}_{\partial,\text{n.d.}}$ with local trivialization on an open $\mathcal{U}_{\Sigma} \subset \Omega_{nd}^1(\Sigma, \mathcal{V})$  $$F_\partial\simeq\mathcal{U}_{\Sigma}\times\mathcal{A}^{red}(\Sigma)\times \mathcal{C^{\infty}}(\Sigma)\times(\Omega^{(0,1)}_\partial/\sim),$$ where
	\begin{align}
		\Pi\sim\tilde{\Pi} \hspace{5mm}&\Leftrightarrow \hspace{5mm}\Pi-\tilde{\Pi}=\gamma \hspace{2mm}\text{with}\hspace{2mm}e^3\gamma=0
	\end{align}
	and $\mathcal{A}^{red}(\Sigma)$ was defined in Section \ref{s:rps_gravity}.
	From now on, We denote  
	$\Omega^{\partial(0,1)}_{red}:=\Omega^{(0,1)}_\partial/\sim$.
	$F_\partial$ is thus a symplectic manifold with symplectic form
	\begin{equation}
		\varpi=\int_\Sigma e\delta e \delta[\omega]+\frac{1}{3!}\delta(e^3[\Pi])\delta\phi.
	\end{equation}
	\begin{remark}\label{p instead of Pi}
		Instead of $\Pi$, we might define a new boundary field $p:=\frac{1}{3!} e^3\Pi$. In this way the prefactor $e^3$ automatically selects the physical part in $\Pi$ without the need of a further symplectic reduction. Furthermore, we obtain a sympletic 2-form whose ``scalar field part'' is written in Darboux coordinates: $\varpi=\int_\Sigma e\delta e \delta[\omega]+\delta p\delta\phi$. 
	\end{remark}

	\begin{remark}
		As for the case without matter, notice that the constraints (as the restrictions of the EL equations from the bulk to the boundary) are not necessarily invariant under $v$-translations and $\gamma$-translations, therefore we fix a convenient set of representatives of the equivalence classes $[\omega]$ and $[\Pi]$. The next subsection deals with choosing such representatives in the ideal way. In order to do so, as in the pure gravity case described in Section \ref{s:rps_gravity}, we choose a section $e_n$ of $\mathcal{V}|_\Sigma$ and we restrict the space of fields by the conditions that $e_1,e_2,e_3,e_n$ form a basis.
	\end{remark}

	\subsubsection{Choice of Representatives via Constraints}\label{rep scalar}
	As mentioned, we need to fix convenient representatives of the classes $[\omega]\in \mathcal{A}^{red}_\Sigma$ and $[\Pi]\in\Omega^{\partial(0,1)}_{red}$. 
	The idea is to take advantage of the constraints to fix the representatives, in particular we will use parts of the dynamical constraints. The constraints to be imposed on the space of preboundary fields are
	\begin{eqnarray}
		&&d_\omega e  =  0  \label{const scal 1};\\
		&&e F_\omega + \frac{\Lambda}{3!}e^{3}  + \frac{1}{2}e^{2}\Pi d\phi + \frac{1}{2\cdot 3!}e^{3}(\Pi,\Pi)=0; \\
		&&d\phi + (e,\Pi)=0 \label{pi constraint}.
	\end{eqnarray}
	We do not impose $e^{d-1}d_\omega\Pi=0 $ because it is an evolution equation. Furthermore, it is a top form on $M$, therefore it cannot be restricted to $\Sigma$.
	The choice of representative of $[\omega]$ uses \eqref{const scal 1} and it follows verbatim the choice done in the gravity theory without additional matter fields. Hence, following the construction described in Section \ref{s:rps_gravity} we fix the representative of $[\omega]$ by choosing the connection $\omega$ satisfying 
	\begin{align*}
	    e_n d_\omega e	\in \Ima W_1^{\partial(1,1)}.
	\end{align*}
	Existence and uniqueness of such connection are proved in Theorem \ref{thm:omegadecomposition}.

	Let us now consider the equivalence class $[\Pi]$. 	We replicate the procedure used for the connection and use a constraint to fix the representative of it. The constraint will be based on \eqref{pi constraint}.
	In particular, we exploit the property of the following Lemma which will be proved in Appendix \ref{a:tec_and_lproofs}.
	\begin{lemma}\label{lem:technical_for_scalar}
	    Suppose that $g^\partial$ is nondegenerate, then the map $A_e\colon \Ker{W_3^{\partial(0,1)}} \rightarrow \Omega_{\partial}^{1,0}$, $A_e(p)= (e,p)$ is bijective.
	\end{lemma}
	
	\begin{remark}
	    In analogy to what happens in gravity alone, the non-degeneracy condition is here fundamental to use this constraint to fix the representative. If the boundary metric is degenerate, the structure of the theory might be different as was shown for gravity alone in \cite{CCT2020}.
	\end{remark}
	
	Using this lemma, the following theorem shows that \eqref{pi constraint} fixes uniquely the representative of the equivalence class in an appropriate way.
	\begin{theorem}\label{constr rep scal}
	 Let $g^\partial$ be nondegenerate. Given any $\widetilde{\Pi}\in \Omega_{\partial}^{0,1}$, there is a unique decomposition $\widetilde{\Pi}= \Pi+p$ such that $p \in \Ker{W_3^{\partial(0,1)}}$ and
	 \begin{align}\label{e:constraintscalar}
	     (e,\Pi)=-d\phi.
	 \end{align}
	\end{theorem}
	\begin{proof}
	    If $\widetilde{\Pi}$ satisfies \eqref{e:constraintscalar} there is nothing to prove. Suppose that $(e,\widetilde{\Pi})-d\phi=K$, then since $A_e$ is bijective, there exists a $p \in \Ker{W_3^{\partial(0,1)}}$ such that $K=(e,p)$. Then $\Pi= \widetilde{\Pi}-p$ satisfies \eqref{e:constraintscalar}. 
	    
	    For uniqueness, suppose that there are two such decompositions $\widetilde{\Pi}= \Pi_1+p_1=\Pi_2+p_2$. Then we would have $(e,\Pi_1)=(e, \Pi_2)$ and consequently $(e,p_1)=(e, p_2)$ with $p_1,p_2 \in \Ker{W_3^{\partial(0,1)}}$. Since $A_e$ is bijective, this implies $p_1=p_2$.
	\end{proof}
	
	Hence from now on we will work on the space of fields given by $e\in\Omega_{nd}^1(\Sigma, \mathcal{V})$, $\omega\in \mathcal{A}(\Sigma)$, $\phi \in \mathcal{C^{\infty}}(\Sigma)$, $\Pi \in\Omega^{(0,1)}_\partial$ such that $e_n d_\omega e	\in \Ima W_1^{\partial(1,1)}$ and $(e,\Pi)=-d\phi$, which is symplectomorphic to $F_\partial$.

	\subsubsection{Poisson Brackets of the Constraints}\label{scal brackets}
	
	We still have to impose the constraints on the space of pre-boundary fields. In order to do so, we recast them into local forms by means of Lagrangian multipliers:
	furthermore, if we split $\tilde{\mu}=\iota_\xi e +\lambda e_n$, from $J_{\tilde{\mu}}$ we obtain two functions:
	\begin{align*}
	    &L_c:=\int_\Sigma ced_\omega e.\\
		&P_\xi=\int_\Sigma\frac{1}{2}\iota_\xi(e^2)F_\omega+\frac{1}{3!}\iota_\xi(e^3\Pi)d\phi + \iota_\xi(\omega-\omega_0)ed_\omega e;\\
		&H_\lambda=\int_\Sigma \lambda e_n \bigg( e F_\omega + \frac{\Lambda}{3!}e^{3}  + \frac{1}{2}e^{2}\Pi d\phi + \frac{1}{2\cdot 3!}e^{3}(\Pi,\Pi) \bigg) .
	\end{align*} 
	\begin{remark}
		It is important to notice that the Lagrange multiplier have the role of the generators of the symmetry. In particular, $c\in\Omega^{(0,2)}_\partial$ generates the internal gauge symmetry,\footnote{Recall that we identify $\mathfrak{so}(3,1)\simeq \wedge^2 \mathcal{V}$} $\xi \in\mathfrak{X}(\Sigma)$ represents the vector field parametrizing the local diffeomorphisms in the direction tangential to the boundary, while $\lambda \in \mathcal{C^{\infty}}(\Sigma)$ is the generator of the local diffeomorphism normal to the boundary. 
	\end{remark}

	For future advantage we added a term in $P_\xi$ proportional to $ed_\omega e$, depending also on a reference connection. The addition of this term does not change the constrained set. It is also important to notice that the terms in $J_{\tilde{\mu}}$ containing $e^3$ disappear in $P_\xi$ because $\iota_\xi(e^4)=0$.  
	
	Furthermore, we assume the Lagrange multipliers to be odd, namely we consider $c \in\Omega^{0,2}_\partial[1]$, $\xi \in\mathfrak{X}[1](\Sigma)$ and $\lambda\in \Omega^{0,0}_\partial[1]$, and we denote with $\mathrm{L}_{\xi}^{\omega}$ the covariant Lie derivative along the odd vector field $\xi$ with respect to a connection $\omega$:
	\begin{align}\label{Lie derivative}
		\mathrm{L}_{\xi}^{\omega} A = \iota_{\xi} d_{\omega} A -  d_{\omega} \iota_{\xi} A \qquad A \in \Omega^{i,j}_{\partial}.
	\end{align}
	
	\begin{theorem} \label{thm:first-class-constraints_scalar}
		With the usual hypothesis that $g^\partial$ is nondegenerate, the functions $L_c$, $P_{\xi}$, $H_{\lambda}$ define a coisotropic submanifold  with respect to the symplectic structure $\varpi_{PC}$. Their Poisson brackets read
		\begin{subequations}\label{brackets-of-constraints_scalar}
			\begin{eqnarray}
				\{L_c, L_c\} = - \frac{1}{2} L_{[c,c]} & \{P_{\xi}, P_{\xi}\}  =  \frac{1}{2}P_{[\xi, \xi]}- \frac{1}{2}L_{\iota_{\xi}\iota_{\xi}F_{\omega_0}} \\
				\{L_c, P_{\xi}\}  =  L_{\mathrm{L}_{\xi}^{\omega_0}c} & \{L_c,  H_{\lambda}\}  = - P_{X^{(a)}} + L_{X^{(a)}(\omega - \omega_0)_a} - H_{X^{(n)}} \\
				\{H_{\lambda},H_{\lambda}\}  =0 & \{P_{\xi},H_{\lambda}\}  =  P_{Y^{(a)}} -L_{ Y^{(a)} (\omega - \omega_0)_a} +H_{ Y^{(n)}} ,
			\end{eqnarray}
		\end{subequations}
		where $X= [c, \lambda e_n ]$, $Y = \mathrm{L}_{\xi}^{\omega_0} (\lambda e_n)$ and $Z^{(a)}$, $Z^{(n)}$ are the components of $Z\in\{X,Y\}$ with respect to the frame $(e_a, e_n)$.	
	\end{theorem}
	\begin{remark}
		As said before, this theorem has the same structure as in \cite{CCS2020}, where the Palatini-Cartan theory without the scalar coupling is analyzed.
	\end{remark}
	
	\begin{proof}
		Theorem \ref{thm:omegadecomposition} allows to have well defined constraint, because of the uniqueness of the representative $\omega$ of $[\omega]$.
		
		In order to compute the brackets of the constraints, we first compute the Hamiltonian vector fields associated to the constraints, defined for a function $f$ on the space of boundary fields as $\mathbb{X}_f$ such that $\iota_{\mathbb{X}_f}\varpi=\delta f$.
		
		Before explicitly computing the vector fields, we recall Remark \ref{p instead of Pi} and notice that $P_\xi$ can also be written as 	
		\begin{equation}
			P_\xi=\int_\Sigma\frac{1}{2}\iota_\xi(e^2)F_\omega+ \iota_\xi(\omega-\omega_0)ed_\omega e + \iota_{\xi}(p)d\phi.
		\end{equation}
		Then the variations of the constraints are
		\begin{align*}
			\delta L_c= \int_{\Sigma}  -\frac{1}{2} c [\delta \omega, ee]  + \frac{1}{2} c d_{\omega}\delta (ee) = \int_{\Sigma} [c, e]  e \delta \omega  +  d_{\omega} c  e\delta e;
		\end{align*}
		\begin{align*}
			\delta P_{\xi} &= \int_{\Sigma}  \iota_{\xi} (e \delta e) F_{\omega}  -\frac{1}{2}\iota_{\xi} (e  e) d_{\omega}\delta \omega + \iota_{\xi} \delta \omega e d_{\omega} e  -\frac{1}{2} \iota_{\xi} (\omega-\omega_0) [\delta \omega, ee]\\
			& \qquad  + \frac{1}{2} \iota_{\xi} (\omega-\omega_0)  d_{\omega}\delta (ee) + \iota_\xi(\delta p)d\phi + \iota_\xi p d(\delta\phi)\\
			& \overset{\diamondsuit}{=} \int_{\Sigma} - e \delta e \iota_{\xi}F_{\omega}  + \frac{1}{2}d_{\omega} \iota_{\xi} (e  e) \delta \omega - \frac{1}{2} \delta \omega \iota_{\xi} d_{\omega}(e e)  + \frac{1}{2}\delta \omega [ \iota_{\xi} (\omega-\omega_0), ee] \\
			& \qquad  + \frac{1}{2} d_{\omega} \iota_{\xi} (\omega-\omega_0) \delta (ee) + \delta p \iota_\xi(d\phi) + d(\iota_\xi p)\delta \phi\\
			& = \int_{\Sigma}  - e \delta e \iota_{\xi}F_{\omega}   -  	(\mathrm{L}_{\xi}^{\omega} e ) e \delta \omega + e \delta \omega [ \iota_{\xi} (\omega-\omega_0), e]  +  d_{\omega} \iota_{\xi} (\omega-\omega_0) e \delta e \\
			& \qquad - \xi(\phi)\delta p - L_\xi^{\omega_0}(p)\delta\phi\\ 
			& = \int_{\Sigma}  - e \delta e (\mathrm{L}_{\xi}^{\omega_0} (\omega-\omega_0) + \iota_ {\xi}F_{\omega_0})  -  (\mathrm{L}_{\xi}^{\omega_0} e ) e \delta \omega - \xi(\phi)\delta p - L_\xi^{\omega_0}(p)\delta\phi.
		\end{align*}
		In the last computation the symbol ($\diamondsuit$) indicates that we used integration by parts.

		\begin{align*}
			\delta H_{\lambda} &= \int_{\Sigma} \lambda e_n  \delta e  F_{\omega}+\frac{1}{2}\Lambda  \lambda e_n e^2 \delta e -\lambda e_n e  d_{\omega} \delta \omega + \delta\bigg[\frac{1}{2\cdot 3!}\lambda e_n e^3 (\Pi,\Pi) + \frac{1}{2}\lambda e_n e^2 \Pi d\phi \bigg]\\
			&=\int_{\Sigma} \lambda e_n \bigg[ \left( F_\omega + \frac{\Lambda}{2}e^2 + \frac{e^2}{4}(\Pi,\Pi) + e\Pi d\phi\right)\delta e + \frac{e^3}{3!}(\Pi,\delta \Pi) + \frac{e^2}{2}d\phi \delta \Pi + \frac{e^2}{2}\Pi d\delta \phi\bigg]\\
			&\qquad + d_\omega(\lambda e_n e)\delta \omega\\
			&\overset{\bigstar \diamondsuit}{=}\int_{\Sigma}  \lambda e_n \bigg[ \left(  F_\omega + \frac{\Lambda}{2}e^2 + \frac{e^2}{4}(\Pi,\Pi) + e\Pi d\phi\right)\delta e + \frac{1}{2}d_\omega(\lambda e_n e^2 \Pi)\delta \phi \\
			&\qquad +\frac{\lambda e_n}{2} e^2 \left[ d\phi + (e,\Pi) \right]\delta \Pi - \lambda \frac{e^3}{3!}(e_n,\Pi)\delta \Pi  + d_\omega(\lambda e_n e)\delta \omega\\
			&=\int_{\Sigma}\left[ \lambda e_n \left( F_\omega + \frac{\Lambda}{2}e^2 +\frac{e^2}{4}(\Pi,\Pi) + e\Pi d\phi\right) - \frac{\lambda}{2}e^2(e_n,\Pi)\Pi \right]\delta e \\
			&\qquad  -\lambda(e_n,\Pi)\delta p + \frac{1}{2}d_\omega(\lambda e_n e^2 \Pi)\delta \phi +  d_\omega(\lambda e_n e)\delta \omega,
		\end{align*}
	
 	where we used $\forall A,B \in \Omega^{(0,1)}_\partial$ the following identity\footnote{a proof can be found in Lemma \ref{useful id boundary}}:
		\begin{equation}\tag{$\bigstar$}
			e_n\frac{e^{N-1}}{(N-1)!}(A,B)=(-1)^{|A|+|B|}\left[ \frac{1}{(N-2)!}e_n e^{N-2}(e,A)B + \frac{e^{N-1}}{(N-1)!}(e_n,A)B \right] \qquad  
		\end{equation}

		The components of the Hamiltonian vector fields of $L_c$ and $P_\xi$ are
		\begin{align*}
				&\mathbb{L}_e = [c,e] &\mathbb{L}_p=0 \\  &\mathbb{L}_\omega = d_{\omega} c + \mathbb{V}_L\label{e:ham_vf_J} &\mathbb{L}_\phi=0\\
				&\mathbb{P}_e = - \mathrm{L}_{\xi}^{\omega_0} e &\mathbb{P}_p=-\mathrm{L}_\xi^{\omega_0}(p)   \\  &\mathbb{P}_\omega = - \mathrm{L}_{\xi}^{\omega_0} (\omega-\omega_0) - \iota_ {\xi}F_{\omega_0} + \mathbb{V}_P   &\mathbb{P}_\phi=-\xi(\phi),
		\end{align*}
		where, e.g., $\mathbb{L}_e \equiv \mathbb{L}(e)$, with $\iota_{\mathbb{L}}\varpi_{PC}= \delta L_c$, and $\mathbb{V}_L, \mathbb{V}_P\in \mathrm{ker}(W_1^{\partial,(1,2)})$.

		The components of the Hamiltonian vector field of $H_\lambda$ are described by
		\begin{equation}\label{e:ham_vf_L}
			\begin{split}
				&\mathbb{H}_e= d_{\omega}(\lambda e_n) + \lambda \sigma \\ 
				&e \mathbb{H}_\omega  =  \lambda e_n \left( F_{\omega}+\frac{1}{2}\Lambda  e^2 + e\Pi d\phi + \frac{1}{4}e^2(\Pi,\Pi) \right) - \frac{\lambda }{2}e^2\Pi(\Pi,e_n) \\ 
				&\mathbb{H}_p=d_\omega\left( \frac{\lambda e_n}{2}e^2 \Pi \right) \\
				&\mathbb{H}_\phi = -\lambda (\Pi,e_n).
			\end{split}
		\end{equation}
		
		As one may see, we did not fully compute $\mathbb{H}_\omega$ from the variation of $H_\lambda$, but we do not need an explicit expression for it, since in the computations we will only need $e\mathbb{H}_\omega$. A similar argument holds for $\mathbb{L}_\omega$ and $\mathbb{P}_{\omega}$, which are defined up to an element in $\mathrm{ker}(W_1^{\partial,(1,2)})$ that will be irrelevant in the following arguments.
		
		\begin{remark}
		We argued that $\lambda$ is the parameter generating the local diffeomorphisms normal to the boundary. We now also see in \eqref{e:ham_vf_L} that $\mathbb{H}_\phi$ depends on $(\Pi,e_n)$. In the cylider $\Sigma\times[0,\epsilon]$ we can apply the equation of motion $(\Pi,e_n)=\partial_n \phi$, hence showing that the (infinitesimal) gauge transformation generated by $\mathbb{H}$ on $\phi$ depends on the transversal component of $\phi$, as predictable.
		\end{remark}
		\bigskip
		We now proceed to compute the Poisson brackets of the constraints. In the following computations we use integration by parts ($\diamondsuit$) and the following identities (for a proof of the second see \cite{CCS2020}):
		\begin{align}
			&\frac{1}{2}\iota_{[\xi,\xi]}A = - \frac{1}{2} \iota_{\xi}\iota_{\xi} d_{\omega_0}A + \iota_{\xi}d_{\omega_0}\iota_{\xi} A- \frac{1}{2} d_{\omega_0} \iota_{\xi}\iota_{\xi} A && \forall A \in \Omega^{i,j}_{\partial} \tag{$\spadesuit$}\\
			& \mathrm{L}_{\xi}^{\omega_0}\mathrm{L}_{\xi}^{\omega_0}B = \frac{1}{2}\mathrm{L}_{[\xi,\xi]}^{\omega_0}B + \frac{1}{2}[\iota_{\xi}\iota_{\xi}F_{\omega_0},B] && \forall B \in \Omega^{i,j}_{\partial}\tag{$\clubsuit$}\\
			&d_{\omega_0}(\omega_0-\omega)= F_{\omega_0} -F_{\omega} +\frac{1}{2}[\omega_0-\omega,\omega_0-\omega]; &&\tag{$\heartsuit$}
		\end{align}
	
		\begin{align*}
			\{L_c,  H_{\lambda}\} & = \int_{\Sigma} [c,e] \lambda e_n F_{\omega} + \frac{\lambda e_n}{4}e^2(\Pi,\Pi)[c,e]+{\lambda e_n}e\Pi d\phi [c,e] - \frac{\lambda}{2}e^2\Pi(\Pi,e_n)[c,e]\\
			&\qquad  +\frac{1}{2}[c,e]\Lambda  \lambda e_n e^2 + d_{\omega} c e (d_{\omega}(\lambda e_n) + \lambda \sigma) \\ 
			& = \int_{\Sigma} \lambda e_n\left( [c,e]  F_{\omega}  +\frac{1}{3!}[c,e^3]\Lambda  + \frac{1}{2\cdot 3!}[c,e^3](\Pi,\Pi) + \frac{1}{2}[c,e^2]\Pi d\phi\right) \\
			&\qquad + d_{\omega} c  d_{\omega}(\lambda e_n e) - \frac{\lambda}{3!}[c,e^3]\Pi(\Pi,e_n) \\
			&\overset{\diamondsuit}{=}  \int_{\Sigma} - [c, \lambda e_n ] \left( eF_{\omega}-\frac{\Lambda}{3!} e^3 - \frac{1}{2\cdot 3!}e^3(\Pi,\Pi) - \frac{1}{2}e^2\Pi d\phi \right)  \\
			&\qquad  -\frac{\lambda e_n}{2}e^2[c,\Pi]d\phi - \frac{\lambda}{3!}e^3[c,\Pi](\Pi,e_n)\\
			& = \int_{\Sigma} -[c, \lambda e_n ]^{(a)}e_a e F_{\omega} -[c, \lambda e_n ]^{(n)}e_n e F_{\omega}-\frac{1}{3!}\Lambda[c, \lambda e_n ]^{(n)}e_n e^3 \\
			&\qquad + [c, \lambda e_n ]^{(a)}e_a \frac{e^2}{2}\Pi d\phi +  [c, \lambda e_n ]^{(n)}e_n \frac{e^2}{2}\Pi d\phi + [c, \lambda e_n ]^{(n)}e_n \frac{e^3}{2\cdot 3!} (\Pi,\Pi)\\
			& = - P_{[c, \lambda e_n ]^{(a)}} + L_{[c, \lambda e_n ]^{(a)}(\omega - \omega_0)_a} - H_{[c, \lambda e_n ]^{(n)}};
		\end{align*}
		In the missing step we used that
		\begin{align*}
			-\frac{\lambda e_n}{2}e^2[c,\Pi]d\phi - \frac{\lambda}{3!}e^3[c,\Pi](\Pi,e_n) & = \frac{\lambda e_n}{3!}e^3\left( [c,\Pi]^{(a)}(\Pi,e_a) + [c,\Pi]^{(n)}(\Pi,e_n) \right)\\
			& = \frac{\lambda e_n}{3!}e^3 ([c,\Pi],\Pi) = \frac{\lambda e_n}{2\cdot 3!}e^3[c,(\Pi,\Pi)]=0.
		\end{align*}
		\begin{align*}
			\{P_{\xi}, P_{\xi}&\}  = \int_{\Sigma} \frac{1}{2}  \mathrm{L}_{\xi}^{\omega_0} (e e )\mathrm{L}_{\xi}^{\omega_0} (\omega - \omega_0) +  \frac{1}{2}  \mathrm{L}_{\xi}^{\omega_0} (e e ) \iota_{\xi}F_{\omega_0} + \xi(\phi) \mathrm{L}^{\omega_0}_\xi(p)   \\
			& \overset{\diamondsuit\clubsuit}{=} \int_{\Sigma} \frac{1}{4} \mathrm{L}_{[\xi,\xi]}^{\omega_0}(e e ) (\omega - \omega_0) + \frac{1}{4}[\iota_{\xi}\iota_{\xi}F_{\omega_0},e e ] (\omega - \omega_0) +  \frac{1}{2}  \mathrm{L}_{\xi}^{\omega_0} (e e ) \iota_{\xi}F_{\omega_0} 	 + d_{\omega_0}(\iota_\xi(p))\iota_\xi(d\phi) \\
			& = \int_{\Sigma} \frac{1}{4} \iota_{[\xi,\xi]}d_{\omega_0}(e e ) (\omega - \omega_0)+ \frac{1}{4}d_{\omega_0}\iota_{[\xi,\xi]}(e e ) (\omega - \omega_0) \\
			& \qquad + \frac{1}{4}[\iota_{\xi}\iota_{\xi}F_{\omega_0},e e ] (\omega - \omega_0) +  \frac{1}{2}  \mathrm{L}_{\xi}^{\omega_0} (e e ) \iota_{\xi}F_{\omega_0} +\iota_\xi(d_{\omega_0}\iota_\xi(p))d\phi\\
			& \overset{\diamondsuit\spadesuit}{=} \int_{\Sigma} \frac{1}{4} \iota_{[\xi,\xi]}d_{\omega}(e e ) (\omega - \omega_0)-\frac{1}{4} \iota_{[\xi,\xi]}[\omega - \omega_0,e e ] (\omega - \omega_0) \\
			& \qquad+ \frac{1}{4}\iota_{[\xi,\xi]}(e e ) d_{\omega_0}(\omega - \omega_0) + \frac{1}{4}[\iota_{\xi}\iota_{\xi}F_{\omega_0},e e ] (\omega - \omega_0) +  \frac{1}{2}  \mathrm{L}_{\xi}^{\omega_0} (e e ) \iota_{\xi}F_{\omega_0} + \frac{1}{2}\iota_{[\xi,\xi]}(p)d\phi\\
			& \overset{\heartsuit}{=} \int_{\Sigma} \frac{1}{4} d_{\omega}(e e ) \iota_{[\xi,\xi]}(\omega - \omega_0)-\frac{1}{4} [\omega - \omega_0,e e ] \iota_{[\xi,\xi]}(\omega - \omega_0)- \frac{1}{4}\iota_{[\xi,\xi]}(e e ) F_{\omega_0} \\
			& \qquad+\frac{1}{4}\iota_{[\xi,\xi]}(e e )F_{\omega} -\frac{1}{8}\iota_{[\xi,\xi]}(e e )[\omega_0-\omega,\omega_0-\omega] \\
			& \qquad + \frac{1}{4}[\iota_{\xi}\iota_{\xi}F_{\omega_0},e e ] (\omega - \omega_0) +  \frac{1}{2}  \mathrm{L}_{\xi}^{\omega_0} (e e ) \iota_{\xi}F_{\omega_0} + \frac{1}{2}\iota_{[\xi,\xi]}(p)d\phi\\
			& \overset{\spadesuit}{=} \int_{\Sigma} \frac{1}{4} d_{\omega}(e e ) \iota_{[\xi,\xi]}(\omega - \omega_0)+\frac{1}{4}\iota_{[\xi,\xi]}(e e )F_{\omega}+ \frac{1}{4}d_{\omega_0}(e e )\iota_{\xi}\iota_{\xi} F_{\omega_0} \\
			& \qquad + \frac{1}{2}d_{\omega_0}\iota_{\xi}(e e ) \iota_{\xi} F_{\omega_0}- \frac{1}{4}\iota_{\xi}\iota_{\xi}F_{\omega_0} [\omega - \omega_0,e e]  \\
			& \qquad +  \frac{1}{2} \left( \iota_{\xi}d_{\omega_0} (e e )-  d_{\omega_0} \iota_{\xi} (e e ) \right) \iota_{\xi}F_{\omega_0} + \frac{1}{2}\iota_{[\xi,\xi]}(p)d\phi \\
			& = \int_{\Sigma} \frac{1}{4} d_{\omega}(e e ) \iota_{[\xi,\xi]}(\omega - \omega_0)+\frac{1}{4}\iota_{[\xi,\xi]}(e e )F_{\omega} + \frac{1}{2}\iota_{[\xi,\xi]}(p)d\phi -\frac{1}{4}d_{\omega}(e e )\iota_{\xi}\iota_{\xi} F_{\omega_0}  \\
			& = \frac{1}{2}P_{[\xi, \xi]} - \frac{1}{2}L_{\iota_{\xi}\iota_{\xi}F_{\omega_0}};
		\end{align*}
		\begin{align*}
			\{L_c, L_c\} & = \int_{\Sigma}  [c,e] e  d_{\omega} c = \int_{\Sigma}  \frac{1}{2} [c,ee]  d_{\omega} c \\
			&= \int_{\Sigma} \frac{1}{4} d_{\omega}[c , c] ee = \int_{\Sigma} -\frac{1}{2} [c , c]  e d_{\omega}e = - \frac{1}{2} L_{[c,c]};
		\end{align*}
		\begin{align*}
			\{L_c, P_{\xi}&\} = \int_{\Sigma} - [c,e] e(\mathrm{L}_{\xi}^{\omega_0} (\omega-\omega_0) + \iota_ {\xi}F_{\omega_0}) - d_{\omega} c e  \mathrm{L}_{\xi}^{\omega_0} e \\
			& = \int_{\Sigma} \frac{1}{2} \left(\mathrm{L}_{\xi}^{\omega_0}c [\omega- \omega_0, ee]+ c [\omega- \omega_0,\mathrm{L}_{\xi}^{\omega_0}( ee)]- c [ee, \iota_ {\xi}F_{\omega_0}]-  d_{\omega} \mathrm{L}_{\xi}^{\omega_0} (e e) c\right)\\
			& = \int_{\Sigma} \frac{1}{2} \mathrm{L}_{\xi}^{\omega_0}c [\omega, ee]- \frac{1}{2} d c \iota_{\xi} d(ee) + \frac{1}{2} [ \iota_{\xi}\omega_0, d(ee)] c\\
			& =  \int_{\Sigma} \frac{1}{2}  \mathrm{L}_{\xi}^{\omega_0}c d_\omega(ee) = \int_{\Sigma} \mathrm{L}_{\xi}^{\omega_0}c e d_\omega e = L_{\mathrm{L}_{\xi}^{\omega_0}c};
		\end{align*}

		\begin{align*}
			\{P_{\xi},H_{\lambda}\} & = \int_{\Sigma} - \mathrm{L}_{\xi}^{\omega_0} e \lambda e_n F_{\omega} -\frac{1}{2}\Lambda \mathrm{L}_{\xi}^{\omega_0} e \lambda e_n e^2 - \frac{\lambda e_n}{2\cdot 3!}(\Pi,\Pi)\mathrm{L}^{\omega_0}_\xi(e^3) - \frac{\lambda e_n}{2}\Pi d\phi \mathrm{L}^{\omega_0}_\xi(e^2)\\
			& \qquad + \frac{\lambda}{3!}\Pi(\Pi,e_n)\mathrm{L}^{\omega_0}_\xi(e^3) - \left( \mathrm{L}_{\xi}^{\omega_0} (\omega-\omega_0)+ \iota_ {\xi}F_{\omega_0}\right)e(d_{\omega}(\lambda e_n) + \lambda \sigma)   \\
			&\qquad + \lambda \mathrm{L}^{\omega_0}_\xi(p)(\Pi,e_n) - \frac{1}{2}d_\omega\left( \lambda e_n e^2 \Pi \right)\iota_\xi d\phi \\  
			& = \int_{\Sigma} - \mathrm{L}_{\xi}^{\omega_0} e \lambda e_n F_{\omega}-\frac{1}{3!}\Lambda \mathrm{L}_{\xi}^{\omega_0} e^3 \lambda e_n  -\left( \mathrm{L}_{\xi}^{\omega_0} (\omega-\omega_0)+ \iota_ {\xi}F_{\omega_0}\right) d_{\omega}(e \lambda e_n)\\
			& \qquad + \mathrm{L}^{\omega_0}_\xi\left( \frac{\lambda e_n}{2\cdot 3!} \right) e^3 (\Pi,\Pi) + \frac{\lambda e_n}{2\cdot 3!} e^3\mathrm{L}^{\omega_0}_\xi(\Pi,\Pi)- \frac{\lambda e_n}{2}\Pi d\phi \mathrm{L}^{\omega_0}_\xi(e^2) \\
			&\qquad  + \frac{\lambda}{3!}\Pi(\Pi,e_n)\mathrm{L}^{\omega_0}_\xi(e^3) + \frac{\lambda}{3!}\mathrm{L}^{\omega_0}_\xi(\pi^n e_n)e^3 (\Pi,e_n) - \frac{\lambda e_n}{3!}\mathrm{L}^{\omega_0}_\xi(e^3)\pi^n \\
			&\qquad + \mathrm{L}^{\omega_0}_\xi\left(\frac{\lambda e_n}{2}\right)e^2 \Pi d\phi + \frac{\lambda e_n}{2}\mathrm{L}^{\omega_0}_\xi(e^2)\Pi d\phi + \frac{\lambda e_n}{2}e^2\mathrm{L}^{\omega_0}_\xi(\Pi)d\phi\\
			& = \int_{\Sigma}  \mathrm{L}_{\xi}^{\omega_0} (\lambda e_n) \left( eF_\omega + \frac{e^2}{2}\Lambda + \frac{e^2}{4}(\Pi,\Pi) + e\Pi d\phi \right)  \\ 
			& \qquad + e   \lambda e_n \mathrm{L}_{\xi}^{\omega_0} F_{\omega} + \left( d_{\omega} \iota_{\xi} (\omega-\omega_0)- \iota_ {\xi}F_{\omega}\right) d_{\omega}(e \lambda e_n)  \\
			& \qquad + \lambda e_n \left[ \frac{e^3}{3!}( \Pi,\mathrm{L}_{\xi}^{\omega_0}(\Pi)) + \frac{e^3}{2} (e,\Pi)\mathrm{L}_{\xi}^{\omega_0}(\Pi) \right] +   \frac{\lambda}{3!}e^3 (\Pi,e_n)\mathrm{L}_{\xi}^{\omega_0}(\Pi) \\
			&\overset{\bigstar}{=}\int_{\Sigma}  \mathrm{L}_{\xi}^{\omega_0} (\lambda e_n) \left( eF_\omega + \frac{e^2}{2}\Lambda + \frac{e^2}{4}(\Pi,\Pi) + e\Pi d\phi \right)  \\ 
			& \qquad + e   \lambda e_n \mathrm{L}_{\xi}^{\omega_0} F_{\omega} + \left( d_{\omega} \iota_{\xi} (\omega-\omega_0)- \iota_ {\xi}F_{\omega}\right) d_{\omega}(e \lambda e_n)  \\
			& \qquad -  \frac{\lambda e_n}{2}e^3 (e,\Pi)\mathrm{L}_{\xi}^{\omega_0}(\Pi) - \frac{\lambda}{3!}e^3 (\Pi,e_n)\mathrm{L}_{\xi}^{\omega_0}(\Pi) \\
			&\qquad + \frac{\lambda e_n}{2}e^3 (e,\Pi)\mathrm{L}_{\xi}^{\omega_0}(\Pi) +   \frac{\lambda}{3!}e^3 (\Pi,e_n)\mathrm{L}_{\xi}^{\omega_0}(\Pi) \\
			& = \int_{\Sigma}    \mathrm{L}_{\xi}^{\omega_0} (\lambda e_n) \left( eF_\omega + \frac{e^2}{2}\Lambda + \frac{e^2}{4}(\Pi,\Pi) + e\Pi d\phi \right) \\
			& = P_{ \mathrm{L}_{\xi}^{\omega_0} (\lambda e_n)^{(a)}} +H_{ \mathrm{L}_{\xi}^{\omega_0} (\lambda e_n)^{(n)}}-L_{ \mathrm{L}_{\xi}^{\omega_0} (\lambda e_n)^{(a)} (\omega - \omega_0)_a},
		\end{align*}

		where we used that $(e,\Pi)=d\phi$

		Finally,	
		\begin{align*}
			\{H_\lambda, H_\lambda\} & =\int_{\Sigma}\left[ \lambda e_n \left( F_\omega + \frac{\Lambda}{2}e^2 +\frac{e^2}{4}(\Pi,\Pi) + e\Pi d\phi\right) - \frac{\lambda}{2}e^2(e_n,\Pi)\Pi \right](d_\omega(\lambda e_n)+\lambda \sigma) \\
			&\qquad - \lambda(e_n,\Pi)d_\omega\left(\frac{\lambda e_n}{2}e^2\Pi\right)\\
			& = \int_{\Sigma} \frac{\lambda}{2}d\lambda e_n e^2 (e_n,\Pi)\Pi - \frac{\lambda}{2}d\lambda e_n e^2 (e_n,\Pi)\Pi = 0 ,
		\end{align*}
		since most of the terms vanish because $e_n^2=0$ and $\lambda^2=0$.
	\end{proof}

	\subsection{BFV Formalism}\label{sec:scalBFV}
	In this section we apply the content of Section \ref{s:BFV_theory}. In particular, we embed $F_\partial$ as the body of a supermanifold $\mathcal{F}$, whose odd coordinates are given by taking the Lagrange multipliers as fields (the ghosts) and adding their momenta (ghost momenta). The result is presented in the following theorem, where we use the notation and quantities of the analogous Theorem \ref{thm:BFVgravity} in which the BFV theory of gravity without matter is described.
	
	\begin{theorem}\label{thm:BFVaction}
		Let $\mathcal{F}_{S}$ be the bundle 
		\begin{equation*}
\mathcal{F}_{S} \longrightarrow \Omega_{nd}^1(\Sigma, \mathcal{V}),
\end{equation*}
with local trivialisation on an open $\mathcal{U}_{\Sigma} \subset \Omega_{nd}^1(\Sigma, \mathcal{V})$
		\begin{equation*}
			\mathcal{F}_S\simeq \mathcal{T}_{PC}\times \Omega^{(0,1)}_{\partial,\text{red}} \times \mathcal{C^\infty}(\Sigma)
		\end{equation*}
		where $\mathcal{T}_{PC}$ was defined in \eqref{LoctrivF1} and the additional fields are denoted by $\Pi\in\Omega^{(0,1)}_{\partial,\text{red}}$ and $\phi\in\mathcal{C^{\infty}}(\Sigma)$ and such that  they satisfy the structural constraints $(e,\Pi)=d\phi$. 
		The symplectic form and the action functional on $\mathcal{F}_S$ are respectively defined by
		\begin{align*}
			\varpi_S &= \varpi_{PC}+\int_{\Sigma}\frac{1}{3!}\delta(e^3\Pi)\delta\phi, \\
			S_S &= S_{PC} +\int_{\Sigma} \frac{1}{3!}\iota_\xi(e^3\Pi)d\phi
			 + \lambda e_n \left( \frac{1}{2\cdot 3!}e^3(\Pi,\Pi) + \frac{1}{2}e^2\Pi d\phi \right).
		\end{align*}
		Then the triple $(\mathcal{F}_S, \varpi_S, S_S)$ defines a BFV structure on $\Sigma$.
	\end{theorem}

	\begin{proof}
		We follow the same strategy of \cite{CCS2020}, from which we also borrow the notation. The only bit that we need to prove, is that the new BFV action $S_S$ still satisfies  the classical master equation
			\begin{equation}
				\{S_S,S_S\}=\iota_{Q_S} \iota_{Q_S} \varpi_S =0,
			\end{equation}
		where $Q_S$ is the Hamiltonian vector field of $S_S$, defined by $\iota_{Q_S} \varpi_S= \delta S_S$. In order to do so, we can exploit the results of \cite{CCS2020} and by linearity we get 
		\begin{align*}
		    \{S_S,S_S\}=\{S_{PC},S_{PC}\}+2\{S_{PC},S_{add}\}+\{S_{add},S_{add}\}
		\end{align*}
		where we denoted by $S_{add}$ the part of $S_S$ containing the scalar field and its momentum. We have that $\{S_{PC},S_{PC}\}=0$ from Theorem \ref{thm:BFVgravity}.
		The remaining part $2\{S_{PC},S_{add}\}+\{S_{add},S_{add}\}=0$ is instead a consequence of Theorem \ref{thm:first-class-constraints_scalar}. Indeed, the explicit computation of the second bracket follows verbatim the computation of the brackets between the constraints in the proof of the aforementioned theorem by just considering only  the terms containing $\Pi$ or $\phi$. Nonetheless, the first bracket produces in a trivial way exactly the results of these brackets, since $S_{add}$ does not depend on ghost momenta.
		\end{proof}

\begin{remark}
    The BFV structure of Theorem \ref{thm:BFVaction} depends on a reference connection $\omega_0$. However, performing a change of variables it is possible to obtain a BFV theory not depending on it that still represents a cohomological resolution of the reduced phase space. The precise expression of the change of variables is given in \cite{CCS2020} for the PC theory without matter and does not change in presence of a scalar field.
\end{remark}

\section{Yang--Mills coupled to gravity}\label{s:YM}
	
	We now move our attention to the more complicated (but also more physically interesting) case of the coupling of Yang--Mills field to gravity. Also in this case it is useful to work in the first order formalism.
	
	We start by considering a principal bundle $(R,G,\pi,M)$ over the $N$-dimensional space--time manifold $M$. We assume $G$ to be a compact Lie group with 
	Lie algebra $\mathfrak{g}$.\footnote{All the following considerations actually work for any Lie algebra g.} 
	
The gauge field is defined to be the connection 1-form $A$.
Let $\{T_I\}$ be a basis for $\mathfrak{g}$, then we express locally $A$ as\footnote{Note that we use uppercase latin letters to denote the indices of this Lie algebra in order to distinguish them from the indices of the vector bundle $\mathcal{V}$ which are denoted with lowercase latin letters.} 
		\begin{equation}
			A=A^{I}(x)T_I= A^I_{\mu}T_I dx^{\mu}.
		\end{equation}
	In particular, the gauge fields are in a space locally modeled on $\Gamma(T^*M\otimes \mathfrak{g})$, which we will denote by $\mathcal{A}_{\text{YM}}$. The curvature two-form is as usual defined to be $F_A:=dA+\frac{1}{2}[A,A]$. In coordinates, it reads
		\begin{equation}
			F_A=\left(dA^I + \frac{1}{2}f^I_{JK}A^J A^K\right)T_I=F^I T_I,
		\end{equation}
	where $F^I=\frac{1}{2}F^I_{\mu\nu}dx^\mu\wedge dx^\nu$.
	
	The gauge invariant quantity that we can construct starting from $A$ is Tr$(F_A \wedge \star F_A)$, where $\star$ denotes the Hodge dual. However, in order to define it, we need to use the metric tensor, which as we know is not the fundamental object of our field theoretical description and is found in terms of the vielbein. As in the case of the scalar field, we then need to find a way to encode the dynamics of the Yang--Mills field in an action functional containing the vielbein. To do so, we introduce an independent field $B\in \Gamma(\wedge^2 \mathcal{V}\otimes \mathfrak{g})$, which is a $\mathfrak{g}$-valued section of the second exterior power of the Minkwoski bundle $\mathcal{V}$. In coordinates, it reads $B=B^{\mu\nu|I}e_\mu e_\nu T_I$, where we used $\{e_{\mu}\}$ as a local basis for $V$.
	
	The Yang--Mills action in the first order formalism is
		\begin{equation}
			S_{\text{YM}}:=\int_M \frac{1}{(N-2)!}e^{N-2}\mathrm{Tr}(B F_A) + \frac{1}{2N!}e^N \mathrm{Tr}(B,B),
		\end{equation}
	where $(\cdot,\cdot)$ is the canonical pairing in $\wedge^2 V$ defined in coordinates for all $C,D\in \wedge^2 V$ by $(C,D):=C^{ab}D^{cd}\eta_{ac}\eta_{bd}$ with respect to an orthonormal basis $\{u_a\}$ of $V$.
		
	We compute the variation of the action $S=S_{PC}+S_{\text{YM}}$ and find
		\begin{equation}
			\begin{split}
			\delta S&=\int_M \left[ \frac{e^{N-3}}{(N-3)!}\left( F_\omega + \mathrm{Tr}(B F_A) \right) + \frac{e^{N-1}}{(N-1)!}\left( \Lambda + \frac{1}{2}\mathrm{Tr}(B,B) \right)\right]\delta e \\
			&\qquad + \frac{1}{(N-2)!}d_\omega(e^{N-2})\delta\omega + \frac{e^{N-2}}{(N-2)!}\mathrm{Tr}\left[\left( F_A + \frac{1}{2}(e^2,B) \right)\delta B\right] \\
			& \qquad + \mathrm{Tr}\left[ d_A\left( \frac{e^{N-2}}{(N-2)!}B \right)\delta A \right] -d\left\{\frac{e^{N-2}}{(N-2)!}\left[ \delta \omega + \mathrm{Tr}\left( B \delta A \right) \right]\right\},
			\end{split}
		\end{equation}	
	where to extract $\delta B$ out of the bracket we used the following identity,\footnote{See Lemma \ref{useful identities} in Appendix \ref{a:tec_and_lproofs}.} holding for all $C\in\Omega^{(0,2)}$ and $D\in\Omega^{0,2}[1]$ (the fact that they might also have values in $\mathfrak{g}$ is here irrelevant):
		\begin{equation}\label{delta B}
		\frac{e^N}{N!}(C,D)=\frac{e^{N-2}}{2(N-2)!}(e^2,C)D.
		\end{equation}	
	First of all, we notice that the variation of the action produces a boundary term, which will be the local 1-form on the space of preboundary fields whose vertical differential will give rise to the presymplectic two-form on the boundary. It is given by
		\begin{equation}\label{boundary YM}
			\tilde{\alpha}_{\text{YM}}=\int_{\partial M} \frac{e^{N-2}}{(N-2)!} \delta \omega + \frac{e^{N-2}}{(N-2)!}\mathrm{Tr}\left( B \delta A \right).
		\end{equation}
	The equations of motion are found to be
		\begin{eqnarray}
			 &d_{\omega}e=0 \label{torsion0 YM};\\
			 &\frac{e^{N-3}}{(N-3)!}\left( F_\omega + \mathrm{Tr}(B F_A) \right) + \frac{e^{N-1}}{(N-1)!}\left( \Lambda + \frac{1}{2}\mathrm{Tr}(B,B) \right) \label{Einstein YM};\\
			 & e^{N-2}\left(F_A + \frac{1}{2}(e^2,B)\right) = 0 \label{B YM_0};\\
			 & d_A(e^{N-2}B)=0 \label{gauss YM_0}.
		\end{eqnarray}
	Equation \eqref{B YM_0} can be further simplified by noticing that $W_{N-2}^{(2,0)}$ is injective.\footnote{See Lemma 	\hyperref[lem: W_N-2 bijective]{\ref*{useful id bulk}.(\ref*{lem: W_N-2 bijective})} in Appendix \ref{a:tec_and_lproofs}.} Therefore we obtain
		\begin{equation}
			F_A + \frac{1}{2}(e^2,B)=0  \label{B YM_1},
		\end{equation}
	which in coordinates gives $B^{\mu\nu}=(-1)^Ng^{\mu\rho}g^{\nu\sigma}F_{\rho\sigma}$ (omitting the Lie algebra indices). With this definition, using Corollary \ref{corollary bulk 1} we then find 
		\begin{equation}
			\frac{e^{N-2}}{(N-2)!}BF_A + \frac{e^N}{2N!}(B,B)=-\frac{1}{2}\text{Vol}_g F_{\mu\nu}F^{\mu\nu},
		\end{equation}
	giving (up to factors) the standard Yang--Mills term in the action. 
	
	In the next section we will analyze the boundary structure.

	\subsection{Boundary Structure in \texorpdfstring{$N=4$}{dimension 4}}
	We assume the manifold $M$ to be $4$-dimensional with boundary $\Sigma:=\partial M$. Unlike the case of the scalar field, we will see that the equations of motion produce an additional constraint, hence modifying the boundary structure (but still preserving the first class condition) and the BFV description.

	The boundary term in \eqref{boundary YM} reads
	\begin{equation*}
		\tilde{\alpha}_\text{YM}=\frac{1}{2}\int_\Sigma e^2\delta\omega + \Tr(e^2 B \delta A).
	\end{equation*}

	Here $B$ and $A$ are the fields restricted to the boundary, while $e$ and $\omega$ are as in the previous section, in particular
	\begin{itemize}
		\item $B$ is an element of $\Omega^{(0,2)}_{\partial,\mathfrak{g}}=\Omega^{(0,2)}_\partial\otimes \mathfrak{g}$;
		\item $A$ is an element of $\mathcal{A}^{\text{YM}}_\partial$, locally represented by $\Omega^{(1,0)}_\partial\otimes \mathfrak{g}$.
	\end{itemize}

	The space of preboundary fields is denoted by $\tilde{F}^{\text{YM}}_\partial=\Omega^{(1,1)}_{\text{n.d.}}\times\mathcal{A}_\partial\times\mathcal{A}^{\text{YM}}_\partial\times\Omega^{(0,2)}_{\partial,\mathfrak{g}}$.
	The presymplectic form on $\tilde{F}^{\text{YM}}_\partial$ is defined as the variation of $\tilde{\alpha}_\text{YM}$
	\begin{equation}
		\tilde{\varpi}_\text{YM}:=\int_{\Sigma}e\delta e \delta \omega + \Tr(eB\delta e \delta A) + \frac{1}{2}\Tr(e^2\delta B \delta A).
	\end{equation}
	We are interested in computing the kernel of $\tilde{\varpi}_\text{YM}$ defined
	as $$\Ker{\tilde{\varpi}_\text{YM}}:=\{X\in T\tilde{F}_\partial^\text{YM}\hspace{1mm}\vert\hspace{1mm} \iota_X \tilde{\varpi}_\text{YM}=0\}.$$ Considering a generic vector field $\mathbb{X}=\mathbb{X}_e\frac{\delta}{\delta e} + \mathbb{X}_\omega\frac{\delta}{\delta \omega} + \mathbb{X}_A\frac{\delta}{\delta A} + \mathbb{X}_B \frac{\delta}{\delta B}$, we find ker$(\tilde{\varpi}_\text{YM})$ as the vector fields satisfying
	\begin{align}
		&e\mathbb{X}_e=0  \label{e ker};\\ 
		& e\mathbb{X}_\omega + e B \mathbb{X}_A=0  \label{omega ker};\\
		& e B \mathbb{X}_e + \frac{1}{2}e^2 \mathbb{X}_B = 0  \label{B ker};\\
		& e^2 \mathbb{X}_A = 0 \label{A ker}.
	\end{align}
	We now see by the previous section that \eqref{e ker} is solved by $\mathbb{X}_e=0$, while \eqref{A ker} is solved by $\mathbb{X}_A=0$ by Lemma \eqref{lem: We7}, therefore we are left with $e\mathbb{X}_\omega=0$ and $e^2 \mathbb{X}_B = 0 $.

	As usual, we define the geometric space $F_\partial^\text{YM}$ to be the symplectic reduction of $\tilde{F}^\text{YM}_\partial$, namely it is a bundle over $\Omega^{(1,1)}_{\partial,\text{n.d.}}$ with local trivialization on an open $\mathcal{U}_{\Sigma} \subset \Omega_{nd}^1(\Sigma, \mathcal{V})$  $$F_\partial^\text{YM}\simeq \mathcal{U}_{\Sigma}\times\mathcal{A}^\text{red}_\partial\times\mathcal{A}^{\text{YM}}_\partial\times\Omega^{(0,2)}_{\partial,\text{red}},$$ where $\mathcal{A}^\text{red}_\partial$  was defined in Section \ref{s:rps_gravity} and $\Omega^{(0,2)}_{\partial,\text{red}}:=\Omega^{(0,2)}_{\partial,\mathfrak{g}}/\sim$ with 
	\begin{align}
		B\sim\tilde{B} \hspace{5mm}&\Leftrightarrow \hspace{5mm}B-\tilde{B}=C \hspace{2mm}\text{with}\hspace{2mm}e^2C=0.
	\end{align}
	
	$F_\partial^\text{YM}$ is thus a symplectic manifold with symplectic form
		\begin{equation}
			\tilde{\varpi}_{\text{YM}}=\int_{\Sigma}e\delta e \delta [\omega] + \frac{1}{2}\Tr(\delta(e^2 [B])\delta A).
		\end{equation}
	\begin{remark}
	As one can easily notice, we can rewrite the part of $\varpi_{\text{YM}}$ depending on $A$ and $B$ in Darboux form, by defining $\rho:=\frac{1}{2}e^2 B$, since in this way the components of $B$ which are in the kernel of $e^2$ are automatically suppressed. Therefore we obtain the symplectic form as
		\begin{equation}
			\tilde{\varpi}_{\text{YM}}=\int_{\Sigma}e\delta e \delta [\omega] +\Tr( \delta \rho \delta A).
		\end{equation}
	We can as well consider a generic vector field $\mathbb{X}=\mathbb{X}_e \frac{\delta}{\delta e} + \mathbb{X}_\omega\frac{\delta}{\delta\omega} + \mathbb{X}_\rho\frac{\delta}{\delta\rho} + \mathbb{X}_A\frac{\delta}{\delta A}$, then it will be useful to consider $\iota_\mathbb{X}\varpi_\mathrm{YM}$
		\begin{equation}
			\iota_\mathbb{X}\varpi_\mathrm{YM}=\int_{\Sigma} e \mathbb{X}_e \delta \omega + e \delta e \mathbb{X}_\omega + \Tr(\mathbb{X}_\rho \delta A) + \Tr(\delta \rho \mathbb{X}_A).
		\end{equation}
		
	\end{remark}
	
	As we saw in the previous section, to obtain the physical space of fields on the boundary (i.e. the reduced phase space) we need to impose constraints on $F_\partial^\text{YM}$. Recall that the equations of motion split into evolution equations (containing the derivatives of the fields in the transversal direction with respect to the boundary) and in the constraints, which contain only derivatives tangential to the boundary. 
	The latter are readily obtained as the restriction of the equations of motion to the boundary. 
	
	\subsubsection{Choice of representative via constraints}	
	
	We now fix the representatives of the fields in the geometric phase space. In order to do so, we make use of the constraints, which in $N=4$ are
		\begin{eqnarray}
			&&d_\omega e  =  0  ;\\
			&&e F_\omega + \frac{\Lambda}{3!}e^{3}  + \Tr\left[eBF_A+\frac{1}{2\cdot3!}e^{3}(B,B) =0\right]  \label{einstein em 1};\\
			&&d_A(e^2B) = 0\label{gaussB_2};\\
			&&F_A + \frac{1}{2}(e^2,B)=0 \label{A eq_2}.
		\end{eqnarray}
	The choice of the representative of $[\omega]$ is performed exactly as in Section \ref{rep scalar}.

	To fix the representative of $[B]$ we use \eqref{A eq_2} in an analogous way.
	In particular, we exploit the property of the following Lemma which will be proved in Appendix \ref{a:tec_and_lproofs}.
	\begin{lemma}\label{lem:technical_for_YM}
	    If  $g^\partial$ is nondegenerate, then the map $\phi_e: \Ker{W_2^{\partial}(0,2)}\rightarrow\Omega_{\partial}^{1,0}$, $\phi_e(b)=\frac{1}{2}(e^2,B)$ is bijective.
	\end{lemma}
	
	Analogously to the case of the scalar field, this lemma provides the tools to prove that \eqref{A eq_2} fixes uniquely the representative of the equivalence class of $[B]$ in an appropriate way:
	\begin{theorem}\label{constr rep ym}
	 Let $g^\partial$ be nondegenerate. Given any $\widetilde{B}\in \Omega_{\partial}^{0,2}\otimes \mathfrak{g}$, there is a unique decomposition $\widetilde{B}= B+b$ such that $b \in \Ker{W_2^{\partial(0,2)}}\otimes \mathfrak{g}$ and
	 \begin{align}\label{e:constraintYM}
	     F_A + \frac{1}{2}(e^2,B)=0
	 \end{align}
	\end{theorem}
	\begin{proof}
	    If $\widetilde{B}$ satisfies \eqref{e:constraintYM} we can just choose $b=0$. On the contrary, suppose that $(e,\widetilde{B})+F_A=K$, then since $\phi_e$ is bijective, there exists a $b \in \Ker{W_3^{\partial(0,1)}}\otimes \mathfrak{g}$ such that $K=-\frac{1}{2}(e^2,b)$. Then $B= \widetilde{B}-b$ satisfies \eqref{e:constraintYM}. 
	    
	    Uniqueness goes exactly as in the case of the scalar field.
	\end{proof}

	\subsubsection{Poisson brackets of the constraints}
	Having defined a symplectic manifold, it is of course possible to define the induced Poisson structure. In this section we will show that also in the case of a Yang--Mills field coupled to gravity the boundary structure is such that it produces first-class constraints, namely a set of functions on the space of fields on the boundary which is algebraically closed with respect to the Poisson bracket. 
	
	As in the case of the scalar field, we use Lagrange multipliers, and we split the constraint \eqref{einstein em 1} (the projection of Einstein's equations to the boundary) into two independent ones. We are left with four constraints:
	\begin{align}
		& L_c:=\int_{\Sigma} c ed_\omega e \label{L constr YM};\\
		& M_\mu:=\int_{\Sigma}\frac{1}{2}\Tr(\mu d_A (e^2B)) \label{M constr YM}; \\
		& 	P_\xi:=\int_{\Sigma} \frac{1}{2}\iota_{\xi}e^2 F_\omega + \frac{1}{2}\iota_\xi e^2 \Tr(BF_A) + \iota_\xi(\omega-\omega_0)ed_\omega e \\
			&\qquad + \frac{1}{2}\Tr\{\iota_{\xi}(A-A_0)d_A(e^2B)\} \label{P constr YM};\\
		& H_\lambda:=\int_{\Sigma} \lambda e_n \left( eF_\omega + \frac{\Lambda}{3!}e^3 + e \Tr(B F_A) + \frac{1}{2\cdot 3!} e^3 \Tr(B,B)\right) \label{H constr YM}	.\\	 
	\end{align}
	\begin{remark}
		Notice that the constraint $M_\mu$ can be rewritten in terms of the fields in Darboux form simply as
		\begin{equation}\label{M constr em Darboux}
			M_\mu=\int_{\Sigma}Tr(\mu d_A\rho).
		\end{equation}
		Concerning $P_\xi$, we added the term $\frac{1}{2}\Tr\{\iota_{\xi}(A-A_0)d_A(e^2B)\}$ with respect to a reference connection $A_0$. Again, this addition does not change the properties of the boundary structure (we are simply adding a term that vanishes on the submanifold defined as the zero-locus of the constraints), but it largely simplifies the calculations, since it allows to find a more explicit form of the Hamiltonian vector field. We might as well rewrite $P_\xi$ in terms of $\rho$ as
		\begin{equation}\label{P constr em Darboux}
			P_\xi=\int_{\Sigma} \frac{1}{2}\iota_{\xi}e^2 F_\omega + \frac{1}{2} \Tr(\iota_\xi\rho F_A) + \iota_\xi(\omega-\omega_0)ed_\omega e + \Tr\{\iota_{\xi}(A-A_0)d_A\rho\}.
		\end{equation}
	\end{remark}

	The Lagrange multipliers are again chosen to be odd, in particular we have $\lambda\in\mathcal{C^{\infty}}[1](\Sigma)$,	$\mu\in\Gamma(\mathfrak{g})[1]$,  $\xi\in\mathfrak{X}[1](\Sigma)$ and $c\in \Omega^{(0,2)}_\partial[1]$.
	
	\begin{remark}
		The new constraint $M_\mu$ is associated with the $G$ gauge symmetry of the Yang--Mills field. In particular, we will see in the proof of Theorem \ref{thm:first-class-constraints_YM} that the Hamiltonian vector field associated to $M_\mu$ exactly generates the infinitesimal $G$ gauge transformations. Furthermore, we notice an analogy between $M$ and $L$, which is not surprising since they both encode the gauge symmetry of the fields, respectively given by a compact Lie group $G$ and by $SO(3,1)$.
	\end{remark}
.
	
	\begin{theorem} \label{thm:first-class-constraints_YM}
		The constraints $L_c$, $M_\mu$, $P_{\xi}$, $H_{\lambda}$ define a coisotropic submanifold  with respect to the symplectic structure $\varpi_{\text{YM}}$. Their Poisson brackets\footnote{We point out that one should not confuse $L$ with $\mathrm{L}$, which respectively indicate the constraint and the Lie derivative} read
		
		\begin{align*}\label{brackets-of-constraints_YM}
			& \{P_{\xi}, P_{\xi}\}  =  \frac{1}{2}P_{[\xi, \xi]}- \frac{1}{2}L_{\iota_{\xi}\iota_{\xi}F_{\omega_0}}-\frac{1}{2}M_{\iota_{\xi}\iota_{\xi}F_{A_0}} &\{H_\lambda,H_\lambda\}=0\\
			&\{M_\mu,M_\mu\}=-\frac{1}{2}M_{[\mu,\mu]}	 & \{M_\mu,L_c\}=0\\
			&\{M_\mu,H_\lambda\}= 0& \{M_\mu,P_\xi\}=M_{\mathrm{L}_\xi^{A_0}\mu}\\
			&\{L_c, P_{\xi}\}  =  L_{\mathrm{L}_{\xi}^{\omega_0}c} & \{L_c, L_c\} = - \frac{1}{2}L_{[c,c]}  	;
		\end{align*}
		\begin{align*}
			& \{L_c,  H_{\lambda}\}  = - P_{X^{(a)}} + L_{X^{(a)}(\omega - \omega_0)_a} -H_{X^{(n)}} + M_{X^{(a)}(A-A_0)_{(a)}} {}\\
			& \{P_{\xi},H_{\lambda}\}  =  P_{Y^{(a)}} -L_{ Y^{(a)} (\omega - \omega_0)_a} +  H_{ Y^{(n)}} - M_{Y^{(a)}(A-A_0)_{(a)}},
		\end{align*}
	where $X= [c, \lambda e_n ]$, $Y = \mathrm{L}_{\xi}^{\omega_0} (\lambda e_n)$ and $Z^{(a)}$, $Z^{(n)}$ are the components of $Z\in\{X,Y\}$ with respect to the frame $(e_a, e_n)$.	
	\end{theorem}

	\begin{proof}
	We start by computing the Hamiltonian vector fields associated to the constraints. Many of the calculations will be exactly the same as in the previous section, therefore we refer to Section \ref{scal brackets} for the parts that we leave out.
	\begin{align*}
		\delta L_c=\int_{\Sigma} [c, e]  e \delta \omega  +  d_{\omega} c  e\delta e;
	\end{align*}
	\begin{align*}
		\delta P_{\xi} & = \int_{\Sigma}  - e \delta e (\mathrm{L}_{\xi}^{\omega_0} 	(\omega-\omega_0) + \iota_ {\xi}F_{\omega_0})  -  (\mathrm{L}_{\xi}^{\omega_0} e ) e \delta \omega +\Tr\left[\delta( \iota_{\xi}\rho F_A)+ \delta(\iota_\xi(A-A_0) d_A\rho) \right]\\
		& = \int_{\Sigma} (\cdots) - \Tr\left\{ \iota_\xi\delta \rho F_A 	-\iota_\xi\rho d_A\delta A -\iota_\xi(\delta A)d_A\rho  -\iota_\xi(A-A_0)[\delta A,\rho] + \iota_\xi(A-A_0)d_A\delta\rho\right\}\\
		&=\int_{\Sigma} (\cdots) - \Tr\left\{\delta\rho (\iota_\xi F_A - 	d_A\iota_\xi(A-A_0)) + (-\iota_\xi d_A\rho + d_A\iota_\xi\rho + [\iota_{\xi}(A-A_0),\rho])\delta A \right\}\\
		& = \int_{\Sigma}  - e \delta e (\mathrm{L}_{\xi}^{\omega_0} 	(\omega-\omega_0) + \iota_ {\xi}F_{\omega_0})  -  (\mathrm{L}_{\xi}^{\omega_0} e ) e \delta \omega  - \Tr\left\{ \delta \rho(\mathrm{L}_\xi^{A_0}(A-A_0) + \iota_\xi F_{A_0}) + \mathrm{L}_\xi^{A_0} \rho \delta A\right\} ;
	\end{align*}
	\begin{align*}
		\delta M_\mu&=\int_{\Sigma} \Tr\left[ \mu \delta(d_A\rho) \right]=\int_{\Sigma} \Tr\left[- \mu([\delta A,\rho] + d_A(\delta \rho))  \right]\\
		&=\int_\Sigma \Tr(\delta A[\mu,\rho] + d_{A}\mu \delta \rho );
	\end{align*}
	\begin{align*}
		\delta H_\lambda&=\int_{\Sigma}(\cdots) + \delta\Tr \left[ \lambda e_n BF_A + \frac{\lambda e_n}{2\cdot 3!}e^3 (B,B) \right]\\
		&=\int_{\Sigma}(\cdots) +\Tr\left\{ \lambda e_n\left[ BF_A + \frac{e^2}{4}(B,B) \right]\delta e + \lambda e_n e \delta B F_A -\lambda e_n e b d_A(\delta A) +\frac{\lambda e_n}{3!}e^3(B,\delta B) \right\}\\
		&\overset{\bigstar}{=}(\cdots) +\int_{\Sigma}\Tr\left\{ \lambda e_n\left[ BF_A + \frac{e^2}{4}(B,B) \right]\delta e + \lambda e_n e \delta B F_A + d_A(\lambda e_n e b )\delta A  \right\}\\
		&\qquad + \Tr\left\{ \frac{\lambda e_n}{2}e(e^2,B)\delta B + \frac{\lambda}{2}(B,e_n e)e^2 \delta B \right\}\\
		&=\int_{\Sigma}(\cdots) +\Tr\left\{ \left[\lambda e_n (BF_A + \frac{e^2}{4}(B,B)) - \lambda e B(B,e_n e) \right]\delta e  \right\}\\
		&\qquad + \Tr\left\{ d_A(\lambda e_n e B )\delta A + \lambda(B,e_n e)\delta \rho \right\},
	\end{align*}
	where we used a generalization of \eqref{delta B} to the boundary in $N=4$. Assuming $C\in\Omega^{(0,2)}_\partial$ and $D\in\Omega^{(0,2)}_\partial$, we find the following useful identity\footnote{See Lemma \ref{useful id boundary} for $N=4$ in Appendix \ref{a:tec_and_lproofs}}
		\begin{equation}\tag{$\bigstar$}
			\frac{\lambda e_n}{3!}e^3(C,D)= \frac{\lambda}{2}(C,e_n e)e^2 D + \frac{\lambda e_n}{2}e(e^2,C)D .
		\end{equation}
	
	The components of the Hamiltonian vector fields therefore are

	\begin{align*}
		&e\mathbb{H}_\omega = \lambda e_n \left( F_{\omega} + \frac{\Lambda}{2}e^2 + \frac{1}{4}e^2\Tr(B,B)
		+ \Tr(B F_A) \right) - \lambda e \Tr(B (B,e_n e))\\				&\mathbb{H}_e=d_\omega(\lambda e_n) + \lambda \sigma\\		&\mathbb{H}_\rho=d_A(\lambda e_n e B) \\
		& \mathbb{H}_A=\lambda(B,e e_n)  	
	\end{align*}
	\begin{align*}
		&\mathbb{L}_e = [c,e] &&\mathbb{L}_A=0 \\  
		&\mathbb{L}_\omega = d_{\omega} c + \mathbb{V}_L &&e^2\mathbb{L}_B=e^2[c,B ]\\
		&\mathbb{M}_e=0 && \mathbb{M}_A=d_A\mu \\		
		&\mathbb{M}_\omega=0 && \mathbb{M}_\rho=[\mu,\rho]\\
		&\mathbb{P}_e=-\mathrm{L}_\xi^{\omega_0}(e) && \mathbb{P}_\omega=-\mathrm{L}_\xi^{\omega_0}(\omega-\omega_0) -\iota_\xi(F_{\omega_0}) + \mathbb{V}_P\\
		&\mathbb{P}_\rho=-\mathrm{L}_\xi^{A_0}(\rho) && \mathbb{P}_A=-\mathrm{L}_\xi^{A_0}(A-A_0) -\iota_\xi(F_{A_0}) .
	\end{align*}
	We can now start computing the Poisson brackets of the constraints. We notice that since $L_c(\rho)=0$ and $L_c(A)=0$, the brackets $\{L_c,L_c\}$ and $\{L_c,P_\xi\}$ will be computed exactly as in Section \ref{scal brackets}.  Also $\{M_\mu,L_c\}=0$ is seen very easily without the need of any calculation.

	\begin{align*}
		\{M_\mu,M_\mu\}&=\int_{\Sigma} \Tr\left( d_A\mu[\mu,\rho] \right)=\int_{\Sigma}-\Tr\left( [\mu,d_A\mu]\rho \right)\\
		&=\frac{1}{2}\int_{\Sigma}\Tr\left( d_A[\mu,\mu]\rho \right)=-\frac{1}{2}\int_{\Sigma}\Tr\left( [\mu,\mu]d_A\rho \right)\\
		&=-\frac{1}{2}M_{[\mu,\mu]}			;
	\end{align*}

	\begin{align*}
		\{M_\mu,P_\xi\}&=\int_{\Sigma}-\Tr\left\{ [\mu,\rho]\left( \mathrm{L}_\xi^{A_0}(A-A_0) +\iota_\xi F_{A_0} \right) + \mathrm{L}_\xi^{A_0}\rho d_A\mu \right\}\\
		&=\int_{\Sigma}\Tr\left\{ \mathrm{L}_\xi^{A_0}\mu[A-A_0,\rho] +\mu[A-A_0,\mathrm{L}_\xi^{A_0}(\rho)]-\mu[\rho,\iota_{\xi}F_{A_0}] - d_A\mathrm{L}_\xi^{A_0}(\rho)\mu \right\}\\
		&=\int_{\Sigma}\Tr\left\{ \mathrm{L}_\xi^{A_0}\mu[A,\rho] -d\mu \iota_{\xi}d\rho + [\iota_{\xi}A_0,d\rho]\mu \right\}\\
		&=\int_{\Sigma} \Tr\left\{ \mathrm{L}_\xi^{A_0}(\mu)d_A\rho \right\}=M_{ \mathrm{L}_\xi^{A_0}\mu};
	\end{align*}

	\begin{align*}
		\{M_\mu,H_\lambda\}&=\Tr\int_{\Sigma} [\mu,\rho] \lambda(B,ee_n) + d_A(\lambda e_n eB)d_A\mu\\
		&=\Tr\int_{\Sigma} d(\lambda e_n e B)[A,\mu]+ [A,\lambda e_n e B]d_\mu + [A,\lambda e_n e B][A,\mu] + \frac{\lambda}{2}e^2(ee_n, B)[\mu,B]\\
		&\overset{\bigstar}{=}\Tr\int_{\Sigma}-\lambda e_n e B [dA,\mu] + \lambda e_n e B [A,d\mu - \lambda e_n e B[A,d\mu] + \frac{\lambda e_n}{2}e B [\mu,[A,A]]\\
		&\qquad \frac{\lambda e_n}{3!}e^3(B,[\mu,B])-\frac{\lambda e_n}{2}e(B,e^2)[\mu,B]\\
		&=\Tr\int_{\Sigma} \lambda e_n e B \left([\mu,F_A] + \frac{1}{2}[\mu,(e^2,B)]\right)+\frac{\lambda e_n}{2\cdot 3!}e^3[\mu,(B,B)]=0 ,
	\end{align*}
	where in the last passage we used that $\frac{(B,e^2)}{2}+F_A=0$ and that $\Tr[\mu,(B,B)]=0$
	
	The computation of the YM part of $\{P_\xi,P_\xi\}$ depending only on $\rho$ and $A$ is exactly equivalent to the computation of the free part of $\{P_\xi,P_\xi\}$ (i.e. the one depending only on $e$ and $\omega$), as one can  notice by substituting $\frac{1}{2}e^2\mapsto \rho$ and $\omega_{(0)}\mapsto A_{(0)}$, then we obtain
	
	\begin{align*}
		\{P_\xi,P_\xi\}&=\int_{\Sigma}\frac{1}{4} d_{\omega}(e e ) \iota_{[\xi,\xi]}(\omega - \omega_0)+\frac{1}{4}\iota_{[\xi,\xi]}(e e )F_{\omega}-\frac{1}{4}d_{\omega}(e e )\iota_{\xi}\iota_{\xi} F_{\omega_0}\\
		&\qquad  + \Tr\left\{ \frac{1}{2} d_{A}(\rho ) \iota_{[\xi,\xi]}(A - A_0)+\frac{1}{2}\iota_{[\xi,\xi]}(\rho )F_{A}-\frac{1}{2}d_{A}(\rho )\iota_{\xi}\iota_{\xi} F_{A_0}  \right\}\\
		&=\frac{1}{2}P_{[\xi,\xi]}-\frac{1}{2}L_{\iota_\xi \iota_\xi F_{\omega_{0}}}-\frac{1}{2}M_{\iota_{\xi}\iota_{\xi}F_{A_0}};
	\end{align*}
	\begin{align*}
		\{H_\lambda,H_\lambda\}&=\int_{\Sigma}(\cdots) - \lambda e B (B, e_n e) d_\omega(\lambda e_n) + \lambda(B,e_n e) d_A(\lambda e_n e B)\\
		&=\int_{\Sigma}(\cdots) -\lambda e B (B,e_n e)d\lambda e_n  + \lambda e B (B,e_n e)d\lambda e_n = 0 ;
	\end{align*}
	\begin{align*}
		\{P_\xi,H_\lambda\} & = \int_{\Sigma} (\cdots)+ \Tr \int_{\Sigma}  -\frac{\lambda e_n}{4} e^2 (B,B)\mathrm{L}^{\omega_0}_\xi(e) - \lambda e_n B F_A \mathrm{L}^{\omega_0}_\xi(e) + \lambda e B(B,e_n e) \mathrm{L}^{\omega_0}_\xi(e) \\
		& \qquad -\lambda (B,e_n e) \mathrm{L}^{A_0}_\xi(\rho)+ d_A(\lambda e_n e B)(-\iota_{\xi}F_A + d_A\iota_{\xi}(A-A_0))\\
		& = \int_{\Sigma} (\cdots) + \Tr \int_{\Sigma}-\frac{\lambda e_n}{2 \cdot 3!}(B,B)\mathrm{L}^{\omega_0}_\xi(e^3) - \lambda e_n B F_A\mathrm{L}^{\omega_0}_\xi(e) + \frac{\lambda}{2}B(B,e_n e) \mathrm{L}^{\omega_0}_\xi(e^2)\\
		&\qquad -\frac{\lambda}{2}B(B,e_n e) \mathrm{L}^{\omega_0}_\xi(e^2) - \frac{\lambda}{2}e^2(B,e_n e) \mathrm{L}^{\omega_0 + A_0}_\xi(B) - \lambda e_n e B d_A(-\iota_{\xi}F_A + d_A\iota_{\xi}(A-A_0))\\
		& \overset{\blacktriangle}{=} \int_{\Sigma}(\cdots) + \Tr \int_{\Sigma}  \mathrm{L}^{\omega_0}_\xi(\lambda e_n)\frac{e^3}{2 \cdot 3!}(B,B)+ \frac{\lambda e_n}{2\cdot 3!}e^3\mathrm{L}^{\omega_0 + A_0}_\xi(B, B) +\mathrm{L}^{\omega_0}_\xi (\lambda e_n)eB F_A\\
		& \qquad +\lambda e_n e \mathrm{L}^{\omega_0 + A_0}_\xi(BF_A) - \frac{\lambda}{2}(B,e_ne) e^2 \mathrm{L}^{\omega_0 + A_0}_\xi(B) - \lambda e_n e B\left\{-d_A\iota_\xi F_A+[F_A,\iota_\xi(A-A_0)]\right\} \\
		& \overset{\bigstar}{=} \int_{\Sigma}(\cdots) + \Tr \int_{\Sigma}+ \mathrm{L}^{\omega_0}_\xi(\lambda e_n)  \left(\frac{1}{2\cdot 3!} e^3(B,B) + eBF_A \right) +\frac{\lambda e_n}{3!}e^3(B,\mathrm{L}^{\omega_0 + A_0}_\xi B) \\
		&\qquad+ \lambda e_n e \mathrm{L}^{\omega_0 + A_0}_\xi(B) F_A + \lambda e_n e B \mathrm{L}^{ A_0}_\xi F_A - \frac{\lambda e_n}{3!}e^3 (B,\mathrm{L}^{\omega_0 + A_0}_\xi B) + \frac{\lambda e_n}{2}e (e^2,B) \mathrm{L}^{\omega_0 + A_0}_\xi B\\
		&\qquad - \lambda e_n e B\left\{-d_{A_0}\iota_\xi F_A+\iota_{\xi}[A-A_0,F_A]\right\} \\
		& = \int_{\Sigma}(\cdots) + \Tr \int_{\Sigma}+ \mathrm{L}^{\omega_0}_\xi(\lambda e_n)  \left(\frac{1}{2\cdot 3!} e^3(B,B) + eBF_A \right) + \lambda e_n e B \mathrm{L}^{A_0}_\xi(F_A) \\
		&\qquad - \lambda e_n e B\left( \iota_{\xi}d_{A_0}F_A -d_{A_0}\iota_\xi F_A\right)\\
		& =\int_{\Sigma}  \mathrm{L}^{\omega_0}_\xi(\lambda e_n) \left( eF_\omega + \frac{\Lambda}{3!}e^3 + \Tr\left[\frac{1}{2\cdot 3!}e^3(B,B) + eBF_A\right] \right)\\
		& = P_{\mathrm{L}^{\omega_0}_\xi(\lambda e_n)^{(a)} } + H_{\mathrm{L}^{\omega_0}_\xi(\lambda e_n)^{(n)} } -L_{\mathrm{L}^{\omega_0}_\xi(\lambda e_n)^{(a)}(\omega-\omega_0)_{(a)}}-M_{\mathrm{L}^{\omega_0}_\xi(\lambda e_n)^{(a)}(\omega-\omega_0)_{(a)}},
	\end{align*}
	where we also used the Bianchi identities
	\begin{equation}\tag{$\blacktriangle$}
		d^2_A\alpha=[F_A,\alpha]\qquad d_AF_A=0.
	\end{equation}

	\begin{align*}
		\{L_c,H_\lambda\} & = \int_{\Sigma}(\cdots) + \Tr \int_{\Sigma}\lambda e_n \left( \frac{1}{4} e^2(B,B)[c,e] + BF_A [c,e] \right) +\lambda e B(B,e_n e)[c,e] \\
		& =\int_{\Sigma} (\cdots) + \Tr \int_{\Sigma} -[c,\lambda e_n]\left( \frac{1}{2\cdot 3!} e^3 (B,B) + eBF_A \right) -\lambda e_n e [c,B] F_A\\
		&\qquad +\frac{\lambda}{2}e^2(B,e_ne)[c,B]\\
		&  \overset{\bigstar}{=}\int_{\Sigma} (\cdots)  + \Tr \int_{\Sigma} -[c,\lambda e_n]\left( \frac{1}{2\cdot 3!} e^3 (B,B) + eBF_A \right) -\lambda e_ne F_A[c,B]\\
		&\qquad + \frac{\lambda e_n}{2\cdot3!}e^3[c,(B,B)] - \frac{\lambda e_n}{2}e(e^2,B)[c,B]\\
		& = \int_{\Sigma}  -[c,\lambda e_n]\left( eF_\omega +\frac{\Lambda}{3!}e^3  + \Tr\left\{\frac{1}{2\cdot 3!} e^3 (B,B) + eBF_A\right\} \right) \\
		& = - P_{[c,\lambda e_n]^{(a)}} + L_{[c,\lambda e_n]^{(a)}(\omega - \omega_0)_a} - H_{\lambda e_n^{(n)}} + M_{[c,\lambda e_n]^{(a)}(A-A_0)_{(a)}}.
	\end{align*}

\end{proof}

\subsection{The BFV Formalism in the YMPC Theory}\label{BFV YM}
	As we did for the case of the scalar field, we replicate the discussion about the BFV formalism applied to the space of boundary fields, which is now promoted to a graded symplectic manifold by considering the Lagrange multipliers as ghost fields and adding ghost momenta. We express the BFV quantities in the following theorem starting from the quantities of gravity alone described in Theorem \ref{thm:BFVgravity}.

	\begin{theorem}\label{thm:BFVactionYM}
		Let $\mathcal{F}^{\mathrm{YM}}$ be the bundle
				\begin{equation*}
\mathcal{F}_{YM} \longrightarrow \Omega_{nd}^1(\Sigma, \mathcal{V}),
\end{equation*}
with local trivialisation on an open $\mathcal{U}_{\Sigma} \subset \Omega_{nd}^1(\Sigma, \mathcal{V})$
		\begin{equation}
			\mathcal{F}_{\mathrm{YM}}\simeq \mathcal{T}_{PC} \times\mathcal{A}^{\text{YM}}_\partial\times\Omega^{(0,2)}_{\partial,\text{red}} \oplus T^* \left(\Gamma[1](\mathfrak{g})\right),
		\end{equation}
		where  where $\mathcal{T}_{PC}$ was defined in \eqref{LoctrivF1} and the additional fields in degree zero are denoted  by  $A\in\mathcal{A}^{\text{YM}}_\partial$ and $B\in \Omega^{(0,2)}_\partial$ and they  satisfy the structural constraint $1/2(e^2,B)+F_A=0$. The additional ghost field is denoted by  $\mu\in \Gamma[1](\mathfrak{g})$ and its antifield by $\mu^\dagger \in \Gamma[-1](\wedge^3 T^*\Sigma \otimes \wedge^4 \mathcal{V}\otimes \mathfrak{g})$.
		
		We define an action functional and a symplectic form on $\mathcal{F}_{\mathrm{YM}}$ by
			\begin{align}	
				S_{\text{YM}}= S_{PC} + \int_{\Sigma} &  \Tr\{\iota_{\xi}(A-A_0)d_A\rho\}  +\Tr(\iota_\xi\rho F_A) +\lambda e_n \left(e \Tr(B 	F_A) + \frac{1}{2\cdot 3!} e^3 \Tr(B,B)\right)\nonumber\\ &+ \Tr(\mu d_A\rho) + \Tr\left\{ \frac{1}{2} [\mu,\mu]\mu^\dagger - 	\mathrm{L}_\xi^{A_0}(\mu)\mu^\dagger + \frac{1}{2}\iota_\xi\iota_\xi F_{A_0}\mu^\dagger\right\} \nonumber\\ &+ \Tr\left\{\left[ \mathrm{L}_\xi^{\omega_0}(\lambda e_n) ^{(a)}-[c,\lambda e_n]^{(a)}\right](A-A_0)_{a}\mu^\dagger \right\}\label{action_NCYM1}.\\
				\varpi_\mathrm{YM} =\varpi_{PC} + \int_{\Sigma} &\Tr(\delta \rho \delta A)  + \Tr(\delta \mu \delta \mu^\dagger) \label{symplectic_form_YM} 
			\end{align}
		Then the triple $(\mathcal{F}_\mathrm{YM}, \varpi_\mathrm{YM}, S_\mathrm{YM})$ defines a BFV structure on $\Sigma$.
	\end{theorem}

	\begin{proof}
		We need to prove that $\{S_\mathrm{YM},S_\mathrm{YM}\}=0$. 
		We split the symplectic form into the classical part and the ghost part	
			\begin{eqnarray}
				&\varpi_{\text{YM},f}=\int_{\Sigma} e\delta e \delta \omega + \Tr(\delta \rho \delta A);\\
				&\varpi_{\text{YM},g}=\int_{\Sigma}\delta c \delta c^\dagger + \delta \lambda \delta \lambda^\dagger + \iota_{\delta \xi}\delta \xi^\dagger + \Tr(\delta \mu \delta \mu^\dagger).
			\end{eqnarray}
		Furthermore it is useful to employ the already known results and split $S_\mathrm{YM}=S_0^{\mathrm{YM}}+S_1^{\mathrm{YM}}$, with $S_0^{\mathrm{YM}}=S_0^0+S_0^1$ and $S_1^{\mathrm{YM}}=S_1^0+S_1^1$ defined such that
			\begin{align}
				S_0^0= \int_{\Sigma} & c e d_{\omega} e + \frac{1}{2}\iota_{\xi}e^2 F_\omega  + \iota_\xi(\omega-\omega_0)ed_\omega e + \lambda e_n \left( eF_\omega + \frac{\Lambda}{3!}e^3 \right);\\
				S_0^1=\Tr\int_{\Sigma} & \iota_\xi\rho F_A + \iota_\xi(A-A_0)d_A\rho + \lambda e_n\left( eBF_A + \frac{e^3}{2\cdot 3!}e^3(B,B) \right) + \mu d_A\rho;\\
				S_1^1=\int_{\Sigma} & \frac{1}{2} [c,c] c^{\dag} - \mathrm{L}_{\xi}^{\omega_0} c c^{\dag} +\frac{1}{2}\iota_{\xi}\iota_{\xi}F_{\omega_0}c^{\dag}  + [c, \lambda e_n ]^{(a)}(\xi_a^{\dag}- (\omega - \omega_0)_a c^\dag) \nonumber\\
				&+ [c, \lambda e_n ]^{(n)}\lambda^\dag  - \mathrm{L}_{\xi}^{\omega_0} (\lambda e_n)^{(a)}(\xi_a^{\dag}- (\omega - \omega_0)_a c^\dag) - \mathrm{L}_{\xi}^{\omega_0} (\lambda e_n)^{(n)}\lambda^\dag \nonumber \\
				& - \frac{1}{2}\iota_{[\xi,\xi]}\xi^{\dag};\\
				S_1^1=\Tr\int_{\Sigma} & \frac{1}{2} [\mu,\mu]\mu^\dagger - \mathrm{L}_\xi^{A_0}(\mu)\mu^\dagger + \frac{1}{2}\iota_\xi\iota_\xi F_{A_0}\mu^\dagger +\mathrm{L}_\xi^{\omega_0}(\lambda e_n) ^{(a)} (A-A_0)_{a}\mu^\dagger \nonumber \\
				& -[c,\lambda e_n]^{(a)}(A-A_0)_{a}\mu^\dagger.
			\end{align}
		The cohomological vector field $Q$ splits into $Q=Q_0^0+Q_0^1+Q_1^0+Q_1^1$, such that $\iota_{Q^i_j}\varpi_{\text{YM}}=\delta S^i_j$.	
		
		The classical master equation reads 
		\begin{align*}
			\{S,S\}=\{S_0,S_0\}_f+2\{S_0,S_1\}_f+2\{S_0,S_1\}_g+\{S_1,S_1\}_f+\{S_1,S_1\}_g.
		\end{align*}
		Of course we have $\{S_0,S_0\}_f+2\{S_0,S_1\}_g=0$ by ``definition'' and $\{S_0,S_0\}_g=0$ since $S_0$ has no antighost part. Again we should prove separately that $2\{S_0,S_1\}_f+\{S_1,S_1\}_g=0$ and $\{S_1,S_1\}_f=0$. This means 
			\begin{align}
				&\{S_0^1,S_1^0\}_f + \{S_0^0,S_1^1\}_f + \{S_0^1,S_1^1\}_f + \{S_1^0,S_1^1\}_g + \frac{1}{2}\{S_1^1,S_1^1\}_g=0 \label{eq 1};\\
				&\{S_1^0,S_1^1\}_f + \frac{1}{2}\{S_1^1,S_1^1\}_f=0 \label{eq 2}.
			\end{align}
		
		We compute them explicitly. In order to do so, we first need to find $Q_1^1$
			\begin{align*}
				\delta S_1^1=&\Tr\int_{\Sigma} \iota_{\delta\xi}\iota_\xi F_{A_0}\mu^\dag + \frac{1}{2}\iota_\xi\iota_\xi F_{A_0}\delta \mu\dag - \delta \mu [\mu,\mu^\dag] - \iota_{\delta\xi}d_{A_0}\mu \mu^\dag - \delta\mu \mathrm{L}_\xi^{A_0}(\mu^\dag)\\
				&\qquad -\mathrm{L}_\xi^{A_0}(\mu)\delta \mu^\dag + \{ (\iota_{\delta \xi}d_{\omega_0}(\lambda e_n))^{(a)} - \mathrm{L}_\xi^{\omega_0}(\delta \lambda e_n)^{(a)} - \mathrm{L}_\xi^{\omega_0}(\lambda e_n)^{(b)}\delta e_b^{(a)}\\
				&\qquad -[\delta c, \lambda e_n]^{(a)}+[c,\delta \lambda e_n]^{(a)} + [c,\lambda e_n]^{(b)}\delta e_b^{(a)}  \}(A-A_0)_a \mu^\dag +\\
				&\qquad (\mathrm{L}_\xi^{\omega_0}(\lambda e_n)^{(a)}-[c,\lambda e_n]^{(a)})\delta A_a \mu^\dag + (\mathrm{L}_\xi^{\omega_0}(\lambda e_n)^{(a)}-[c,\lambda e_n]^{(a)})(A-A_0)_a \delta \mu^\dag.
			\end{align*}
		From this variation we find that $Q_{1A}^1, Q_{1e}^1, Q_{1\lambda}^1, Q_{1c}^1, Q_{1\xi}^1$ vanish. In particular, we are also able to explicitly compute $Q_{1\mu}^1$ and $Q_{1\mu^\dag}^1$ 
			\begin{align*}
				&Q_{1\mu}^1=\frac{1}{2}\iota_\xi\iota_\xi F_{A_0}\mu^\dag + \frac{1}{2}[\mu,\mu] -\mathrm{L}_\xi^{A_0}(\mu) + (\mathrm{L}_\xi^{\omega_0}(\lambda e_n)^{(a)}-[c,\lambda e_n]^{(a)})(A-A_0)_a\\
				&Q_{1\mu^\dag}^1=-[\mu,\mu^\dag] - \mathrm{L}_\xi^{A_0}(\mu^\dag).
			\end{align*}
		
		The components of $Q_0^0$ and $Q_0^1$ are recovered from the Hamiltonian vector fields in the previous sections, while $Q_1^0$ is the same as in \cite{CCS2020}.
		
		We now prove \eqref{eq 2} and we leave the other identity for the appendix. First, we notice that $\{S_1^1,S_1^1\}_f=0$ because $Q_{1A}^1=0$ and $Q_{1e}^1=0$. Furthermore
		
			\begin{align*}
				\{S_1^0,S_1^1\}_f &=	\iota_{Q_1^0}\iota_{Q_1^1}\int_{\Sigma} e\delta e \delta \omega + \Tr(\delta \rho \delta A)\\
				&=\iota_{Q_1^0}\int_{\Sigma} ([c,\lambda e_n]^{(b)}-\mathrm{L}_\xi^{\omega_0}(\lambda e_n)^{(b)})\delta_b^{(a)}(A-A_0)_a \mu^\dag\\
				&=\int_{\Sigma} \underbrace{([c,\lambda e_n]^{(b)}-\mathrm{L}_\xi^{\omega_0}(\lambda e_n)^{(b)})}_{\propto \lambda} \underbrace{(Q_{1e}^0)^{(a)}_b}_{\propto \lambda} (A-A_0)_a\mu^\dag \propto \int_{\Sigma} \lambda^2=0.
			\end{align*}	
	\end{proof}

\begin{remark}
    As in the case of the scalar field the BFV structure of Theorem \ref{thm:BFVactionYM} depends on reference connections $\omega_0$ and $A_0$. In this case the change of variables that brings to a BFV theory not depending on them, is slightly different, having to account also for $A_0$:
    \begin{align*}
        c' = c + \iota_\xi (\omega-\omega_0) \qquad \xi^{'\dag}_a = \xi^{\dag}_a - (\omega - \omega_0)_a c^\dag - \operatorname{Tr}\left[(A - A_0)_a \mu^\dag\right] \qquad \mu'= \mu + \iota_\xi (A-A_0).
    \end{align*}
\end{remark}

\section{Spinor field coupled to gravity}\label{s:spinor}

We now want to describe the interaction of gravity with fermionic spin $1/2$ matter, i.e.\ with those particles that obey the Fermi--Dirac statistics: \text{fermions}. The standard discussion about fermions in Quantum Field Theory is developed on a flat $4$-dimensional space--time with a Minkowskian signature by means of an algebraic construction involving Clifford algebras (see appendix \ref{clif spin} for the definitions and properties). In Minkowski space--time, fermions are described by spinors, which are sections of a vector bundle with fibers carrying a linear representation of the Clifford algebra and therefore with an induced action of the universal covering of the group of rotations. In particular, a rotation of $2\pi$ will not act as the identity, but a rotation of $4\pi$ will. This property is expressed mathematically by asking that they transform under the spin 1/2 representation of the double cover of the Lorentz group: the spin group.

The traditional approach in the construction of spinor fields on a curved background involves the definition of \text{spin structures}. A spin structure is defined as a principal bundle morphisms $\Lambda\colon\overline{\Sigma}\rightarrow \mathrm{SO}(M,g)$, where $\overline{\Sigma}$ is a principal fiber bundle having Spin$(N-1,1)$ as its structure group and $\mathrm{SO}(M,g)$ is the space of orthonormal frames on the $N$-dimensional pseudoriemannian manifold $M$ with respect to a metric $g$ of signature $(N-1,1)$.

$\overline{\Sigma}$ is usually called \text{spin bundle} and such a structure exists  only if $M$ meets certain topological requirements (see \cite{spingeom}). If this is the case, one takes a (complex) $N$-dimensional vector space $V$ and defines the bundle $E_\lambda:=\overline{\Sigma}\times_\lambda V$ associated to the spin bundle via a representation $\lambda$ (with half-integer spin). $E_\lambda$ is called a  \text{spinor bundle} and spinor fields are defined to be sections of it.

Spin structures are useful since they allow to overcome technical difficulties in the definition of the Dirac equation on the manifold $M$, in the sense that they resolve possible glueing issues, and are in general used when describing spinors on a fixed (curved) background.

However, the dependence of the spin structure on a reference metric does not allow for a coherent description in which gravity interacts with the matter field as a dynamical field. Furthermore, once a metric has been fixed, there might be more inequivalent choices of spin bundles on $M$.

Therefore we need to move our attention to a more general construction allowing to consider the metric as a a dynamical field while preserving the possibility of introducing spinors. This is done in terms of \text{spin frames}.

This section is largely based on \cite{fatibene2018}, \cite{FatiNoris22}, \cite{Fatibene1996}, \cite{Romero2021} and \cite{spingeom}.

\subsection{Spin frames and spinor fields}

As usual, before moving to the $4$-dimensional case, we will be looking at the general construction on an $N$-dimensional pseudoriemannian manifold $M$. 
As we explain in Appendix \ref{clif spin}, there exists a group homomorphism $l\colon\mathrm{Spin}(N-1,1)\rightarrow \mathrm{SO}(N-1,1)$ which is a double covering. The spin group is defined within a Clifford algebra $\mathcal{C}(N-1,1)$ whose basis is given in terms of \text{gamma matrices} (in the gamma representation) that satisfy
\begin{equation}\label{e:anticommrelation_gamma}
	\{\gamma_a,\gamma_b\}=-2\eta_{ab}\mathbb{1}.
\end{equation}
We can also define $\gamma_a^\dagger$ (the adjoint gamma matrix) by
\begin{equation}
	\gamma_0\gamma_a^\dagger \gamma_0=\gamma_a.
\end{equation}
The covering map $l \colon \text{Spin}(N-1,1)\rightarrow \text{SO}(N-1,1)$, $S\mapsto l(S)$ is defined via
\begin{equation}
	S\gamma_aS^{-1}=\gamma_bl^b_a(S).
\end{equation}
Now let us consider a principal fiber bundle $\hat{P}$ whose structure group is $\mathrm{Spin}(N-1,1)$. We find that $\hat{P}$ is a double covering of a principal orthogonal bundle $P$ such that the following diagram commutes
\begin{equation}
	\begin{tikzcd}
		\hat{P} \arrow[d, "\hat{p}"] \arrow[r, "\hat{l}"] & P \arrow[d, "p"] \\
		M  \arrow[r,"\mathrm{id}"]           & M                      
	\end{tikzcd}
\end{equation}
where $\hat{l}\colon \hat{P}\rightarrow P:[x,S]\mapsto[x,l(S)]$ (one can prove that it is global and independent of the trivialization). 

In analogy with the vielbein map, we define a \text{spin frame} to be an equivariant principal morphism $\hat{e}\colon\hat{P}\rightarrow LM$, namely such that the following diagram commutes

\begin{align}
\tag{equivariance}
	&R_{i\circ l (S)}\circ \hat{e}=\hat{e}\circ R_S	\quad &&\begin{tikzcd}[ampersand replacement=\&]
		\hat{P} \arrow[d, "R_S"] \arrow[r, "\hat{e}"] \& LM \arrow[d, "R_{i\circ l (S)}"] \\
		\hat{P} \arrow[r, "\hat{e}"]                  \& LM                      
	\end{tikzcd}
\end{align}
where $LM$ is the frame bundle (a principal-$GL(N, \mathbb{R})$ bundle).

As in the case of the vielbein, thanks to equivariance, we can uniquely determine a \text{spin frame} $\hat{e}$ once we know it on a local section.

As usual, a family of local sections $\hat{\sigma}_{(\alpha)}\colon U_{(\alpha)}\rightarrow \hat{P}$ induces a local trivialization on $P$. For any spin frame $\hat{e}$, this defines a local moving frame $\hat{e}(\hat{\sigma}_{(\alpha)})=(x,e_a^{(\alpha)})$, where $e_a^{(\alpha)}=(e^{(\alpha)})^\mu_a\partial_{\mu}$. On the overlap of two local trivializations the moving frames change by an orthogonal transformation defined by 
\begin{equation}
	e_{(\beta)}=e_{(\alpha)}\cdot S^{(\alpha\beta)} \hspace{2mm} \Rightarrow \hspace{2mm} e^{(\beta)}_a=e^{(\alpha)}_bl^{b}_a(S^{(\alpha\beta)}).
\end{equation}

Also in this case we can define \text{spin coframes} as duals of spin frames, then for each frame we obtain a unique metric $g$ induced by	
\begin{equation}
	g_{\mu\nu}=e^a_\mu \eta_{ab} e^{b}_\nu.
\end{equation}
\begin{remark}
	The image $\hat{e}(\hat{P})\subset LM$ coincides with the orthonormal frames defined by means of the induced metric $g$, namely $\hat{e}(\hat{P})=\mathrm{SO}(M,g)$. 
\end{remark}

The trivialization on $\hat{P}$ induces a trivialization on $P$  by post-composition with $\hat{l}\colon \hat{P}\rightarrow P$. For each family of local sections $\hat{\sigma}_{(\alpha)}$ we obtain $\sigma_{(\alpha)}:=\hat{l}\circ\hat{\sigma}_{(\alpha)}\colon U_{(\alpha)}\rightarrow P$, which is equivalent to having the following diagram commute
\begin{equation}\label{spin fram factorization}
	\begin{tikzcd}
		\hat{P} \arrow[rdd, "\hat{p}", bend right] \arrow[rr, "\hat{e}"] \arrow[rd, "\hat{l}"] &                                  & LM \arrow[ldd, "\pi"', bend left] \\
		& P \arrow[d, "p"] \arrow[ru, "e"] &                                   \\
		& M                                &                                  
	\end{tikzcd}
\end{equation}

\begin{remark}
    Notice that we do not need a metric to define spin frames, indeed one is induced by spin coframes. However, when dealing with spin geometry, one usually considers \text{spin structures}, which are defined in terms of a (pseudo--Riemannian) metric $g$ on $M$. In particular, a spin structure is an equivariant morphism $\Lambda \colon \hat{P}\rightarrow SO(M,g)$, where $\hat{P}$ is a spin bundle. An important result (chap. 2 \cite{spingeom}) states that a spin structure exists if and only if the second Stiefel--Whitney class of $M$ vanishes. 
    Then the following question arises naturally: when do spin frames exist? An answer is given by the following result \cite{FatiNoris22}:
        A spin frame $\hat{e}$ on $M$ exists if and only if there exists a spin structure $\Lambda\colon \hat{P}\rightarrow SO(M,g)$ for a suitable metric $g$ on $M$.
\end{remark}

Now we can as usual define the \text{Minkowski bundle} $\hat{\mathcal{V}}:=\hat{P}\times_{\hat{\rho}}V$, where $V$ is an $N$-dimensional (real) vector space and $\hat{\rho}:=\rho\circ l$ is the vector (i.e.\ spin 1) representation of Spin$(N-1,1)$ on $V$ corresponding to the fundamental representation of SO$(N-1,1)$. 

A \text{spin coframe} can then be seen as an isomorphism $TM\rightarrow \hat{\mathcal{V}}$ which produces the same dynamics of the vielbein. Indeed diagram \eqref{spin fram factorization} exactly tells us that the dynamics of the spin frame factorizes through the dynamics of the vielbein. This is also true for any matter field coupled to spin frames transforming under a \text{tensor representation} $\hat{\lambda}$ (i.e. with integer spin) of the spin group, since in this case we also have the factorization 
\begin{equation}
	\hat{\lambda}(S)=\lambda(l(S)),
\end{equation}
where $\lambda$ is the corresponding representation of SO$(N-1,1)$.

This is not the case for spinors and that is precisely why we needed to introduce spin frames (indeed spinors are defined to be those matter fields which couple to spin frames ``non-tensorially'')

\begin{definition}[Spinor bundle and spinor fields]
	Let $W$ be an $N$-dimensional complex vector space and let $\lambda\colon\mathrm{Spin}(N-1,1)\times W\rightarrow W$ be a non-tensorial representation of the spin group on $W$. 
	The \text{spinor bundle} $E_\lambda$ is defined to be the associated bundle to $\hat{P}$ 
	\begin{equation}
		E_\lambda:=\hat{P}\times_{\lambda} W.
	\end{equation}	
	
	When considering $2m$--dimensional manifolds, thanks to the gamma representation, we can define the \text{bundle of Dirac spinors} as 
	    \begin{equation}
	        S:=\hat{P}\times_{\gamma} \mathbb{C^{2^m}}.
	    \end{equation}
    Sections of $S$ are called \text{Dirac spinors}, indicated as $\psi\in S(M):=\Gamma(M,S)$.

    \begin{remark}
        In our case, given that Dirac spinors obey the so-called Fermi-Dirac statistics, it is appropriate to consider $\psi \in \Gamma(M , \Pi S)$, where $\Pi $ indicates the parity-reversal operation. In other words, the components of Dirac spinors are now defined to take values in Grassmann numbers, hence the parity of $\psi$ is set to be 1.
    \end{remark}
\end{definition}

\subsection{Coupling of the spinor field and the Dirac Lagrangian}

In the coupling of the spinor field, we start by considering an orthogonal principal connection $\omega$ on the principal bundle $P$ with structure group SO$(N-1,1)$. Since the map $\hat{l}\colon \hat{P}\rightarrow P$ is a local diffeomorphism, we can pull back connections from $P$ onto $\hat{P}$. Locally we have
\begin{equation}
	\omega=\omega^{ab}_\mu v_a\wedge v_b dx^\mu.
\end{equation}
When we pull it back to $\hat{P}$, in the gamma representation, we obtain
\begin{equation}
	\hat{\omega}=-\frac{1}{4}\omega^{ab}_\mu \gamma_a \gamma _b dx^\mu.
\end{equation}
Indeed this is a $\mathfrak{spin}(N-1,1)$-valued 1-form, since $\mathfrak{spin}(N-1,1)=\mathfrak{so}(N-1,1)$ and since $-\frac{1}{4}[\gamma_a,\gamma_b]$ provides a basis for it.

At this point it is easy to define the covariant derivative of a Dirac spinor field $\psi$:
\begin{equation}\label{e:cov_der_spinor}
	d_\omega \psi:=d\psi + [\omega,\psi]=d\psi - \frac{1}{4}\omega^{ab}\gamma_a\gamma_b \psi.
\end{equation}
We briefly check that it transforms well under a gauge transformation $\psi\mapsto\psi'=S(x)\psi$, where for each $x$, $S(x)\in\mathrm{Spin}(N-1,1)$

\begin{align*}
	d_{\omega'}\psi'&=d_{\omega'}(S\psi)=(d_{\omega'}S)\psi + Sd_{\omega'}\psi\\
	&=(dS)\psi + [\omega',S]\psi + Sd\psi + S[\omega',\psi]\\
	&=(dS)\psi + \omega'(S\psi)- S \omega'(\psi) + Sd\psi  +S\omega'(\psi)\\
	&=(dS)\psi + S\omega(\psi) -(dS\psi) +Sd\psi=S\{ d\psi +\omega(\psi)\}\\
	&=S\{d\psi + [\omega,\psi]\}=Sd_\omega\psi=(d_\omega \psi)',
\end{align*}
where we used $\omega'=S\omega S^{-1} - (dS)S^{-1}$.

We now need to construct a invariant Lagrangian. In order to do so, we proceed in the standard way and introduce the \text{hermitian conjugate} $\overline{\psi}$ of the field $\psi$.\footnote{In this case we consider the parity of $\psi$ to be 1, meaning that it anticommutes with other odd quantities.} To define it properly, we consider the hermitian conjugate $\overline{W}$ of the complex vector space $W$. Of course the representation $\lambda$ of Spin$(N-1,1)$ on $W$ will induce a representation $\overline{\lambda}$ on $\overline{W}$. We can then define the \text{adjoint spinor bundle} to be $\overline{E}_\lambda:=\hat{P}\times_{\overline{\lambda}}\overline{W}$.
Hence we take $\overline{\psi}\in\Gamma(\overline{E}_\lambda)=:\overline{S}(M)$.

The relation between $\psi$ and its hermitian conjugate in our setting reads $\overline{\psi}:=\psi^{\dagger}\gamma_0$. As we will see later, this relation gives the right equations of motion.

We denote the canonical (hermitian) pairing between sections of $E_\lambda$ and $\overline{E}_\lambda$ by
\begin{equation}
	\overline{\psi}\psi:=<\overline{\psi},\psi>=\overline{\psi}_A \psi^A,
\end{equation}	
where $A=1,\cdots,N$ are the spinor indices.

We define the covariant derivative of the hermitian conjugate of $\psi$ such that $\overline{d_\omega \psi}=d_\omega \overline{\psi}$, hence obtaining
\begin{equation}
	d_\omega \overline{\psi}=d\overline{\psi}+ [\omega,\overline{\psi}]= d\overline{\psi} - \frac{1}{4}\omega^{ab}\overline{\psi}\gamma_a \gamma_b.
\end{equation}

The definition of covariant derivative extends also to the gamma matrices and we get the following result
\begin{lemma}
    Let $\gamma:=\gamma^a v_a\in V\otimes \mathcal{C}(N-1,1)$ be an element of the vector space $V$ with values in the Clifford algebra (seen as endomorphisms of the spin bundle $E_\lambda$). Then
    $$d_\omega \gamma=0.$$
\end{lemma}
\begin{proof}
    $\gamma$ is a section of $\mathcal{V}$ and an endomorphism of the spin bundle $E_{\lambda}$. Hence its covariant derivative reads $$(d_{\omega}\gamma)^b=(d\gamma)^b +  \omega^{bc}\gamma_c - \frac{1}{4}\omega^{ac}(\gamma_a\gamma_c \gamma^b - \gamma^b \gamma_a \gamma_c) .$$
    Note that this formula implies the correct Leibniz rule for $d_{\omega}(\gamma\psi)$.
    Using the anti-commutation relation \eqref{e:anticommrelation_gamma} we can show that $\omega^{bc}\gamma_c - \frac{1}{4}\omega^{ac}\eta^{bd} (\gamma_a\gamma_c \gamma_d - \gamma_d \gamma_a \gamma_c) = 0$ and conclude the proof by choosing $\gamma$ constant. 
\end{proof}

At this point the action functional containing the spinor field is written as
\begin{equation}
	S_{\text{Dirac}}:=\int_{M} i\frac{e^{(N-1)}}{2(N-1)!}\left[ \overline{\psi} \gamma d_\omega \psi - d_\omega \overline{ \psi}\gamma \psi \right],
\end{equation}
Alternatively, we can write the action as
\begin{equation}
	S_{\text{Dirac}}=\int_{M}\frac{e^N}{2N!}\left\{ i\overline{\psi}\gamma^a\nabla_a\psi - i \nabla_a\overline{\psi}\gamma^a \psi  \right\},
\end{equation}
with $\nabla_a\psi:=e^\mu_a \left(\partial_\mu \psi -\frac{1}{4}\omega^{ab}_\mu \gamma_a\gamma_b\psi\right)$.

	\subsubsection{Equations of motion}
	
	We now consider the full action $S:=S_{\text{PC}}+S_{\text{Dirac}}$ and take its variation:
		\begin{equation}
			\begin{split}
				\delta S=& \delta_\omega S + \int_M \left[ 	\frac{e^{N-3}}{(N-3)!}F_\omega + i\frac{e^{N-2}}{2(N-2)!}\left( \overline{\psi}\gamma d_\omega\psi - d_\omega \overline{\psi}\gamma \psi\right)\right]\delta e\\
				&\qquad\qquad + i \delta \overline{\psi}\left[ 	\frac{e^{N-1}}{(N-1)!}\gamma d_\omega \psi - d_\omega\left( \frac{e^{N-1}}{2(N-1)!} \right)\gamma \psi  \right]\\
				&\qquad \qquad +i\left[\frac{e^{N-1}}{(N-1)!}d_\omega 	\overline{\psi}\gamma +  d_\omega\left( \frac{e^{N-1}}{2(N-1)!}\right)\overline{\psi}\gamma   \right]\delta \psi,\\
			\end{split}
		\end{equation}
	where we used that $d_\omega \gamma=0$.

	To compute $\delta_\omega S$, first we define the internal contraction on $\mathcal{V}$. In particular, for any $X\in V$ and for all $\alpha\in \wedge^k V$, we define for all $\alpha=\frac{1}{k!}\alpha^{i_1\cdots i_k}v_{i_1}\wedge\cdots\wedge v_{i_k}$
		\begin{equation}
			j_X \alpha:=\frac{\eta_{ab}}{(k-1)!} X^a \alpha^{b {i_2}\cdots {i_k}} 	v_{i_2}\wedge \cdots \wedge v_{i_k}.
		\end{equation}

    With this definition, we obtain 
        \begin{equation}
            [\alpha,\psi]=\frac{1}{4}j_\gamma j_\gamma \alpha \psi, \qquad \text{and}  \qquad [\alpha,\overline{\psi}]=-\frac{(-1)^{|\alpha||\psi|}}{4}\overline{\psi}j_{\gamma} j_{\gamma} \alpha.
        \end{equation}

	We know $\delta_\omega S=\delta_\omega S_{PC} + \delta_\omega S_{\text{Dirac}}$, with
		\begin{equation}
			\delta_\omega S_{PC}=\int_M \frac{e^{N-3}}{(N-3)!}d_\omega e 	\delta \omega + \int_{\partial M} \frac{e^{N-2}}{(N-2)!}\delta \omega,
		\end{equation}
	while
		\begin{equation}
			\begin{split}
				\delta_\omega S_{\text{Dirac}}&= \int_M \frac{i e^{N-1} }{ 2\cdot (N-1)!} \left[ \bar{\psi}\gamma [\delta \omega. \psi] - [\delta \omega. \bar{\psi}] \gamma \psi \right] \\
                &= \int_M \frac{i e^{N-1} }{ 8\cdot (N-1)!} \bar{\psi} \left[ \gamma j_\gamma j_\gamma \delta \omega + j_\gamma j_\gamma \delta \omega \gamma \right] \psi\\
                &=\int_{M} 	\frac{i }{ 8\cdot (N-1)!} \bar{\psi} \left[ \gamma j_\gamma j_\gamma e^{N-1}  + j_\gamma j_\gamma e^{N-1} \gamma \right] \psi \delta \omega ,
			\end{split}
		\end{equation}
	where in the last step we used the following result
        \begin{equation}\label{jgamma omega e}
            e^{N-1} \left[ \gamma j_\gamma j_\gamma \delta \omega + j_\gamma j_\gamma \delta \omega \gamma \right] = \delta \omega \left[ \gamma j_\gamma j_\gamma e^{N-1}  + j_\gamma j_\gamma e^{N-1} \gamma \right],
        \end{equation}

Now, before looking at the equations of motion, we notice that we obtain a boundary term after imposing Stokes' theorem:
\begin{equation}
	\tilde{\alpha}=\int_{\partial M}\frac{e^{N-2}}{(N-2)!} \delta \omega + i\frac{e^{N-1}}{2(N-1)!}\left(  \overline{\psi}\gamma\delta\psi - \delta \overline{\psi}\gamma \psi \right).
\end{equation}
The equations of motion become
\begin{align}
	&\frac{e^{N-3}}{(N-3)!}F_\omega + i\frac{e^{N-2}}{2(N-2)!}\left( \overline{\psi}\gamma d_\omega\psi - d_\omega \overline{\psi}\gamma \psi\right)=0 \label{einstein spinor},\\
	&\frac{e^{N-3}}{(N-3)!}d_\omega e +  \frac{i}{8(N-1)!}\overline{\psi}\left( j_{\gamma} j_{\gamma} e^{N-1} \gamma + \gamma j_{\gamma} j_{\gamma} e^{N-1} \right)\psi=0 \label{torsion spinor 0},\\
	&\frac{e^{N-1}}{(N-1)!}\gamma d_\omega \psi - d_\omega\left( \frac{e^{N-1}}{2(N-1)!} \right)\gamma \psi   =0 \label{Dirac spinor 00},\\
	&\frac{e^{N-1}}{(N-1)!}d_\omega \overline{\psi}\gamma +  d_\omega\left( \frac{e^{N-1}}{2(N-1)!}\right)\overline{\psi}\gamma  =0 \label{Dirac spinor 10}.
\end{align}
Notice that, once we impose $\overline{\psi}=\psi^\dag \gamma_0$, then equations \eqref{Dirac spinor 00} and \eqref{Dirac spinor 10} are one the Hermitian conjugate of the other, representing Dirac equation on a curved background.

\subsection{Boundary structure in \texorpdfstring{$N=4$}{Lg}}

We now look at the boundary structure of the fields in the theory.
As usual, we restrict the space of fields to the boundary, obtaining $\tilde{F}^{\text{s}}_\partial=\mathcal{A}_\Sigma\times\Omega_{\partial\text{,n.d.}}^{(1,1)}\times S(\Sigma)\times \overline{S}(\Sigma)$, where $S(\Sigma):=\Gamma(\Sigma,E_{\lambda}\vert_\Sigma)$. 

The presymplectic form on the space of preboundary fields is given as usual by the variation of the boundary 1-form resulting from the variation of the action. We obtain
\begin{equation}
	\tilde{\varpi}_\text{s}=\int_{\Sigma} e\delta e \delta \omega + i\frac{e^2}{4}\left( \overline{\psi}\gamma \delta \psi - \delta \overline{\psi}	\gamma \psi \right)\delta e + i \frac{e^3}{ 3!}  \delta \overline{\psi} \gamma \delta \psi,
\end{equation}
while
    \begin{align*}
        \iota_{\mathbb{X}}\tilde{\varpi}_\text{s}&=\int_\Sigma e \mathbb{X}_e \delta \omega +\left[ e\mathbb{X}_\omega +\ \frac{i}{4}e^2(  \overline{\psi} \gamma \mathbb{X}_\psi -  \mathbb{X}_{\overline{\psi}}\gamma \psi )\right]\delta e \\
        &\qquad + i\delta \overline{\psi} \left( -\frac{e^2}{4}\gamma \psi \mathbb{X}_e + \frac{e^3}{3!}\gamma \mathbb{X}_\gamma\right) + i \left( \frac{e^2}{4}\overline{\psi}\gamma \mathbb{X}_e + \frac{e^3}{3!}\mathbb{X}_{\overline{\psi}}\gamma \right)\delta \psi
    \end{align*}
The kernel of the presymplectic form is hence given by the following system of equations:
\begin{align*}
    & e \mathbb{X}_e=0 & & e\mathbb{X}_\omega + i\frac{e^2}{4}\left( \overline{\psi}\gamma \mathbb{X}_\psi +- \mathbb{X}_{\overline{\psi}}	\gamma \psi \right)=0 \\
    &  -\frac{e^2}{4}\gamma \psi \mathbb{X}_e + \frac{e^3}{3!}\gamma \mathbb{X}_\psi=0 & & \frac{e^2}{4}\overline{\psi}\gamma \mathbb{X}_e + \frac{e^3}{3!}\mathbb{X}_{\overline{\psi}}\gamma=0.
\end{align*}
We can first solve the last two equations, using that $\gamma$ is invertible and that  $W_3^{\partial,(0,0)}$ is injective (see Lemma \ref{lem:We_boundary}.\ref{lem: ker We0} in Appendix \ref{a:tec_and_lproofs}). We then find $\mathbb{X}_e=\mathbb{X}_\psi=\mathbb{X}_{\overline{\psi}}=0$ and 
$e\mathbb{X}_\omega=0$. The geometric phase space is a bundle over $\Omega_{\partial\text{,n.d.}}^{(1,1)} $ with local trivialization $F^\partial_\text{s}\simeq F^{\partial}\times S(\Sigma)\times \overline{S}(\Sigma)$.

\subsubsection{Choice of representative of \texorpdfstring{$[\omega]$}{Lg}}

First of all, we point out that in $N=4$ we can further simplify eq. \eqref{torsion spinor 0} in the bulk by noticing that $j_\gamma j_\gamma e^3 = 6 e j_\gamma e j_\gamma e  = 3 e j_\gamma j_\gamma e $, hence obtaining

    \begin{equation}\label{torsion spin}
        e\left[ d_\omega e + \frac{i}{4}(\bar{\psi}\gamma[e^2,\psi] - [e^2,\bar{\psi}]\gamma \psi) \right]=0
    \end{equation}

We will provide a generalization of Theorem \ref{thm:omegadecomposition} which allows to consider the newly found constraint \eqref{torsion spin}. We notice that, when restricted to the boundary, it splits into two equations
\begin{align*}
	&e\left[ d_\omega e + \frac{i}{4}(\bar{\psi}\gamma[e^2,\psi] - [e^2,\bar{\psi}]\gamma \psi )\right],\\
	&e_n\left[ d_\omega e + \frac{i}{4}(\bar{\psi}\gamma[e^2,\psi] - [e^2,\bar{\psi}]\gamma \psi )\right]=e (d_{\omega}e)_n.
\end{align*}
We will take the second one as an inspiration for the structural constraint (which enables us to fix the representative of $[\omega]$), while the first one is the invariant one.
Following \cite{CCS2020}, we reformulate the theorem fixing the representative of $\omega$ in the new setting.
	
\begin{theorem}\label{thm: omegadecompspin}
 Suppose that  $g^{\partial}$, the metric induced on the boundary, is nondegenerate. Given any $\widetilde{\omega} \in \Omega^{1,2}$, there is a unique decomposition 
	\begin{equation} \label{omegadecompspin}
		\widetilde{\omega}= \omega +v,
	\end{equation}
	with $\omega$ and $v$ satisfying 
	\begin{equation}\label{omegareprfix2spin}
		ev=0 \quad \text{ and } \quad e_n\left[ d_\omega e + \frac{i}{4}(\bar{\psi}\gamma[e^2,\psi] - [e^2,\bar{\psi}]\gamma \psi )\right]\in \Ima W_1^{\partial,(1,1)}.
	\end{equation}
\end{theorem}
\begin{proof}
	Let $\widetilde{\omega} \in \Omega_\partial^{1,2}$. From Lemma \ref{lem:Omega2,2_d4} we deduce that there exist unique $\sigma \in \Omega_\partial^{1,1}$ and $v \in \text{Ker} W_1^{\partial,(1,2)}$ such that 
	\begin{align*}
		e_n\left[ d_\omega e + \frac{i}{4}(\bar{\psi}\gamma[e^2,\psi] - [e^2,\bar{\psi}]\gamma \psi )\right]  = e \sigma + e_n [v,e].
	\end{align*}
	We define $\omega := \widetilde{\omega} - v $. Then $\omega$ and $v$ satisfy \eqref{omegadecompspin} and \eqref{omegareprfix2spin}.
	
	For uniqueness, suppose that $\widetilde{\omega}= \omega_1 + v_1 = \omega_2 +v_2$ with $ev_i =0$ and $e_n d_{\omega_i} e \in \Ima W_1^{\partial,(1,1)}$ for $i=1,2$. Hence 
	$$e_n  d_{\omega_1} e- e_n  d_{\omega_2} e = e_n  [v_2-v_1, e] \in  \Ima W_1^{\partial,(1,1)}.$$ Hence from Lemma \ref{lem:Omega2,1_d4} and Lemma \ref{lem:Omega2,2_d4} (for which we need nondegeneracy of $g^\partial$), we deduce $v_2-v_1 =0$, since $v_2-v_1 \in Ker W_1^{\partial,(1,2)}$.
\end{proof}

\subsubsection{Poisson brackets of the constraints}

We now project the equations of motion to the boundary, casting them into constraints, which will define the physical space of boundary fields as the coisotropic submanifold of the geometric phase space given by the zero-locus of the constraints. Dirac equation \eqref{Dirac spinor 00} does not project to the boundary because it is a top form, hence we only obtain the constraints

\begin{align*}
	L_c=&\int_{\Sigma}c\left( e d_{{\omega}}e + \frac{i}{(8 \cdot 3!)}\overline{\psi}\left( j_{\gamma}j_{\gamma} e^3 \gamma + \gamma j_{\gamma}j_{\gamma} e^3 \right)\psi\right),\\
		P_\xi&=\int_{\Sigma} \frac{1}{2}\iota_\xi e^2 F_\omega + \iota_\xi(\omega-\omega_{0}) \left( e d_{{\omega}}e - \frac{i}{8 \cdot 3!}\overline{\psi}\left( j_{\gamma}j_{\gamma} e^3 \gamma + \gamma j_{\gamma}j_{\gamma} e^3 \right)\psi\right)\\
		&\qquad+  \frac{i}{2\cdot 3!}\iota_\xi e^3(\overline{\psi}\gamma d_\omega \psi -  d_\omega\overline{\psi}\gamma \psi),\\
	 H_\lambda =& \int_{\Sigma} \lambda e_n\left[ eF_\omega + \frac{e^3}{3!}\Lambda + i\frac{e^2}{4}\left( \overline{\psi}\gamma d_\omega \psi - d_\omega\overline{\psi}\gamma \psi \right) \right].
\end{align*}
	
	\begin{remark}
		As it turns out, using \eqref{jgamma alpha e }, we can rewrite $L_c$ to make the action of the internal symmetry group on the fields more evident. In particular we obtain 
			\begin{equation}
				L_c=\int_{\Sigma} c e d_\omega e - i \frac{e^3}{2\cdot 3!} \left( [c,\overline{\psi}]\gamma \psi - \overline{\psi} \gamma [c,\psi] \right),
			\end{equation}
		while $P_\xi$ becomes
			\begin{equation}
				\begin{split}
				P_\xi&= \int_\Sigma \frac{1}{2}\iota_\xi e^2 F_\omega + \iota_\xi (\omega- \omega_0) e d_\omega e  - \frac{i}{8\cdot 3!}\iota_\xi e^3 \overline{\psi}\left( -[\omega-\omega_0,\overline{\psi}\gamma\psi + \overline{\psi}\gamma[\omega-\omega_0,\psi] \right)\psi\\
				&\qquad + \frac{i}{2\cdot 3!}\iota_\xi e^3 (\overline{\psi}\gamma d_\omega \psi - d_\omega \overline{\psi}\gamma \psi)\\
				&=\int_{\Sigma} \frac{1}{2}\iota_\xi e^2 F_\omega + \iota_\xi(\omega-\omega_{0}) e  d_{{\omega}}e - i \frac{e^3}{2\cdot 3!} \left( \overline{\psi} \gamma \iota_\xi d_{\omega_0}(\psi) - \iota_\xi d_{\omega_0}(\overline{\psi})\gamma \psi  \right)\\
				&=\int_{\Sigma} \frac{1}{2}\iota_\xi e^2 F_\omega + \iota_\xi(\omega-\omega_{0}) e  d_{{\omega}}e - i \frac{e^3}{2\cdot 3!} \left( \overline{\psi} \gamma \mathrm{L}_{\xi}^{\omega_0}(\psi) - \mathrm{L}_\xi^{\omega_0}(\overline{\psi})\gamma \psi  \right).
				\end{split}
			\end{equation}
	\end{remark}	
	\begin{theorem} \label{thm:first-class-constraints_spin}
		The constraints $L_c$, $P_{\xi}$, $H_{\lambda}$ define a coisotropic submanifold  with respect to the symplectic structure $\varpi_{\text{s}}$. Their Poisson brackets\footnote{We point out that one should not confuse $L$ with $\mathrm{L}$, which respectively indicate the constraint and the Lie derivative} read
		
		\begin{align*}\label{brackets-of-constraints_YM}
			& \{P_{\xi}, P_{\xi}\}  =  \frac{1}{2}P_{[\xi, \xi]}- \frac{1}{2}L_{\iota_{\xi}\iota_{\xi}F_{\omega_0}} &\{H_\lambda,H_\lambda\}=0\\
			&\{L_c, P_{\xi}\}  =  L_{\mathrm{L}_{\xi}^{\omega_0}c} & \{L_c, L_c\} = - \frac{1}{2}L_{[c,c]}  	;
		\end{align*}
		\begin{align*}
			& \{L_c,  H_{\lambda}\}  = - P_{X^{(a)}} + L_{X^{(a)}(\omega - \omega_0)_a} -H_{X^{(n)}} {}\\
			& \{P_{\xi},H_{\lambda}\}  =  P_{Y^{(a)}} -L_{ Y^{(a)} (\omega - \omega_0)_a} +  H_{ Y^{(n)}} ,
		\end{align*}
	where $X= [c, \lambda e_n ]$, $Y = \mathrm{L}_{\xi}^{\omega_0} (\lambda e_n)$ and $Z^{(a)}$, $Z^{(n)}$ are the components of $Z\in\{X,Y\}$ with respect to the frame $(e_a, e_n)$.	
	\end{theorem}
	\begin{proof}
	We first compute the Hamiltonian vector fields of the constraints:
	
		\begin{align*}
			\delta L_c & = \int_{\Sigma} [c,e] e \delta \omega + \left( e d_\omega c   + \frac{i}{4} e^2  \left(\overline{\psi} \gamma [c,\psi]  - [c,\overline{\psi}] \gamma \psi \right)\right)\delta e       \\
			& \qquad + \frac{i}{2\cdot 3!} e^3 \left[ \delta \bar{\psi} \gamma [c,\psi] - \bar{\psi}\gamma [c,\delta \psi]  + [c,\delta \bar{\psi}] \gamma \psi + [c,\bar{\psi}]\gamma \delta \psi \right]\\
            & \overset{\blacktriangledown}{=} \int_{\Sigma} [c,e] e \delta \omega + \left( e d_\omega c   + \frac{i}{4} e^2  \left(\overline{\psi} \gamma [c,\psi]  - [c,\overline{\psi}] \gamma \psi \right)\right)\delta e\\
            & \qquad + \frac{i e^3}{3!} \left[ [c,\bar{\psi}]\gamma \delta \psi + \delta\bar{\psi}\gamma[c,\psi] \right] - \frac{i e^3}{2\cdot 3!} \left[ \delta\bar{\psi}[c,\gamma]\psi + \bar{\psi}[c,\gamma]\delta\psi \right]\\
			& =\int_{\Sigma} [c,e] e \delta \omega + \left( e d_\omega c   + \frac{i}{4} e^2  \left([c,\overline{\psi}] \gamma \psi + \overline{\psi} \gamma [c,\psi]  \right)\right)\delta e \\
			& \qquad + \frac{i}{3!} e^3 \left[ \delta\overline{\psi}\left( \frac{1}{2}[c,\gamma] \psi + \gamma[c,\psi] \right) + \left( [c,\overline{\psi}]\gamma - \frac{1}{2}\bar{\psi}[c,\gamma] \right)\delta\psi \right],
		\end{align*}
	where in the last passage we used that
        \begin{align}
            \bar{\psi}\gamma[c,\delta\psi]&=\bar{\psi}[c,\gamma]\delta\psi-[c,\bar{\psi}]\gamma\delta\psi\\
            [c,\delta\bar{\psi}]\gamma\psi&=\delta\bar{\psi}\gamma[c,\psi]-\delta\bar{\psi}[c,\gamma]\psi
        \end{align}
    which can easily be proved using the following identity\footnote{A proof of this identity can be found in appendix \ref{proof id iotagammac}}
		\begin{equation}\tag{$\blacktriangledown$}\label{identity on iotagammac}
			j_{\gamma}j_{\gamma} c  \gamma= -\gamma j_{\gamma} j_{\gamma} c - 4 j_{\gamma} c =  -\gamma j_{\gamma} j_{\gamma} c + 4 [c,\gamma] .
		\end{equation}
		We also get
		\begin{align*}
			\delta P_\xi & =\int_{\Sigma} - e \delta e \left( \mathrm{L}_\xi^{\omega_0}(\omega-\omega_0) + \iota_\xi F_{\omega_0} - \frac{i}{4} e \left( \overline{\psi}\gamma \mathrm{L}_\xi^{\omega_0}(\psi) - \mathrm{L}_\xi^{\omega_0}(\overline{\psi})\gamma \psi  \right) \right)\\
			& \qquad- \mathrm{L}_\xi^{\omega_0}(e) e \delta \omega
			+ i \delta\overline{\psi} \left( - \frac{e^3}{2\cdot 3!} \gamma \mathrm{L}_\xi^{\omega_0}(\psi) \right) + \frac{i e^3}{2\cdot 3!} \overline{\psi} \gamma \mathrm{L}_\xi^{\omega_0}(\delta \psi) \\
			& \qquad - \frac{i e^3}{2\cdot 3!} \mathrm{L}_\xi^{\omega_0}(\delta \overline{\psi}) \gamma \psi -  \frac{i}{2\cdot 3!} e^3 \mathrm{L}_\xi^{\omega_0}(\overline{\psi}) \gamma \delta \psi \\
			&= \int_{\Sigma}  - e \delta e \left( \mathrm{L}_\xi^{\omega_0}(\omega-\omega_0) + \iota_\xi F_{\omega_0} - \frac{i}{4} e \left( \overline{\psi}\gamma \mathrm{L}_\xi^{\omega_0}(\psi) - \mathrm{L}_\xi^{\omega_0}(\overline{\psi})\gamma \psi  \right) \right)\\
			& \qquad - \mathrm{L}_\xi^{\omega_0}(e) e \delta \omega 	- i \delta\overline{\psi} \left(  \frac{e^3}{ 3!} \gamma \mathrm{L}_\xi^{\omega_0}(\psi) - \frac{1}{2\cdot 3!} \mathrm{L}_\xi^{\omega_0}(e^3) \gamma \psi \right) \\
			& \qquad -i \left(  \frac{e^3}{3!} \mathrm{L}_\xi^{\omega_0}(\overline{\psi}) \gamma + \frac{1}{2\cdot 3!} \mathrm{L}_\xi^{\omega_0}(e^3)  \overline{\psi}\gamma \right) \delta \psi,
		\end{align*}	
		
	    \begin{align*}
	        \delta H_\lambda &=\int_{\Sigma} \lambda e_n \left[ F_\omega + \frac{\Lambda }{2}e^2 + i \frac{e}{2}\left(\overline{\psi}\gamma d_\omega \psi - d_\omega \overline{\psi}\gamma \psi \right) \right]\delta e +  d_\omega (\lambda e_n e) \delta \omega\\
	        &\qquad + \frac{i}{4}\lambda e_n e^2 \bigg[ \delta \overline{\psi} \gamma d_\omega \psi - \overline{\psi}\gamma d_\omega \delta \psi + d_\omega\delta \overline{\psi} \gamma \psi + d_\omega \overline{\psi}\gamma \delta\psi\\
	        &\qquad+ \overline{\psi}\gamma[\delta \omega,\psi] -[\delta\omega,\overline{\psi}]\gamma\psi  \bigg]\\
	        &= \int_{\Sigma} \lambda e_n \left[ F_\omega + \frac{\Lambda }{2}e^2 + i \frac{e}{2}\left(\overline{\psi}\gamma d_\omega \psi - d_\omega \overline{\psi}\gamma \psi \right) \right]\delta e +  d_\omega (\lambda e_n e) \delta \omega\\
	        &\qquad + i \delta \overline{\psi}\left[ \lambda e_n \frac{e^2}{4}\gamma d_\omega \psi - d_\omega\left( \lambda e_n \frac{e^2}{4}\gamma \psi \right)  \right]\\
	        &\qquad +i \left[ \lambda e_n \frac{e^2}{4}d_\omega \overline{\psi}\gamma + d_\omega \left( \lambda e_n \frac{e^2}{4}\overline{\psi}\gamma \right)\right]\delta \psi \\
	        &\qquad + \frac{i}{16}\lambda \overline{\psi}\left( j_{\gamma} j_{\gamma} (e_n e^2)\gamma - \gamma j_{\gamma} j_{\gamma} (e_n e^2) \right)\psi \delta \omega.
	    \end{align*}	
	
	We are then left with
		\begin{align*}
			&\mathbb{L}_e = [c,e] &\mathbb{L}_\psi=[c,\psi] \\  &\mathbb{L}_\omega = d_{\omega} c + \mathbb{V}_L\label{e:ham_vf_J} &\mathbb{L}_{\overline{\psi}}=[c,\overline{\psi}]\\
			&\mathbb{P}_e = - \mathrm{L}_{\xi}^{\omega_0} e &\mathbb{P}_\psi=-\mathrm{L}_\xi^{\omega_0}(\psi)   \\  &\mathbb{P}_\omega = - \mathrm{L}_{\xi}^{\omega_0} (\omega-\omega_0) - \iota_ {\xi}F_{\omega_0} + \mathbb{V}_P   &\mathbb{P}_{\overline{\psi}}=-\mathrm{L}_{\xi}^{\omega_0}(\overline{\psi}).
		\end{align*}
	    \begin{align*}
	         &\mathbb{H}_e = d_\omega (\lambda e_n ) + \lambda \sigma + \frac{i}{4}\lambda\overline{\psi}\left(j_{\gamma} e_n j_{\gamma} e \gamma - \gamma j_{\gamma} e_n j_{\gamma} e   \right)\psi\\
	         &e \mathbb{H}_\omega = \lambda e_n \left( F_\omega + \frac{\lambda}{2}e^2 \right) - i \frac{\lambda e_n}{4} e (\overline{\psi}\gamma d_\omega \psi - d_\omega\overline{\psi} \gamma \psi)\\
	         & \frac{e^3}{3!}\gamma \mathbb{H}_{\psi}= \frac{\lambda e_n}{2} e^2 \gamma d_\omega \psi - \frac{\lambda e_n}{4} ed_\omega e \gamma \psi+ \frac{i}{64} \lambda e \left[\overline{\psi}\left( j_{\gamma}j_{\gamma}(e_n e^2)\gamma - \gamma j_{\gamma} j_{\gamma}(e_n e^2) \right)  \psi\right]\gamma\psi\\
	         & \frac{e^3}{3!}\mathbb{H}_{\overline{\psi}} \gamma = \frac{\lambda e_n}{2} e^2  d_\omega \overline{\psi} \gamma + \frac{\lambda e_n}{4} ed_\omega e \overline{\psi}  \gamma- \frac{i}{64} \lambda e \overline{\psi}\gamma\left[\overline{\psi}\left( j_{\gamma}j_{\gamma}(e_n e^2)\gamma - \gamma j_{\gamma} j_{\gamma}(e_n e^2) \right)  \psi\right]
	    \end{align*}

	The Poisson brackets of the constraints are:
		
		\begin{align*}
			\{L_c,L_c\}&=\int_{\Sigma} (\cdots) - \frac{i}{4}e^2\left( -\frac{1}{4}\overline{\psi}j_{\gamma}j_{\gamma} c \gamma \psi + \frac{1}{4}\overline{\psi}\gamma j_{\gamma} j_{\gamma} c \psi \right)[c,e] \\
			&\qquad + \frac{i}{3!}e^3[c,\overline{\psi}]\gamma[c,\psi]\\
			&=\int_{\Sigma} (\cdots) + \frac{i}{8\cdot 3!}\overline{\psi}(j_{\gamma} j_{\gamma} c \gamma - \gamma j_{\gamma} j_{\gamma} c) \psi [c,e^3] \\
			&\qquad + \frac{i}{16\cdot 3!}e^3 \overline{\psi}j_{\gamma}j_{\gamma} c \gamma j_{\gamma}j_{\gamma} c \psi\\
			&\overset{\blacktriangledown}{=}\int_\Sigma (\cdots) - \frac{i}{2 \cdot 3!}\left( [c,\overline{\psi}]\gamma \psi - \overline{\psi}\gamma [c,\psi] \right)[c,e^3] \\
			&\qquad + \frac{i}{32}\frac{e^3}{3!}\overline{\psi}\left( -\gamma j_{\gamma}j_{\gamma} c j_{\gamma}j_{\gamma} c + 4 [c,\gamma]j_{\gamma}j_{\gamma} c - j_{\gamma}j_{\gamma} c j_{\gamma}j_{\gamma} c \gamma + 4 j_{\gamma}j_{\gamma} c [c,\gamma] \right)\psi\\
			&=\int_{\Sigma} (\cdots) + \frac{i}{2\cdot 3!}e^3 \left( \overline{\psi}\gamma[c,[c,\psi]] - [c,[c,\overline{\psi}]]\gamma\psi \right) \\
			&=\int_\Sigma -\frac{1}{2}[c,c]e d_\omega e + \frac{i}{4\cdot 3!}e^3 \left( [[c,c],\overline{\psi}]\gamma \psi - \overline{\psi}\gamma[[c,c],\psi] \right)\\
			&=-\frac{1}{2}L_{[c,c]},
		\end{align*}
	where in last few steps we used the graded Jacobi identity to prove
	    \begin{equation*}
	        [c,[c,\psi]=-\frac{1}{2}[[c,c],\psi]
	    \end{equation*}
	and the fact that
		\begin{equation*}
			\gamma j_{\gamma}j_{\gamma} c j_{\gamma}j_{\gamma} c= j_{\gamma}j_{\gamma} c j_{\gamma}j_{\gamma} c \gamma  + 4 j_{\gamma}j_{\gamma} c  j_{\gamma} c + 4 j_{\gamma} c j_{\gamma}j_{\gamma} c . 
		\end{equation*}

	\begin{align*}
		\{L_c,P_\xi\}&=\int_\Sigma (\cdots) - \frac{i}{2\cdot 3!}e^3 \left( [c,\overline{\psi}]\gamma \mathrm{L}_\xi^{\omega_0}\psi - \overline{\psi}\gamma \mathrm{L}_\xi^{\omega_0}([c,\psi]) + \mathrm{L}_\xi^{\omega_0}([c,\overline{\psi}])\gamma\psi + \mathrm{L}_\xi^{\omega_0}\overline{\psi}\gamma [c,\psi]\right)\\
		&\qquad - \frac{i}{2\cdot 3!}[c,e^3] \left( \overline{\psi}\gamma \mathrm{L}_\xi^{\omega_0}\psi - \mathrm{L}_\xi^{\omega_0}\overline{\psi}\gamma \psi \right)\\
		&=\int_\Sigma (\cdots) - \frac{i}{2\cdot 3!}e^3\big( [c,\overline{\psi}]\gamma \mathrm{L}_\xi^{\omega_0}\psi + \mathrm{L}_\xi^{\omega_0}\overline{\psi}\gamma [c,\psi] - \overline{\psi}\gamma [\mathrm{L}_\xi^{\omega_0}c,\psi] + \overline{\psi}\gamma [c,\mathrm{L}_\xi^{\omega_0}\psi]\\
		&\qquad -[\mathrm{L}_\xi^{\omega_0}c,\overline{\psi}]\gamma\psi - [c,\mathrm{L}_\xi^{\omega_0}\overline{\psi}]\gamma\psi \big) - \frac{i}{2\cdot 3!}[c,e^3](\overline{\psi}\gamma \mathrm{L}_\xi^{\omega_0} \psi - \mathrm{L}_\xi^{\omega_0}\overline{\psi}\gamma \psi)\\
		&=\int_\Sigma (\cdots) - \frac{i}{2\cdot 3!}\big( [c,\overline{\psi}]\gamma \mathrm{L}_\xi^{\omega_0}\psi + \mathrm{L}_\xi^{\omega_0} \overline{\psi}\gamma [c,\psi] - \overline{\psi}\gamma [\mathrm{L}_\xi^{\omega_0}c,\psi] + [\mathrm{L}_\xi^{\omega_0}c, \overline{\psi}]\gamma \psi \\
		&\qquad -[c,\overline{\psi}]\gamma \mathrm{L}_\xi^{\omega_0}\psi - \overline{\psi}[c,\gamma]\mathrm{L}_\xi^{\omega_0}\psi -\mathrm{L}_\xi^{\omega_0}\overline{\psi}[c,\gamma]\psi - \mathrm{L}_\xi^{\omega_0}\overline{\psi}\gamma [c,\psi]\\
		&\qquad +\overline{\psi}[c,\gamma]\mathrm{L}_\xi^{\omega_0}\psi + \mathrm{L}_\xi^{\omega_0}\overline{\psi}[c,\gamma]\psi \big)\\
		&=\int_{\Sigma} \mathrm{L}_\xi^{\omega_0}c e d_\omega e - \frac{i}{2\cdot 3!}e^3 \big( [\mathrm{L}_\xi^{\omega_0}c,\overline{\psi}]\gamma \psi - \overline{\psi}\gamma [\mathrm{L}_\xi^{\omega_0}c,\psi]  \big)\\
		&= L_{\mathrm{L}_\xi^{\omega_0}c},
	\end{align*}
	where in the second to last passage we used that 
	\begin{align*}
	    \overline{\psi}\gamma[c,\mathrm{L}_\xi^{\omega_0}\psi]&=-[c,\overline{\psi}]\gamma \mathrm{L}_\xi^{\omega_0}\psi - \overline{\psi}[c,\gamma]\mathrm{L}_\xi^{\omega_0}\psi,\\
	    [c,\mathrm{L}_\xi^{\omega_0}\overline{\psi}]\gamma \psi&=\mathrm{L}_\xi^{\omega_0}\overline{\psi}[c,\gamma]\psi + \mathrm{L}_\xi^{\omega_0}\overline{\psi}\gamma[c,\psi].
	    \end{align*}
	
		\begin{align*}
		\{L_c,H_\lambda\}&=\mathbb{L}_c(H_\lambda)=\int_{\Sigma}(\cdots) + \lambda e_n \bigg\{ \frac{i}{4}[c,e^2]\big( \overline{\psi}\gamma d_\omega \psi -  d_\omega \overline{\psi}\gamma \psi \big) + \frac{i}{4}e^2 \big( [c,\overline{\psi}]\gamma d_\omega \psi \\
		&\qquad - \overline{\psi}\gamma d_\omega [c,\psi] + d_\omega ([c,\overline{\psi}])\gamma\psi + d_\omega \overline{\psi}\gamma [c,\psi]\\
		&\qquad +\overline{\psi}\gamma[d_\omega c ,\psi] - [d_\omega c, \overline{\psi}]\gamma \psi \big) \bigg\}\\
		&\overset{\triangledown}{=}\int_{\Sigma}(\cdots) - [c,\lambda e_n]\frac{i}{4}e^2\big( \overline{\psi}\gamma d_\omega  \psi -  d_\omega \overline{\psi}\gamma \psi\big) - i\frac{\lambda e_n}{4}e^2 \bigg\{ -\overline{\psi}[c,\gamma]d_\omega \psi - d_\omega \overline{\psi}[c,\gamma]\psi \\
		&\qquad + \overline{\psi}[c,\gamma d_\omega\psi -\overline{\psi}\gamma [c,d_\omega\psi] + [d_\omega c, \overline{\psi}]\gamma \psi - [c, d_\omega \overline{\psi}]\gamma \psi + [c,d_\omega\overline{\psi}]\gamma \psi\\
		&\qquad + d_\omega \overline{\psi}[c,\gamma]\psi + \overline{\psi}\gamma [d_\omega c,\psi] - [d_\omega c, \overline{\psi}] \gamma \psi - \overline{\psi}\gamma [d_\omega c,\psi] +  \overline{\psi} \gamma [c,d_\omega \psi]   \bigg\}\\
		&=\int_\Sigma -[c,\lambda e_n]\big( e F_\omega + \frac{\Lambda}{2}e^2 + \frac{i}{4}e^2 (\overline{\psi}\gamma d_\omega \psi - d_\omega \overline{\psi}\gamma \psi)  \big)\\
		&= - P_{[c,\lambda e_n]^{(a)}} - H_{[c,\lambda e_n]^{(a)}} + L_{ [c,\lambda e_n]^{(a)}(\omega-\omega_0)_{(a)} }
	\end{align*}
	having used the following identities, which can be easily found
	    \begin{equation}\tag{$\triangledown$}
	        \begin{split}
	             d_\omega \overline{\psi}\gamma [c,\psi]&=[c,d_\omega \overline{\psi}]\gamma \psi + d_\omega \overline{\psi}[c,\gamma]\psi,\\
            [c,\overline{\psi}]\gamma d_\omega  \psi&= \overline{\psi}[c,\gamma] d_\omega \psi - \overline{\psi}\gamma [c,d_\omega \psi].
	        \end{split}
	    \end{equation}

	\begin{align*}
		\{P_\xi,P_\xi\}&=\int_{\Sigma}  (\cdots) + \frac{i}{2\cdot 3!}\mathrm{L}_\xi^{\omega_0}(e^3)\big( \overline{\psi}\gamma \mathrm{L}_\xi^{\omega_0}\psi - \mathrm{L}_\xi^{\omega_0}\overline{\psi}\gamma\psi  \big) - \frac{i}{2\cdot 3!}e^3 \big\{ -\mathrm{L}_\xi^{\omega_0}\overline{\psi }\gamma \mathrm{L}_\xi^{\omega_0}\psi\\
		& \qquad +  \overline{\psi}\gamma \mathrm{L}_\xi^{\omega_0}\mathrm{L}_\xi^{\omega_0}\psi -\mathrm{L}_\xi^{\omega_0}\mathrm{L}_\xi^{\omega_0}\overline{\psi}\gamma \psi  -\mathrm{L}_\xi^{\omega_0}\overline{\psi}\gamma \mathrm{L}_\xi^{\omega_0}\psi \big\}\\
		&=\int_\Sigma (\cdots) - \frac{i}{2\cdot 3!}e^3 \big\{ \mathrm{L}_\xi^{\omega_0}\overline{\psi}\gamma \psi + \overline{\psi}\gamma \mathrm{L}_\xi^{\omega_0}\mathrm{L}_\xi^{\omega_0}\psi - \mathrm{L}_\xi^{\omega_0}\mathrm{L}_\xi^{\omega_0}\overline{\psi }\gamma \psi + \mathrm{L}_\xi^{\omega_0}\overline{\psi}\gamma \mathrm{L}_\xi^{\omega_0} \psi \\
		&\qquad  -\mathrm{L}_\xi^{\omega_0}\overline{\psi }\gamma \mathrm{L}_\xi^{\omega_0}\psi +  \overline{\psi}\gamma \mathrm{L}_\xi^{\omega_0}\mathrm{L}_\xi^{\omega_0}\psi - \mathrm{L}_\xi^{\omega_0}\mathrm{L}_\xi^{\omega_0}\overline{\psi}\gamma \psi - \mathrm{L}_\xi^{\omega_0} \overline{\psi} \gamma \mathrm{L}_\xi^{\omega_0}\psi \big\} \\
		&=\int_\Sigma (\cdots) - \frac{i}{3!}e^3 \big( \overline{\psi}\gamma\mathrm{L}_\xi^{\omega_0}\mathrm{L}_\xi^{\omega_0}\psi -\mathrm{L}_\xi^{\omega_0}\mathrm{L}_\xi^{\omega_0}\overline{\psi}\gamma \psi \big)\\
		&\overset{\clubsuit}{=}\int_\Sigma (\cdots) -\frac{i}{2\cdot 3!}e^3 \big( \overline{\psi}\gamma \mathrm{L}_{[\xi,\xi]}^{\omega_0}\psi - \mathrm{L}_{[\xi,\xi]}^{\omega_0}\overline{\psi}\gamma\psi \big)\\
		&\qquad + \frac{i}{2\cdot 3!}e^3 \big( [\iota_\xi\iota_\xi F_{\omega_0},\overline{\psi}]\gamma\psi - \overline{\psi}\gamma[\iota_\xi\iota_\xi F_{\omega_0},\psi] \big)\\
		&=\frac{1}{2}P_{[\xi,\xi]}-\frac{1}{2}L{\iota_\xi\iota_\xi F_{\omega_{0}}};
	\end{align*}

    \begin{align*}
	 	\{P_\xi,H_\lambda\}&=\int_{\Sigma} (\cdots) + \lambda e_n \bigg\{ -\frac{i}{4}\mathrm{L}_\xi^{\omega_0}(e^2)(\overline{\psi}\gamma d_\omega \psi - d_\omega \overline{\psi}\gamma \psi + \frac{i}{4}e^2 \big[ -\mathrm{L}_\xi^{\omega_0}\overline{\psi}\gamma d_\omega\psi \\
	 	&\qquad +\overline{\psi}\gamma d_\omega\mathrm{L}_\xi^{\omega_0}\psi - d_\omega \mathrm{L}_\xi^{\omega_0}\overline{\psi}\gamma \psi - d_\omega \overline{\psi}\gamma \mathrm{L}_\xi^{\omega_0}\psi \\
	 	&\qquad - \overline{\psi}\gamma[ \iota_\xi F_{\omega_0} + \mathrm{L}_\xi^{\omega_0}(\omega- \omega_0), \psi ]+ [\iota_\xi F_{\omega_0}+ \mathrm{L}_\xi^{\omega_0}(\omega-\omega_0),\overline{\psi}]\gamma\psi  \big]\bigg\}\\
	 	&=\int_\Sigma (\cdots) + i  \mathrm{L}_\xi^{\omega_0}(\lambda e_n)\frac{e^2}{4}\big( \overline{\psi}\gamma d_\omega \psi - d_\omega \overline{\psi}\gamma\psi\big) +i \frac{\lambda e_n}{4}e^2 \big\{ \overline{\psi}\gamma \mathrm{L}_\xi^{\omega_0} d_\omega \psi \\
	 	&\qquad -\mathrm{L}_\xi^{\omega_0}d_\omega \overline{\psi}\gamma \psi + \overline{\psi} \gamma d_\omega \mathrm{L}_\xi^{\omega_0}\psi - d_\omega \mathrm{L}_\xi^{\omega_0}\overline{\psi}\gamma \psi\\
	 	&\qquad - \overline{\psi}\gamma[ \iota_\xi F_{\omega_0} + \mathrm{L}_\xi^{\omega_0}(\omega- \omega_0), \psi ]+ [\iota_\xi F_{\omega_0}+ \mathrm{L}_\xi^{\omega_0}(\omega-\omega_0),\overline{\psi}]\gamma\psi \big\}\\
	 	&\overset{\blacklozenge}{=}\int_\Sigma (\cdots) + i  \mathrm{L}_\xi^{\omega_0}(\lambda e_n)\frac{e^2}{4}\big( \overline{\psi}\gamma d_\omega \psi - d_\omega \overline{\psi}\gamma\psi\big) +i \frac{\lambda e_n}{4}e^2 \big\{ \overline{\psi}\gamma[ \mathrm{L}_\xi^{\omega_0}\omega,\psi] - [\mathrm{L}_\xi^{\omega_0}\omega,\overline{\psi}]\gamma\psi\\
	 	&\qquad - \overline{\psi}\gamma[ \iota_\xi F_{\omega_0} + \mathrm{L}_\xi^{\omega_0}(\omega- \omega_0), \psi ]+ [\iota_\xi F_{\omega_0}+ \mathrm{L}_\xi^{\omega_0}(\omega-\omega_0),\overline{\psi}]\gamma\psi  \big\}\\
	 	&=\int_\Sigma (\cdots) + i  \mathrm{L}_\xi^{\omega_0}(\lambda e_n)\frac{e^2}{4}\big( \overline{\psi}\gamma d_\omega \psi - d_\omega \overline{\psi}\gamma\psi\big) + i \frac{\lambda e_n}{4}e^2 \big\{ \overline{\psi}\gamma [\mathrm{L}_\xi^{\omega_0}\omega_0 - \iota_\xi F_{\omega_0},\psi]\\
	 	&\qquad -[\mathrm{L}_\xi^{\omega_0}\omega_0 - \iota_\xi F_{\omega_0},\overline{\psi}]\gamma\psi\big\}\\
	 	&=\int_\Sigma \mathrm{L}_\xi^{\omega_0}(\lambda e_n) \left( eF_\omega + \frac{\Lambda}{2}e^2 + \frac{i}{4}e^2 \big( \overline{\psi}\gamma d_\omega \psi-d\omega \overline{\psi}\gamma\psi \big) \right)\\
        &=P_{ \mathrm{L}_\xi^{\omega_0}(\lambda e_n)^{(a)}} + H_{ \mathrm{L}_\xi^{\omega_0}(\lambda e_n)^{(a)}} - L_{ \mathrm{L}_\xi^{\omega_0}(\lambda e_n)^{(a)}(\omega-\omega_0)},
	 \end{align*}
	 where we used that $\mathrm{L}_\xi^{\omega_0}\omega_0 - \iota_\xi F_{\omega_0}=-d\iota_\xi \omega_0$ and the following identity:
	    \begin{equation}\tag{$\blacklozenge$}
	        \begin{split}
	            \mathrm{L}_\xi^{\omega_0}d_\omega\psi = - d_\omega \mathrm{L}_\xi^{\omega_0} \psi  + [\mathrm{L}_\xi^{\omega_0}\omega, \psi].
	        \end{split}
	    \end{equation}
    Furthermore, recalling that $d_{\omega_0}\gamma=0$, it is quite easy to see that
        \begin{equation*}
            \begin{split}
                &\overline{\psi}\gamma [d\iota_\xi \omega_0,\psi]- [d \iota_\xi \omega_0,\overline{\psi}]\gamma\psi=-[d\iota_\xi \omega_0,\overline{\psi}\gamma\psi]=0.
            \end{split}
        \end{equation*}

    Now, before computing $\{H_\lambda,H_\lambda\}$, we first notice that the Hamiltonian vector field associated to $H_\lambda$ can be rewritten as
    \begin{align*}
        e\gamma \mathbb{H}_\psi&= 3 \lambda e_n \gamma d_\omega \psi - \frac{3}{2}\lambda \sigma \gamma \psi+ \frac{3i}{8}\lambda \beta\\ 
        e \mathbb{H}_{\overline{\psi}}\gamma &= 3 \lambda e_n d_\omega \overline{\psi}\gamma + \frac{3}{2}\overline{\psi}\gamma \lambda \sigma - \frac{3}{8}i \lambda \overline{\beta},
    \end{align*}
    with $\beta:=\overline{\psi}\big( j_{\gamma} e_n j_{\gamma} e \gamma - \gamma j_{\gamma} e_n j_{\gamma} e \big)\gamma\psi$, hence
    \begin{align*}
        \{ H_\lambda , H_\lambda \}&= \int_\Sigma i \left[ \frac{\lambda e_n}{2} \mathbb{H}_{\overline{\psi}}e^2 \gamma d_\omega \psi - \frac{1}{4}d_\omega(\lambda e_n) e^2 \mathbb{H}_{\overline{\psi}}\gamma\psi - \frac{\lambda e_n}{2}d_\omega e e \mathbb{H}_{\overline{\psi}}\gamma\psi  \right]\\ 
        & \qquad +i \left[  \left( \frac{\lambda e_n}{2}d_\omega \overline{\psi} + \frac{1}{4}d_\omega(\lambda e_n) \overline{\psi}  \right) e^2 \gamma \mathbb{H}_\psi + \frac{\lambda e_n}{2}d_\omega e \overline{\psi} e \gamma \mathbb{H}_\psi  \right]\\
        &=\int_\Sigma \frac{3}{4\cdot 32} d_\omega(\lambda e_n) \lambda \overline{\psi}\gamma \psi \big[ \overline{\psi} \big( j_{\gamma}j_{\gamma}(e_n e^2) \gamma - \gamma j_{\gamma} j_{\gamma} (e_n e^2) \big)\psi \big]\\
        &\qquad - \frac{3}{4\cdot 32} d_\omega(\lambda e_n) \lambda \overline{\psi}\gamma \psi \big[ \overline{\psi} \big( j_{\gamma}j_{\gamma}(e_n e^2) \gamma - \gamma j_{\gamma} j_{\gamma} (e_n e^2) \big)\psi \big]=0,
    \end{align*}
    where all the remaining terms vanish because they are either proportional to $\lambda^2=0$ or $e_n^2=0$.
	\end{proof}
	
	\subsection{The BFV formalism of the theory of gravity coupled to the spinor field}\label{BFV_spinor}

     Since there are no additional constraints, the BFV discussion of the spinor field coupled to gravity is very similar to the case of the scalar field.
     
     \begin{theorem}\label{thm:BFVaction spinor}
		Let $\mathcal{F}_{s}$ be the bundle 
		\begin{equation*}
			\mathcal{F}_s= \mathcal{F}_{PC} \times S(\Sigma) \times \overline{S}(\Sigma),
		\end{equation*}
		where the additional fields are denoted by $\psi\in S(\Sigma)$ and $\overline{\psi}\in\overline{S}(\Sigma)$. 
		The symplectic form and the action functional on $\mathcal{F}_s$ are respectively defined by
		\begin{align*}
			\varpi_s &= \varpi_{PC}+\int_{\Sigma} i \frac{e^3}{3!}\delta\overline{\psi}\gamma \delta \psi - \frac{i}{4}e^2\delta e\left( \delta \overline{\psi} \gamma \psi - \overline{\psi}\gamma \delta \psi\right) , \\
			S_s &= S_{PC} +\int_{\Sigma}i \frac{e^3}{2\cdot3!}\left( \overline{\psi}\gamma d_\omega \psi - d_\omega\overline{\psi} \gamma \psi \right).
		\end{align*}
		Then the triple $(\mathcal{F}_s, \varpi_s, S_s)$ defines a BFV structure on $\Sigma$.
	\end{theorem}

	\begin{proof}
		The proof can be copied \emph{mutatis mutandis} from the one of Theorem \ref{thm:BFVaction}
\end{proof}

\newpage
\appendix
\section{Clifford algebras and spin groups}\label{clif spin}

This first appendix is useful when defining exactly what spinor fields appear in the context of field theory. We will mainly follow \cite{fatibene2018}, \cite{spingeom} and \cite{Kostant1987}.

\subsection{Clifford algebras}
Let $V$ be a real vector space of dimension $N$ with an inner product of signature $(r,s)$. Let $\eta_{ab}$ be the matrix diag$(-1,\cdots,-1,1,\cdots,1)$ with $r$ plus 1 and $s$ minus 1, giving the inner product on $V$ with respect to an orthonormal basis $\{v_a\}$. 

We define the Clifford algebra on $V$ by means of its universal property. In particular	
\begin{definition}[Clifford map]
	A \text{Clifford map} is a pair consisting of an associative algebra $A$ with unity and a linear map $\phi \colon V\rightarrow A$ satisfying $\forall u,v\in V$
	\begin{equation}
		\phi(u)\phi(u)=-\eta(u,u)\mathbb{1}_A
	\end{equation}
\end{definition}

The Clifford algebra of $V$ is the solution corresponding to the universal problem, that is
\begin{definition}[Clifford algebra]
	The Clifford algebra $\mathcal{C}(V)$ is an associative algebra with unit together with a Clifford map $i\colon V\rightarrow \mathcal{C}(V)$ such that any Clifford map factors through a unique algebra homomorphism from $C(V)$. In other words, given any Clifford map $(A,\phi)$ there is a unique algebra homomorphism $\Phi\colon \mathcal{C}( V) \rightarrow A$ such that $\phi = \Phi \circ i$
	\begin{equation}
		\begin{tikzcd}
			V \arrow[r, "\phi"] \arrow[d, "i"'] & A \\
			\mathcal{C}(V) \arrow[ru, "\Phi"']            &  
		\end{tikzcd}
	\end{equation}
\end{definition}

The Clifford algebra of $V$ is unique up to isomorphisms.

We give a model for such an algebra. Consider the tensor algebra $T(V):=\mathbb{R}\oplus V \oplus V \oplus \cdots$ and quotient it out  by the two-sided ideal $I(V)$ generated by $v\otimes v + \eta(v,v)\mathbb{1}$, i.e.
\begin{equation}
    \mathcal{C}(V):=\frac{T(V)}{I(V)}.
\end{equation}

Notice that $T(V)$ is a $\mathbb{Z}$-graded algebra. The ideal $I(V)$ is spanned by elements that are not necessarily homogeneous, therefore the $\mathbb{Z}$--grading is lost in the Clifford algebra. However, the generators of $I(V)$ are even, therefore $\mathcal{C}(V)$ will be $\mathbb{Z}_2$-graded. In particular, it splits into
\begin{equation}
	\mathcal{C}(V)=\mathcal{C}_0(V) \oplus \mathcal{C}_1(V)
\end{equation}
Another important property, for any two vectors $v,w \in V$, is the following
\begin{equation}
	\begin{split}
		(v+w)^2& =v^2+ vw + wv + w^2 = - \eta(v,v)\mathbb{1} - \eta(w,w)\mathbb{1} + \{v,w\}\\
		& = - \eta(v+w,v+w)\mathbb{1}=-\eta(v,v)\mathbb{1}  - \eta(w,w)\mathbb{1} -2\eta(v,w)\mathbb{1}\\
		\Rightarrow&\quad \{v,w\}:=vw+wv=-2\eta(v,w)\mathbb{1}
	\end{split}
\end{equation}
Now, considering an orthonormal basis $\{v_a\}$ of $V$, setting the first $s$ elements $\{v_A\}$ such that $\eta(v_A,v_A)=-1$ and the second $r$ elements $\{v_i\}$ such that $\eta(v_i,v_i)=1$, we obtain $\{v_a,v_b\}=-2\eta_{ab}\mathbb{1}$. This means that when $a\neq b$, $v_av_b=-v_bv_a$ and that $v_av_a=\pm \mathbb{1}$.

At this point, since every element in the tensor algebra $T(V)$ is a finite linear combination of the product of finite elements in the basis of $V$, then to obtain elements in $\mathcal{C}(V)$ we simply apply the constraint $\{v_a,v_b\}=-2\eta_{ab}\mathbb{1}$. In other words, a basis of Clifford algebra is in the form
\begin{equation}
	\mathbb{1} \quad v_a \quad v_{ab}:=v_a v_b \quad v_{abc}:=\underset{a<b<c}{v_a v_b v_c} \quad \cdots \quad v:=v_0v_1\cdots v_{N-1}
\end{equation}
The $\mathbb{Z}_2$-grading is now clearer, as we can interpret even (odd) elements of $\mathcal{C}(V)$ to be finite linear combinations of products of an even (odd) number of elements of the basis $V$. In particular, the even part $\mathcal{C}_0(V)$ is a sub-algebra of $\mathcal{C}(V)$, while the odd part $\mathcal{C}_1(V)$ is not (it does not contain the unity). They are both $2^{N-1}$-dimensional, making $\mathcal{C}(V)$ $2^N$-dimensional.

\subsection{Pin and spin groups}

\begin{definition}[grading map]
	Consider the Clifford map $i\colon V\rightarrow \mathcal{C}(V)$. By abuse of notation, this map sends $v$ to $v$ inside $\mathcal{C}(V)$. Defining $a:=-i\colon v\rightarrow \mathcal{C}(V):v\mapsto -v$, it has the property that $a(v)a(v)=-\eta(v,v)\mathbb{1}$. We can extend it to the whole $\mathcal{C}(V)$ as $\alpha\colon \mathcal{C}(V)\rightarrow \mathcal{C}(V)$ by restricting it to the identity on even elements, to minus the identity on odd elements. This map is called \text{grading} since it essentially defines the $\mathbb{Z}_2$-grading on $\mathcal{C}(V)$.
\end{definition}
Clearly we have that $\alpha\circ\alpha=\mathbb{1}$, therefore $\alpha$ is invertible and equal to its inverse.

\begin{definition}[transpose]
	Let $S=v_1 v_2\cdots v_k\in \mathcal{C}(V)$. We define the \text{transpose} of $S$ to be
	\begin{equation}
		{}^t(S)={}^t(v_1 v_2\cdots v_k):=v_k\cdots v_2 v_1=:\overline{S}
	\end{equation}
	It is well defined since the generators of the Clifford ideal are invariant under the transposition.
\end{definition}
Furthermore, the transpose preserves the grading, namely ${}^t(\alpha(S))=\alpha(^t(S))$.

\begin{definition}[Pin and Spin groups]
	It is a well known fact that not all elements in $\mathcal{C}(V)$ are invertible. Let us define the multiplicative subgroup $C(V)\subset \mathcal{C}(V)$ of invertible elements and the further subgroup $S(V)\subset C(V)\subset \mathcal{C}(V)$ of invertible elements $S$ whose inverse is proportional to their transpose, namely such that $S \overline{S}\propto \mathbb{1}$. 
	
	We define the \text{Pin group} Pin$(V)$ to be the subgroup of $S(V)$ generated by unit vectors (i.e. such that $v^2=\eta(v,v)=\pm {1}$). The \text{Spin group} Spin$(V)$ is defined to be the intersection of Pin$(V)$ with the even Clifford subalgebra $\mathcal{C}(V)$. 
\end{definition}

Elements in Spin$(V)$ are products of an even number of unit vectors, $S=v_1 v_2\cdots v_{2k}$. In this case it is easy to find the inverse of $S$, as	
\begin{equation}
	S^{-1}=\frac{(-1)^{n}}{|v_1|^2\cdots |v_{n}|^2}v_n\cdots v_2 v_1=\pm {}^tS
\end{equation}

\subsection{The covering of spin groups}
Consider an element $S$ in Pin$(V)$, namely $S=v_1v_2\cdots v_k$ and a vector $w\in V$. By abuse of notation, we denote $w:=i(w)\in \mathcal{C}(V)$. We also denote $w^\parallel:=\frac{\eta(v,w)}{\eta(v,v)}v$ to be the component of $w$ parallel to $v\in V$, assuming $v$ to be a unit vector. The perpendicular component is defined as $w^\perp := w - w^\parallel$

We define a linear map on $V$ depending on the unit vector $v$ as $$l(v)\colon V\rightarrow V:w\mapsto \alpha(v)w v^{-1}$$

\begin{lemma}
	The map $l(v)$ is a reflection of $w$ about the plane orthogonal to the unit vector $v$
\end{lemma}
\begin{proof}
	Recalling that $uu=-\eta(v,v)\mathbb{1}=-|v|^2\mathbb{1}$, we have 
	\begin{equation}
		\begin{split}
			\alpha(v)w v^{-1}&=-v w v^{-1} = |v|^{-2} v w v = |v|^{-2}\left( u w^\perp v + v w^\parallel v \right) \\
			&= |v|^{-2}(-vv w^\perp - \eta(v,w^\perp)v - |v|^2w^\parallel)\\
			& = w^\perp - w^\parallel
		\end{split}
	\end{equation}
\end{proof}

Furthermore, being $l(v)$ a reflection, it is an element of $O(V)$, the group of orthogonal transformations on $V$.

\begin{definition}[covering]
	For any $S$ in Pin$(V)$, we can extend the definition of $l$ as $l(S)\colon V\rightarrow V$ and $l(S)\in O(V)$,
	\begin{equation}
		l(S)(w):=\alpha(S)w S^{-1}=(l(v_1)\circ l(v_2)\circ \cdots \circ l(v_k))(w)
	\end{equation}
	In particular, $l\colon \text{Pin}(V)\rightarrow O(V)$ is called \text{covering of the Pin group}.
	
	Since reflections are transformations with determinant -1, the composition of an even number of reflections will have determinant +1, therefore when we restrict to Spin$(V)$, we have the \text{covering of the Spin group} $l\colon \text{Spin}(V)\rightarrow SO(V)$.
\end{definition}

It can be checked that the map $l\colon \mathrm{Pin}(V)\rightarrow O(V)$ is a group homomorphism, and so is $l$ when restricted to Spin$(V)$.

\begin{proposition}
	The covering map is not injective but is surjective. Furthermore, there is a short exact sequence
	    \begin{equation}
	        \begin{tikzcd}
        0 \arrow[r] & \mathbb{Z}_2 \arrow[r] & \mathrm{Spin} \arrow[r, "l"] & SO \arrow[r] & 0.
        \end{tikzcd}
	    \end{equation}
\end{proposition}

\subsection{Spinor representations}

Let $S$ be a (complex) vector space. A complex representation of the Clifford algebra $\mathcal{C}(V)$ is an algebra homomorphism
    \begin{equation}
        \mathcal{C}(V)\rightarrow \mathrm{End}(S).
    \end{equation}

S is called \text{spinor space}.

We now are interested in the case where $N=\dim V=2m$ is even.
\begin{definition}[Dirac spinor and gamma representation]
    Let $S:=\mathbb{C}^{2^m}$. A \text{Dirac Spinor} is any element of $S$, on which $\mathcal{C}(V)$ acts as the full algebra of $2^m\times 2^m$ complex matrices.
    
    In particular, considering $V:=\mathbb{C}^{N}$, it acts on $S$ via the \text{gamma representation}
    
    \begin{equation}
        \begin{split}
            \gamma\colon & \mathcal{C}(V)\rightarrow \mathrm{End}(S) \\
            & v_i \mapsto \gamma_i:=\gamma(v_i),
        \end{split}
    \end{equation}
    where $\gamma_i$ is the i--th Dirac gamma matrix in $N$ dimensions. In general, for $1\leq j \leq m$
        \begin{align*}
            \gamma_j&:= 1\otimes 1 \otimes \cdots \otimes\underbrace{ \sigma_1}_{\textsf{$j$--th elem.}} \otimes \sigma_3 \otimes \cdots \otimes \sigma_3\\;
            \gamma_{j+m}&:= 1 \otimes 1 \otimes \cdots \otimes \underbrace{\sigma_2}_{\text{$j$-th elem.}}\otimes \sigma_3 \otimes \cdots \otimes \sigma_3.
        \end{align*}
        
    Here $\sigma_i$ are the usual Pauli matrices.    
\end{definition}

\begin{remark}
    Considering the ``volume element'' $v_1 \cdots v_N$ on $\mathcal{C}(V)$, its image under the gamma representation defines 
        \begin{equation}
            \gamma_{2m+1}:=(-i)^m \gamma_1 \cdots \gamma_{2m}=(-i)^m \gamma(v_1 \cdots v_{2m})
        \end{equation}
    In particular it can be proven that $\gamma_{2m+1}$ has eigenvalues $\pm1$, hence there is a splitting into eigenspaces $S=S^+ \oplus S^-$. $S^\pm$ are called \text{spaces of Weyl spinors} of positive/negative chirality.
\end{remark}

\subsection{Lie Algebra of Spin group}

    It is quite easy to see that Lie$({C}(V))=\mathcal{C}(V)$. We are interested in the Lie algebra of Spin$(V)\subset C(V)$.
\begin{proposition}
    Let $V$ be an $N$--dimensional real vector space. Lie$($Spin$(V))$ is a Lie subalgebra of $\mathcal{C}(V)$, given by
        \begin{equation}
            \mathrm{Lie(Spin}(V))=\wedge ^2 V
        \end{equation}
\end{proposition}
This can be seen by noticing that the double cover $l\colon \text{Spin}(V)\rightarrow SO(V)$ reduces to an isomorphism of Lie algebras (locally their tangent space at the identity is the same)
    \begin{equation}
        \begin{split}
            \dot{l}\colon &\mathfrak{spin}(V) \rightarrow \mathfrak{so}(V)\\
            & a \longmapsto \dot{l}(a)=[a,\cdot],
        \end{split}
    \end{equation}
    where, for all $u\in V$, the $[a,u]\in SO(V)$ is given by
        \begin{equation}
            [a,u]:=\frac{\partial}{\partial t}\big\vert_{t=0}(e^{-ta}u e^{ta}).
        \end{equation}

Now, knowing that $\{ v_a \wedge v_a \}$ is a basis for $\mathfrak{so}(V)$, we compute a basis for $\mathfrak{spin}(V)$. 

Define $v_{ab}:=\frac{1}{4}[v_a,v_b]$, then, for all $u=u^c v_c\in V$ 
    \begin{align*}
        \dot{l}(v_{ab})u&=\frac{1}{4}[[v_a,v_b],u]=\frac{1}{2}[v_a v_b, u]\\
        &=\frac{1}{2}\left( v_a v_b u - u v_a v_b \right)= \frac{1}{2}\left( v_a v_b u - u v_a v_b + v_a u v_b - v_a u v_b \right)\\
        &= \eta(u,v_a)v_b - \eta(v_b,u)v_a=u^c (\delta^d_b \eta_{ac} - \delta_a^d \eta_{bc}) v_d,\\
    \end{align*}
    hence  
        \begin{equation}
            \dot{l}(v_{ab})^d_c = \delta^d_b \eta_{ac} - \delta_a^d \eta_{bc} = -(M_{ab})^d_c
        \end{equation}
    where $M_{ab}$ are the generators of the Lorentz group $SO(V)$ in the fundamental representation. This implies that $-\frac{1}{4}[v_a,v_b]$ defines a basis for $\mathfrak{spin}(V)$.

\section{Technical results and lenghty proofs}\label{a:tec_and_lproofs}	\subsection{Technical results}\label{tec res}
In this appendix we present a collection of results that are useful throughout the paper, especially in the constraint analysis of the theories and in some calculations. We refer to \cite{CCS2020} for the proofs that we leave out.

First we present a precise definition of the brackets $(\cdot,\cdot)$ we employed in the previous chapters.
	\begin{definition}[internal product]
		Let $A,B\in \Omega^{(0,1)}$ and $C,D\in \Omega^{(0,2)}$. Expanding them in the bases $\{e_\mu\}$ and $\{e_\mu e_\nu\}$, we obtain
			\begin{align*}
				&(A,B)=g_{\mu\nu}A^\mu B^\nu,\\
				&(C,D)=g_{\mu\rho}g_{\nu\sigma}C^{\mu\nu}B^{\rho\sigma},
			\end{align*}
		which is a simple consequence of the fact that $(e_\mu,e_\nu)=g_{\mu\nu}$ by definition of the vielbein.
		
		We also notice that
			\begin{align*}
				&(e,A)=(e_\mu dx^\mu,A)=-dx^\mu g_{\mu\nu} A^\nu;\\
				&(e^2,C)=-dx^\mu dx^\nu g_{\mu\rho}g_{\nu\sigma}C^{\rho\sigma}.
			\end{align*}
	\end{definition}

	\begin{lemma}
		For all $N>0$, define
		\begin{equation}
			\left[ \frac{N}{2} \right]:=\begin{cases}
				\frac{N}{2} \quad \text{if $N$ even}\\
				\frac{N-1}{2} \quad \text{if $N $ odd}
			\end{cases}.
		\end{equation}
		Then 
		$$e^n=(-1)^{\left[ \frac{n}{2} \right]}e_{\mu_1}\cdots e_{\mu_n}dx^{\mu_1}\cdots dx^{\mu_n}$$.
	\end{lemma}
	\begin{proof}
		We proceed by induction. In the case $n=1$ we have $e^n=e_\mu dx^\mu=(-1)^{\vert\frac{n}{2}\vert}e_\mu dx^\mu$.
		
		Assuming that the identity holds for $n$, we prove that it is true also for $n+1$, in fact
			\begin{align*}
				e^{n+1}&=e^n e_\mu dx^\mu=(-1)^{\left[ \frac{n}{2} \right]}e_{\mu_1}\cdots e_{\mu_n}dx^{\mu_1}\cdots dx^{\mu_n}  e_\mu dx^\mu \\
				&=(-1)^{\left[ \frac{n}{2} \right]+n}e_{\mu_1}\cdots e_{\mu_{n+1}}dx^{\mu_1}\cdots dx^{\mu_{n+1}}\\
				&=(-1)^{\left[ \frac{n+1}{2} \right]}e_{\mu_1}\cdots e_{\mu_{n+1}}dx^{\mu_1}\cdots dx^{\mu_{n+1}}.
			\end{align*}
	
		In fact $(-1)^{\left[ \frac{n+1}{2} \right]}=(-1)^{{\left[ \frac{n}{2} \right]}+n}$:
		\begin{itemize}
			\item if $n=even$, then $(-1)^{\left[ \frac{n+1}{2} \right]}=(-1)^{{\left[ \frac{n}{2} \right]}}=(-1)^{{\left[ \frac{n}{2} \right]}+n}$;
			\item if $n=odd$, then $(-1)^{\left[ \frac{n+1}{2} \right]}=(-1)^{\left[ \frac{n}{2} \right]+1}=(-1)^{\left[ \frac{n}{2} \right]+n}$.
		\end{itemize}
	\end{proof}
\subsection{In the bulk}

\begin{lemma}\label{useful identities}
	Let $C,D\in \Omega^{(0,2)}$ and $A,B\in \Omega^{(0,1)}$. Then the following identities hold
		\begin{enumerate}
			\item $\frac{e^N}{N}(A,B)=(-1)^{|A|+|B|}e^{N-1}(e,A)B$	\label{(1,0)};
			\item $\frac{e^{N-2}}{2(N-2)!}(e^2,C)D=\frac{e^N}{N!}(C,D)$ \label{(2,0)}.
		\end{enumerate}
\end{lemma}
\begin{proof}
	
	\begin{enumerate}
		\item We use $e_\mu$ as a basis for $\Omega^{(0,1)}$. Then
			\begin{align*}
				e^{N-1}(e,A)B&=(-1)^{\left[ \frac{N-1}{2} \right]+1}e_{\mu_1}\cdots e_{\mu_{N-1}}dx^{\mu_1}\cdots dx^{\mu_{N-1}}dx^\mu g_{\mu\nu}A^\nu  B^\rho e_\rho\\
				&=(-1)^{\left[ \frac{N-1}{2} \right] + |A| + |B| + N +1}e_{\mu_1}\cdots e_{\mu_{N-1}} e_\rho dx^{\mu_1}\cdots dx^{\mu_{N-1}}dx^\mu  g_{\mu\nu}A^\nu B^\rho\\
				&=(-1)^{\left[ \frac{N-1}{2} \right] + |A| + |B| + N +1}(N-1)!e_1\cdots e_N d^Nx g_{\mu\nu}A^\mu B^\nu\\
				&=(-1)^{\left[ \frac{N-1}{2} \right] + |A| + |B| + N +1+ \left[ \frac{N}{2} \right]}\frac{e^N}{N}(A,B) \\
				&=(-1)^{|A|+|B|}\frac{e^N}{N}(A,B);
			\end{align*}
		\item We now use $e_\mu e_\nu$ as a basis for $\Omega^{(0,2)}$. Then
			\begin{align*}
				e^{N-2}(e^2,C)D&=(-1)^{\left[ \frac{N-2}{2} \right]+1} e_{\mu_1}\cdots e_{\mu_{N-2}}e_\rho e_\sigma dx^{\mu_1}\cdots dx^{\mu_{N-2}}dx^{\mu}dx^{\nu} g_{\mu\alpha} g_{\nu\beta}C^{\alpha \beta}D^{\rho\sigma}\\
				&=(-1)^{\left[ \frac{N-2}{2} \right]+1} 2(N-2)!e_1\cdots e_N dx^1\cdots dx^N  g_{\mu\alpha} g_{\nu\beta}C^{\alpha \beta}D^{\mu\nu}\\
				&=(-1)^{\left[ \frac{N-2}{2} \right]+\left[ \frac{N}{2} \right] +1} \frac{2 (N-2)!}{N!}e^N (C,D) \\
				&=(-1)^{2\left[\frac{N}{2}\right]}\frac{2 (N-2)!}{N!}e^N (C,D)\\
				&= \frac{2 (N-2)!}{N!}e^N (C,D).
			\end{align*}
	\end{enumerate}
\end{proof}

\begin{lemma}\label{varro}
	Let $\varrho_n:=(e^n,\cdot)\colon \Omega^{(0,n)}\rightarrow\Omega^{(n,0)}$ and let $e$ be nondegenerate. Then for $N\geq 2$ we have that $\varrho_n$ is bijective for $n=1,2$.
\end{lemma}
\begin{proof}
	It is a simple consequence of the fact that the metric $g_{\mu\nu}$ is invertible.
\end{proof}

\begin{corollary}\label{corollary bulk 1}
	We then have a corollary of the two previous lemmas. Let $\alpha\in\Omega^{(1,0)}$, $\pi\in\Omega^{(0,1)}$, $\omega\in\Omega^{(2,0)}$ and $C\in \Omega^{(0,2)}$, then
	\begin{enumerate}
		\item $e^{N-1}\alpha B=(-1)^{|B|+1}$$\frac{e^N}{N}\alpha_\mu B^\mu $;
		\item $e^{N-2}\omega D=-\frac{2(N-2)!}{N!}e^N\omega_{\mu\nu}D^{\mu\nu}$.
	\end{enumerate}
\end{corollary}
\begin{proof}

		By Lemma \ref{varro} there have to exist $B\in \Omega^{(0,1)}$ and $D\in\Omega^{(0,2)}$ such that $\alpha=(e,A)$ and $\omega=(e^2,C)$\footnote{Of course $|\alpha|=|A|$ and $|\omega|=|C|$}. In particular, this means 
			\begin{align*}
				&C^{\rho\sigma}=-g^{\rho\mu}g^{\sigma\nu}\omega_{\mu\nu},\\
				&A^{\nu}=(-1)^{1+|\alpha|}g^{\nu\mu}\alpha_\mu.
			\end{align*}
		We then simply apply Lemma \ref{useful identities}

\end{proof}

\begin{lemma}\label{useful id bulk}
	Let $W_k^{(i,j)}$ be such that $W_k^{(i,j)}\colon \Omega^{(i,j)}\rightarrow \Omega^{(i+k,j+k)}:\alpha\mapsto e^k\wedge \alpha$. Then the following propositions are true
	\begin{enumerate}
		\item $W_{N-1}^{(1,0)}$ is injective\label{lem: W_N-1 bijective};
		\item $W_{N-2}^{(2,0)}$ is injective\label{lem: W_N-2 bijective}.
	\end{enumerate}
\end{lemma}	
\begin{proof}
	We prove the statements locally. Choosing as usual a local basis $\{e_\mu\}$ of $\mathcal{V}$, we have that
	\begin{enumerate}
		\item $\Ker{W^{(1,0)}_{N-1}}:=\{\alpha\in\Omega^{(1,0)}\hspace{1mm}\vert\hspace{1mm}e^{N-1}\alpha=0\}$. In particular, this means 
			\begin{equation}
				e_{\mu_1}\cdots e_{\mu_{N-1}}\alpha_\nu dx^{\mu_1}\cdots 	dx^{\mu_{N-1}}dx^{\nu}\propto e_{[1}\cdots e_{N-1}\alpha_{N]}d^Nx=0\quad \Leftrightarrow \quad \alpha_\mu=0,
			\end{equation}
		hence proving that $\Ker{W_{N-1}^{(1,0)}}=\{0\}$;
		\item 
		$\Ker{W_{N-2}^{(2,0)}}:=\{\omega\in\Omega^{(2,0)}\hspace{1mm}\vert\hspace{1mm}e^{N-2}\omega=0\}$. Similarly as before, we find:
			\begin{equation}
				e_{\mu_1}\cdots e_{\mu_{N-2}}\omega_{\nu\rho} dx^{\mu_1}\cdots dx^{\mu_{N-2}}dx^{\nu}dx^{\rho}\propto e_{[1}\cdots e_{N-2}\omega_{N_1, N]}d^Nx=0\quad \Leftrightarrow \quad \omega_{\mu\nu}=0	,
			\end{equation}
			hence proving that $\Ker{W_{N-2}^{(2,0)}}=\{0\}$.
	\end{enumerate}
\end{proof}
\subsection{On the boundary}
We now generalize Lemma \ref{useful identities} to the boundary. We can simply do this by setting $e_N dx^N\leadsto e_n$. Then it is easy to see that
	\begin{equation}\label{sub to bound}
		e^N\leadsto N e_n e^{(N-1)}.
	\end{equation}
Hence we have
	\begin{lemma}\label{useful id boundary}
			Let $C,D\in \Omega^{(0,2)}_\partial$ and $A,B\in \Omega^{(0,1)}_\partial$. Then the follwing identities hold
		\begin{enumerate}
			\item $e_n\frac{e^{N-1}}{(N-1)!}(A,B)=(-1)^{|A|+|B|}\left[ \frac{1}{(N-2)!}e_n e^{N-2}(e,A)B + \frac{e^{N-1}}{(N-1)!}(e_n,A)B \right]$;
			\item $e_n \frac{e^{N-1}}{(N-1)!}(C,D)=\left[ e_n \frac{e^{N-3}}{2(N-3)!}(e^2,C)D + \frac{e^{N-2}}{(N-2)!}(e_n e,C)\right]$.
		\end{enumerate}
	\end{lemma}
	\begin{proof}
		We simply impose the substitution defined in equation \eqref{sub to bound}, noticing also that $(A,B)\longrightarrow (A,B)$ and $(C,D)\rightarrow(C,D)$.
	
		\begin{enumerate}
			\item \begin{align*}
				\frac{e^N}{N!}(A,B)&\leadsto e_n \frac{e^{N-1}}{(N-1)!}(A,B);
			\end{align*}
			\begin{align*}
				\frac{e^{N-1}}{(N-1)!}(e,A)B&\leadsto \frac{N}{(N-1)!}\left\{ \frac{N-1}{N}e_n e^{N-2}(e,A)B + \frac{e^{N-1}}{N}(e_n,A)B \right\}\\
				&\leadsto \frac{1}{(N-2)!}e_n e^{N-2}(e,A)B + \frac{1}{(N-1)!}e^{N-1}(e_n,A)B;
			\end{align*}
			\item \begin{align*}
				\frac{e^N}{N!}(C,D)&\leadsto e_n \frac{e^{N-1}}{(N-1)!}(C,D);
			\end{align*}
			\begin{align*}
					\frac{e^{N-2}}{2(N-2)!}(e^2,C)D&\leadsto N\left\{ \frac{N-2}{N}e_n e^{N-3}(e^2,C)D + \frac{2}{N}e^{N-2}(e_n e,C)D \right\}\\
					&\leadsto \frac{1}{2(N-3)!}e_n e^{N-3}(e^2,C)D + \frac{1}{(N-2)!}e^{N-2}(e_ne,C)D.
			\end{align*}
		\end{enumerate}
	\end{proof}

We recall
\begin{equation}
	\begin{split}
		W_k^{\partial(i,j)}&\colon \Omega^{(i,j)}_\partial\longrightarrow\Omega^{(i+k,j+k)}_\partial\\
		&:\alpha\longmapsto e^k\wedge\alpha.
	\end{split}
\end{equation}
Then we have 

\begin{lemma} \label{lem:We_boundary}
	The maps $W_{k}^{ \partial, (i,j)}$ have the following properties for $N \geq 4$:
	
	\begin{enumerate}
		\item $W_{N-3}^{\partial, (2,1)}$ is surjective; \label{lem:Wep21}
		\item $W_{N-3}^{\partial, (1,1)}$ is injective; \label{lem:Wep11}
		\item $W_{N-3}^{\partial, (1,2)}$ is surjective; \label{lem:kerWe12} 
		\item $W_{k}^{\partial(0,0)}$ is injective; \label{lem: ker We0}
		\item $\dim \mathrm{Ker} W_{N-3}^{\partial, (1,2)} = \dim \mathrm{Ker} W_{N-3}^{\partial, (2,1)}$;\label{lem:kernel12-21}
		\item $W_{N-4}^{\partial, (2,1)}$ is injective. ($N \geq 5$)\label{lem:We5p21}	;
		\item $W_{N-2}^{\partial, (1,0)}$ is injective. \label{lem: We7}
	\end{enumerate}

\end{lemma}
\begin{proof}
	The proofs of the statements $(1)-(6)$ can be found in \cite{CCS2020} and \cite{CCS2020}, with the exception of $(4)$, which is easily seen since any $\phi\in\Omega^{(0,0)}$ is a function, hence $e^k$ just acts as a multiplication.  We just need to prove \ref{lem: We7}. Considering $A\in\Omega^{(1,0)}_\partial$, then	
		\begin{align*}
			e^{N-2}A&=e_{\mu_1}\cdots e_{\mu_{N-2}} A_\rho dx^{\mu_1}\cdots dx^{\mu_{N-2}} dx^{\rho}=0
		\end{align*}
	is satisfied if and only if $A_\rho=0$ for all $\rho=1,\cdots,N$, hence showing that $A=0$.
\end{proof}
\begin{lemma}\label{lem:Omega2,1_d4}
	Let $\alpha \in \Omega^{2,1}_\partial$. Then 
	\begin{align}\label{ConditionforOmega21_d4}
		\alpha=0 \qquad \Longleftrightarrow  \qquad \begin{cases}
			e^{N-3}\alpha =0 \\
			e_n e^{N-4}\alpha \in \Ima W_{N-3}^{\partial, (1,1)}
		\end{cases}.
	\end{align}
\end{lemma}

\begin{lemma}\label{lem:Omega2,2_d4}
	Let $\beta \in \Omega^{N-2,N-2}_\partial$. If $g^\partial$ is nondegenerate, there exist a unique $v \in \mathrm{Ker} W_{N-3}^{\partial, (1,2)}$ and a unique $\gamma \in \Omega_{\partial}^{1,1}$ such that 
	\begin{align*}
		\beta = e^{N-3} \gamma + e_n e^{N-4} [v, e].
	\end{align*}
\end{lemma}

\begin{proof}
	The proofs of the previous two lemmas are found in \cite{CCS2020}
\end{proof}

\subsection{Proofs of Lemmas \ref{lem:technical_for_scalar} and \ref{lem:technical_for_YM}}

\begin{proof}[Proof of Lemma \ref{lem:technical_for_scalar}]
    Let $N=4$. In this proof we fix as a basis of $\mathcal{V}$ the set $\epsilon_n, e_{\mu}$, $\mu=1,2,3$ where $\epsilon_n$ is a vector completing the basis. 
    
    With this choice, consider the kernel of the map ${W_3^{\partial(0,1)}}$. It is defined by the equation
    \begin{align*}
        X^a e_a e_{\mu_1}e_{\mu_2}e_{\mu_3} dx^{\mu_1}dx^{\mu_2}dx^{\mu_3}=0
    \end{align*}
    which implies $X^n=0.$ Hence we get that the components in the kernel are $X^1$, $X^2$ and $X^3$.
    
    Let us now consider the map $A_e$. Let $p\in\Ker{W_3^{\partial(0,1)}}$ be generated by $p=p^1e_1+p^2e_2+p^3e_3$. Then we have
    \begin{align*}
        (e, p)_{\mu}= e_{\mu}^a p^b \eta_{ab}= e_{\mu}^a e_{\nu}^b p^{\nu} \eta_{ab}= g^{\partial}_{\mu\nu}p^{\nu}
    \end{align*}
    where we used that $e_{\nu}^b=\delta_{\nu}^b$ in our basis. Now, if $g^\partial$ is nondegenerate, using normal geodesic coordinates, we obtain
    \begin{align*}
        (e, p)_{\mu} = \pm p^{\mu}
    \end{align*}
    depending on the sign of the elements in the diagonal of the diagonalized boundary metric. This shows that the map $A_e$ is injective and surjective.
\end{proof}

\begin{proof}[Proof of Lemma \ref{lem:technical_for_YM}]
    Let $N=4$. In this proof we fix as a basis of $\mathcal{V}$ the set $\epsilon_n, e_{\mu}$, $\mu=1,2,3$ where $\epsilon_n$ is a vector completing the basis. 
    
    With this choice, consider the kernel of the map ${W_2^{\partial(0,1)}}$. It is defined by the equationa
    \begin{align*}
        X^{ab} e_a e_b e_{\mu_1}e_{\mu_2} dx^{\mu_1}dx^{\mu_2}=0.
    \end{align*}
    which imply $X^{na}=0$ for $a=1,2,3$. Hence we get that the components in the kernel are $X^{ab}$ for $a,b=1,2,3$.
    
    Let us now consider the map $\phi_e$. Let $b\in\Ker{W_2^{\partial(0,1)}}\otimes\mathfrak{g}$ be generated by $b=b^{ab}e_ae_b$ for $a,b=1,2,3$. Then we have
    \begin{align*}
        \frac{1}{2}(e^2, b)_{\mu\nu}= e_{\mu}^a e_{\nu}^b b^{cd} \eta_{ac}\eta_{bd}=e_{\mu}^a e_{\nu}^b e_{\rho}^c e_{\sigma}^d b^{\rho\sigma} \eta_{ac}\eta_{bd}= g^{\partial}_{\mu\rho}g^{\partial}_{\nu\sigma}b^{\rho\sigma}
    \end{align*}
    where we used that $e_{\nu}^b=\delta_{\nu}^b$ in our basis. Now, if $g^\partial$ is nondegenerate, using normal geodesic coordinates, we obtain
    \begin{align*}
        \frac{1}{2}(e^2, b)_{\mu\nu} = \pm b^{\mu\nu}
    \end{align*}
    depending on the sign of the elements in the diagonal of the diagonalized boundary metric. This shows that the map $\phi_e$ is injective and surjective.
\end{proof}

\subsection{Proof of gamma matrices identities}

\begin{itemize}
    \item \texorpdfstring{\eqref{identity on iotagammac}}{triangle down}\label{proof id iotagammac}:
    
    Recall $c\in \Omega^{(0,2)}[1]$. In general if we let $c$ be any section of $\wedge^2 \mathcal{V}$, then
    \begin{align*}
        j_{\gamma} j_{\gamma} c \gamma &= \gamma^a \gamma^b \gamma ^c \iota_{v_a}\iota_{v_b}c v_c=(-1)^{|c|}\gamma^a \gamma^b \gamma ^c v_c \iota_{v_a}\iota_{v_b}c\\
        &= (-1)^{|c|}\left[ -\gamma^a \gamma^c\gamma^b - 2 \gamma^a \eta^{bc} \right]v_c \iota_{v_a}\iota_{v_b}c\\
        &= (-1)^{|c|}\left[ \gamma^c \gamma^a\gamma^b + 2\gamma^b \eta^{ac} - 2 \gamma^a \eta^{bc} \right]v_c \iota_{v_a}\iota_{v_b}c\\
        &=(-1)^{|c|}\gamma j_{\gamma}j_{\gamma} c + 2 (-1)^|c|v_c\left[ \gamma^b \eta^{ac} - \gamma^a \eta^{bc} \right]\eta_{ad}\eta_{bf}c^{df}\\
        &=(-1)^{|c|}\gamma j_{\gamma}j_{\gamma} c + 2 (-1)^|c|v_c\left[ \gamma^b \delta^c_d\eta_{bf} - \gamma^a \delta^c_f\eta_{ad} \right]c^{df}\\
        &=(-1)^{|c|}\gamma j_{\gamma}j_{\gamma} c - 2 \gamma_d(c^{dc}-c^{cd}) v_c\\
        &=(-1)^{|c|}\gamma j_{\gamma}j_{\gamma} c - 4 j_{\gamma} c,
    \end{align*}
    hence, in general
        \begin{equation*}
            [j_{\gamma} j_{\gamma} c, \gamma]= 4 [c,\gamma]:
        \end{equation*}

    \item eq. \eqref{jgamma omega e}:
    Consider any $\alpha \in \Omega^{(i,2)}$ of arbitrary parity. We rely on the fact that $e^{N-1}\gamma j_\gamma j_\gamma \alpha $ is a top form in $V$ and that $j_{v_a}$ respects the (graded) Leibniz rule for all $v_a$'s forming a basis of $V$
    \begin{align*}
            e^{N-1} \gamma j_\gamma j_\gamma \alpha &= \gamma^a \gamma^b \gamma^c e^{N-1} v_a j_{v_b}j_{v_c} \alpha \\
            &=  \gamma^a \gamma^b \gamma^c \left[ j_{v_b}(e^{N-1}) v_a + e^{N-1} \eta_{ab} \right] j_{v_c}\delta \alpha \\
            &= - \gamma^a \gamma^b \gamma^c \left[ v_a j_{v_c}j_{v_b} e^{N-1} - j_{v_b}e^{N-1} \eta_{ac} + j_{v_c}e^{N-1} \eta_{ab} \right]\alpha \\ 
            &= \left[ \gamma j_\gamma j_\gamma e^{N-1} - 2 \eta_{ab} \gamma^a \gamma^b \gamma^c j_{v_c}e^{N-1}  \right] \alpha \\
            &= \left[ \gamma j_\gamma j_\gamma e^{N-1} + 2N j_\gamma e^{N-1} \right]\alpha,
        \end{align*}
    analogously we find 
    \begin{align*}
        e^{N-1}  j_\gamma j_\gamma \alpha \gamma &= (-1)^{|\alpha|}\gamma^b \gamma^c \gamma^a e^{N-1} v_a j_{v_b}j_{v_c} \alpha \\
        &= (-1)^{|\alpha|} \left[ \gamma j_\gamma j_\gamma e^{N-1} - 2N j_\gamma e^{N-1} \right]\alpha,
    \end{align*}
    hence
    \begin{equation}\label{jgamma alpha e }
        \begin{split}
            &e^{N-1} \left[ \gamma j_\gamma j_\gamma \alpha + (-1)^{|\alpha|}j_\gamma j_\gamma \alpha \gamma \right] = \left[ \gamma j_\gamma j_\gamma e^{N-1} + j_\gamma j_\gamma e^{N-1} \gamma \right]\alpha,\\
            & e^{N-1} \left[ \bar{\psi} \gamma [\alpha,\psi] - [\alpha, \bar{\psi}]\gamma \psi \right]=\left[ \bar{\psi} \gamma [e^{N-1},\psi] - [e^{N-1}, \bar{\psi}]\gamma \psi \right]\alpha
        \end{split}
    \end{equation}
\end{itemize}

\subsection{Lenghty proofs of Section \ref{BFV YM}}
In this section we show explicitly equation	
\begin{eqnarray}
	&& \{S_0^1,S_1^0\}_f + \{S_0^0,S_1^1\}_f + \{S_0^1,S_1^1\}_f + \{S_1^0,S_1^1\}_g + \frac{1}{2}\{S_1^1,S_1^1\}_g=0 \label{EM BFV B.2}.
\end{eqnarray} 
\begin{equation}\label{S01S10f}
	\begin{split}
		&\{S_0^1,S_1^0\}_f= \iota_{Q_0^1}\iota_{Q_1^0}\varpi_f \\
		& = \iota_{Q_0^1} \int_{\Sigma} \iota_{Q_1^0}(e\delta e \delta \omega) + \Tr\iota_{Q_1^0}(\delta \rho \delta A)\\
		& = \int_\Sigma e Q_{1e}^0 Q_{0\omega}^1 \propto \int_{\Sigma} \lambda^2 =0, 
	\end{split}
\end{equation}
because $Q_{1\rho}^0=0$, $Q_{1A}^0=0$ and both  $ Q_{1e}^0$ and $e Q_{0\omega}^1$ are proportional to $\lambda$.
\begin{equation}\label{S00S11f}
	\begin{split}
		&\{S_0^0,S_1^1\}_f = \\
		& = \Tr\int_{\Sigma} \unl{[c,[c,\lambda e_n]^{(b)}e_b]^{(a)}(A-A_0)_{a}\mu^{\dagger}}{0011.1} - \unl{[c,\lambda e_n]^{(b)}\mathrm{L}^{\omega_0}_\xi(e_b)^{(a)}(A-A_0)_{a}\mu^{\dagger}}{0011.2} \\
		& \qquad - \unl{[c,\lambda e_n]^{(b)}\partial_b\xi^c e_c^{(a)}(A-A_0)_{a}\mu^\dag}{0011.3} - \unl{[c,\mathrm{L}^{\omega_0}_\xi(\lambda e_n)^{(b)}e_b]^a(A-A_0)_{a}\mu^{\dagger}}{0011.4}\\
		& \qquad  + \unl{\mathrm{L}^{\omega_0}_\xi(\lambda e_n)^{(b)}\mathrm{L}^{\omega_0}_\xi(e_b)^{(a)}(A-A_0)_{a}\mu^{\dagger}}{0011.5} + \unl{\mathrm{L}^{\omega_0}_\xi(\lambda e_n)^{(b)}\partial_b\xi^c e_c^{(a)}(A-A_0)_{a}\mu^\dag}{0011.6},
	\end{split}
\end{equation}
	where we used $\mathrm{L}^{\omega_0}_\xi(e)_b=\mathrm{L}^{\omega_0}_\xi(e_b) + \partial_b \xi^c e_c $.
\begin{equation}\label{S01S11f}
	\begin{split}
		&\{S_0^1,S_1^1\}_f = \\
		& = \Tr\int_{\Sigma} \unl{[c,\lambda e_n]^{(a)}(\mathrm{L}^{\omega_0}_\xi(A-A_0))_{a}\mu^\dagger}{0111.1}  - \unl{\mathrm{L}^{\omega_0}_\xi(\lambda e_n)^{(a)}(\mathrm{L}^{\omega_0}_\xi(A-A_0))_{a}\mu^\dagger}{0111.2}\\
		&\qquad +\unl{[c,\lambda e_n]^{(a)}(\iota_{\xi}F_{A_0})_a\mu^\dagger}{0111.3} - \unl{\mathrm{L}_\xi^{\omega_0}(\lambda e_n)^{(a)}(\iota_{\xi}F_{A_0})_a\mu^\dagger}{0111.4}  \\
		& \qquad - \unl{[c,\lambda e_n]^{(a)}(d_A\mu)_a  \mu^\dagger}{0111.5} + \unl{\mathrm{L}^{\omega_0}_\xi(\lambda e_n)^{(a)} (d_A\mu)_a \mu^\dagger}{0111.6};
	\end{split}
\end{equation}

\begin{equation}
	\begin{split}\label{S10S11g}
		&\{S^0_1,S_1^1\}_g = \\
		& =  \Tr\int_{\Sigma} \big\{  -\unl{\mathrm{L}_\xi^{\omega_0}([c,\lambda e_n]^{(n)}e_n)^{(a)}}{1011g.1} + \unl{\mathrm{L}_\xi^{\omega_0}(\mathrm{L}_\xi^{\omega_0}(\lambda e_n)^{(n)}e_n)^{(a)}}{1011g.2} + \unl{[\mathrm{L}^{\omega_0}_\xi(c),\lambda e_n]^{(a)}}{1011g.3}\\
		&\qquad-\unl{\frac{1}{2}\left[ [c,c] , \lambda e_n \right]^{(a)}}{1011g.4} -\unl{[c,\mathrm{L}_\xi^{\omega_0}(\lambda e_n)^{(n)}e_n]^{(a)}}{1011g.5} +  \unl{[c,[c,\lambda e_n]^{(n)}e_n]^{(a)}}{1011g.6}\\
		&\qquad  - \unl{\frac{1}{2}[\iota_{\xi}\iota_{\xi}F_{\omega_0},\lambda e_n]^{(a)}}{1011g.7}\big\}(A-A_0)_a\mu^\dag - \unl{[c,\lambda e_n]^{(a)}(\iota_\xi F_{A_0})_a \mu^\dag}{1011g.8} + \unl{\mathrm{L}_\xi^{\omega_0}(\lambda e_n)^{(a)}(\iota_\xi F_{A_0})_a \mu^\dag}{1011g.9}\\
		&\qquad + \unl{\frac{1}{2}\iota_{[\xi,\xi]}\iota_\xi F_{A_0}\mu^\dag}{1011g.10} -\unl{[c,\lambda e_n]^{(a)}d_{A_{0a}}\mu \mu^\dagger}{1011g.11} + \unl{\mathrm{L}^{\omega_0}_\xi(\lambda e_n)^{(a)}d_{A_{0a}}\mu \mu^\dagger}{1011g.12} + \unl{\frac{1}{2}\iota_{[\xi,\xi]}d_{A_0}\mu \mu^\dagger}{1011g.13}\\
		&\qquad  -\unl{\frac{1}{2}(\iota_{[\xi,\xi]}d_{\omega_0}(\lambda e_n))^{(a)}(A-A_0)_a\mu^\dag}{1011g.14}\\
		&\qquad + \left\{ \unl{([c,\lambda e_n]^{(b)}(d_{\omega_0}(\lambda e_n))_b)^{(a)}}{1011g.15} -  \unl{(\mathrm{L}_\xi^{\omega_0}(\lambda e_n)^{(b)}(d_{\omega_0}(\lambda e_n))_b)^{(a)}}{1011g.16} \right\}(A-A_0)_{(a)}\mu^\dag;
	\end{split}
\end{equation}

\begin{equation}
	\begin{split}\label{S11S11g}
		&\frac{1}{2}\{S_1^1,S_1^1\}_g = \\
		& = \Tr\int_{\Sigma} \unl{\frac{1}{2}\mathrm{L}_\xi^{A_0}(\iota_\xi\iota_\xi F_{A_0})\mu^\dag}{1111g.1} + \unl{[\mu,\mathrm{L}_\xi^{A_0}(\mu)]\mu^\dag}{1111g.2} - \unl{\mathrm{L}_\xi^{A_0}\mathrm{L}_\xi^{A_0}\mu \mu^\dag}{1111g.3}\\
		&\qquad \unl{\mathrm{L}_\xi^{\omega_0}(\mathrm{L}_\xi^{\omega_0}(\lambda e_n)^{(a)})(A-A_0)_a\mu^\dag}{1111g.4} - \unl{\mathrm{L}_\xi^{\omega_0}([c,\lambda e_n]^{(a)})(A-A_0)_a\mu^\dag}{1111g.5}\\
		&\qquad +\unl{\mathrm{L}_\xi^{\omega_0}(\lambda e_n)^{(a)}\mathrm{L}_\xi^{A_0}((A-A_0)_a)\mu^\dag}{1111g.6} - \unl{[c,\lambda e_n]^{(a)}\mathrm{L}_\xi^{A_0}((A-A_0)_a)\mu^\dag}{1111g.7}\\
		&\qquad +\unl{\frac{1}{2}[\iota_\xi\iota_\xi F_{A_0},\mu]\mu^\dag}{1111g.8} + \unl{\frac{1}{2}[[\mu,\mu],\mu]\mu^\dag}{1111g.9} - \unl{[\mu,\mathrm{L}_\xi^{A_0}(\mu)]\mu^\dag}{1111g.10}\\
		&\qquad \unl{\mathrm{L}^{\omega_0}_\xi(\lambda e_n)^{(a)}[(A-A_0)_a,\mu]\mu^\dag}{1111g.11} - \unl{[c,\lambda e_n]^{(a)}[(A-A_0)_a,\mu]\mu^\dag} {1111g.12} .
	\end{split}
\end{equation}
Now we check term by term that the sum is zero
\begin{itemize}
	\item \reft{S00S11f}{0011.1}, \reft{S10S11g}{1011g.6} and \reft{S10S11g}{1011g.4} give
	\begin{align*}
		&[c,[c,\lambda e_n]^{(b)}e_b]^{(a)} + [c,[c,\lambda e_n]^{(n)}e_n]^{(a)} - \frac{1}{2}[[c,c]\lambda e_n]^{(a)}\\
		& =[c,[c,\lambda e_n]]^{(a)} - \frac{1}{2}[[c,c]\lambda e_n]^{(a)} =0,
	\end{align*}
	because of graded Jacobi identity
	\item \reft{S00S11f}{0011.2}, \reft{S00S11f}{0011.4}, \reft{S10S11g}{1011g.1},\reft{S10S11g}{1011g.3}, \reft{S10S11g}{1011g.5} and \reft{S11S11g}{1111g.5} sum to zero, in fact
	\begin{align*}
		-\mathrm{L}^{\omega_0}_\xi([c,\lambda e_n])^{(a)}&= - \mathrm{L}^{\omega_0}_\xi ([c,\lambda e_n]^{(n)}e_n + [c,\lambda e_n]^{(b)}e_b)^{(a)}\\
		&= - \mathrm{L}^{\omega_0}_\xi([c,\lambda e_n]^{(n)}e_n)^{(a)} - \mathrm{L}^{\omega_0}_\xi([c,\lambda e_n])^{(a)}-[c,\lambda e_n]^{(b)}\mathrm{L}^{\omega_0}_\xi(e_b)^{(a)}\\
		&= -[\mathrm{L}^{\omega_0}_\xi(c),\lambda e_n]^{(a)} + [c,\mathrm{L}^{\omega_0}_\xi(\lambda e_n)^{(b)}e_b]^{(a)}+[c,\mathrm{L}^{\omega_0}_\xi(\lambda e_n)^{(n)}e_n]^{(a)}\\
		\Rightarrow \hspace{1mm} &- \mathrm{L}^{\omega_0}_\xi([c,\lambda e_n]^{(n)}e_n)^{(a)} - \mathrm{L}^{\omega_0}_\xi([c,\lambda e_n])^{(a)}-[c,\lambda e_n]^{(b)}\mathrm{L}^{\omega_0}_\xi(e_b)^{(a)} \\
		&- [\mathrm{L}^{\omega_0}_\xi(c),\lambda e_n]^{(a)} +  [c,\mathrm{L}^{\omega_0}_\xi(\lambda e_n)^{(b)}e_b]^{(a)} + [c,\mathrm{L}^{\omega_0}_\xi(\lambda e_n)^{(n)}e_n]^{(a)} =0 ;
	\end{align*}
	\item \reft{S00S11f}{0011.5}, \reft{S10S11g}{1011g.2},  \reft{S10S11g}{1011g.7}, \reft{S10S11g}{1011g.14},  \reft{S11S11g}{1111g.4} sum to zero, in fact
	\begin{align*}
		\mathrm{L}^{\omega_0}_\xi(\mathrm{L}^{\omega_0}_\xi(\lambda e_n))^{(a)} &= \mathrm{L}^{\omega_0}_\xi(\mathrm{L}^{\omega_0}_\xi(\lambda e_n)^{(b)}e_b + \mathrm{L}^{\omega_0}_\xi(\lambda e_n)^{(n)}e_n)^{(a)}\\
		&= \mathrm{L}^{\omega_0}_\xi(\mathrm{L}^{\omega_0}_\xi(\lambda e_n)^{(a)}) - \mathrm{L}^{\omega_0}_\xi(\lambda e_n)^{(b)}\mathrm{L}^{\omega_0}_\xi(e_b)^{(a)} - \mathrm{L}^{\omega_0}_\xi(\lambda e_n)^{(n)}\mathrm{L}^{\omega_0}_\xi(e_n)^{(a)}\\
		& = \frac{1}{2}\mathrm{L}^{\omega_0}_{[\xi,\xi]}(\lambda e_n)^{(a)} + \frac{1}{2}[\iota_\xi\iota_{\xi}F_{\omega_0},\lambda e_n]^{(a)}\\
		&=\frac{1}{2}(\iota_{[\xi,\xi]}d_{\omega_0}\lambda e_n)^{(a)} + \frac{1}{2}[\iota_\xi\iota_{\xi}F_{\omega_0},\lambda e_n]^{(a)}\\
		\intertext{}
		\Rightarrow\hspace{1mm}  &\mathrm{L}^{\omega_0}_\xi(\mathrm{L}^{\omega_0}_\xi(\lambda e_n)^{(a)}) - \mathrm{L}^{\omega_0}_\xi(\lambda e_n)^{(b)}\mathrm{L}^{\omega_0}_\xi(e_b)^{(a)} - \mathrm{L}^{\omega_0}_\xi(\lambda e_n)^{(n)}\mathrm{L}^{\omega_0}_\xi(e_n)^{(a)}\\
		& - \frac{1}{2}(\iota_{[\xi,\xi]}d_{\omega_0}\lambda e_n)^{(a)} - \frac{1}{2}[\iota_\xi\iota_{\xi}F_{\omega_0},\lambda e_n]^{(a)}=0;
	\end{align*}
	\item Now consider the following identity: $(\mathrm{L}_\xi^{(A_0)}(A-A_0))_a=\mathrm{L}_\xi^{(A_0)}(A-A_0)_a + \partial_a\xi^b(A-A_0)_b$. Then, considering the terms \reft{S00S11f}{0011.3}, \reft{S01S11f}{0111.1} and \reft{S11S11g}{1111g.7} we find
		\begin{align*}
			&-[c,\lambda e_n]^{(b)}\partial_b\xi^a + [c,\lambda e_n]^{(a)}(\mathrm{L}_\xi^{(A_0)}(A-A_0))_a - [c,\lambda e_n]^{(a)}\mathrm{L}_\xi^{A_0}(A-A_0)_a=\\
			=&-[c,\lambda e_n]^{(b)}\partial_b\xi^a + [c,\lambda e_n]^{(a)}\mathrm{L}_\xi^{A_0}(A-A_0)_a\\
			&- [c,\lambda e_n]^{(b)}\partial_b\xi^a - [c,\lambda e_n]^{(a)}\mathrm{L}_\xi^{A_0}(A-A_0)_a =0;
		\end{align*}
	\item the same can be done with the terms \reft{S00S11f}{0011.6}, \reft{S01S11f}{0111.2} and \reft{S11S11g}{1111g.6};
	\item the following pairs of terms simply cancel each other out
	\begin{itemize}
		\item[--] \reft{S01S11f}{0111.3} and \reft{S10S11g}{1011g.8};
		\item[--] \reft{S01S11f}{0111.4} and \reft{S10S11g}{1011g.9};
		\item[--] \reft{S11S11g}{1111g.2} and \reft{S11S11g}{1111g.10};
	\end{itemize}
	\item the terms \reft{S10S11g}{1011g.15} and \reft{S10S11g}{1011g.16} vanish because they are proportional to $\lambda^2$.
	They are separately zero because both $\mathrm{L}_\xi^{\omega_0}(\lambda e_n)^{(b)}$ and $[c,\lambda e_n]^{(b)}$ are proportional to $\lambda$, and 
	\begin{align*}
		(d_{\omega_0}(\lambda e_n)_{(b)})^{(a)} & =\partial_b \lambda e_n^{(a)} - \lambda (d_{\omega_0}e_n)^{(a)}_{(b)} \\
		& = - \lambda (d_{\omega_0}e_n)^{(a)}_{(b)} ;
	\end{align*}
	\item \reft{S11S11g}{1111g.9} vanishes because of the graded Jacobi identity;
	\item Considering \reft{S10S11g}{1011g.11} and \reft{S11S11g}{1111g.12} we find
		\begin{align*}
			&-[[c,\lambda e_n]^{(a)}(A-A_0)_a,\mu]\mu^\dag-[c,\lambda e_n]^{(a)}d_{A_{0a}}\mu \mu^\dag=\\
			&=-[c,\lambda e_n]^{(a)}d_{(A-A_0)_a}\mu \mu^\dag=[c,\lambda e_n]^{(a)}(d_A\mu)_a \mu^\dag
		\end{align*},
	which cancels out \reft{S01S11f}{0111.5};
	\item the same holds also for  \reft{S01S11f}{0111.5} \reft{S10S11g}{1011g.12} and \reft{S11S11g}{1111g.11};
	\item the terms \reft{S10S11g}{1011g.10} and \reft{S11S11g}{1111g.1} sum to a boundary term, in fact
		\begin{align*}
			&\frac{1}{2}\iota_{[\xi,\xi]}\iota_\xi F_{A_0}\mu^\dag + \frac{1}{2}\mathrm{L}_\xi^{A_0}(\iota_\xi\iota_\xi F_{A_0})\mu^\dag=\\
			&=\frac{1}{2}\iota_{[\xi,\xi]}\iota_\xi F_{A_0}\mu^\dag - \frac{1}{2}\iota_\xi\iota_\xi F_{A_0}+\mathrm{L}_\xi^{A_0}(\mu^\dag)\\
			&=\frac{1}{2}d_{A_0}(\iota_\xi\iota_\xi F_{A_0} \iota_\xi \mu^\dag)
		\end{align*};
	\item Finally, we are left with \reft{S10S11g}{1011g.13}, \reft{S11S11g}{1111g.3} and \reft{S11S11g}{1111g.8}, they sum up to zero, in fact
		\begin{align*}
			\mathrm{L}_\xi^{A_0}\mathrm{L}_\xi^{A_0}\mu&=\frac{1}{2}\mathrm{L}_{[\xi,\xi]}^{A_0}(\mu)+\frac{1}{2}[\iota_\xi\iota_\xi F_{A_0},\mu]\\
			&=\frac{1}{2}\iota_{[\xi,\xi]}d_{A_0}\mu + \frac{1}{2}[\iota_\xi\iota_\xi F_{A_0},\mu],
		\end{align*}
	
	which proves equation \eqref{EM BFV B.2}.
\end{itemize}

\newrefcontext[sorting=nty]
\printbibliography

\end{document}